\documentclass[a4paper,onecolumn,11pt,accepted=2024-06-24]{quantumarticle}
\pdfoutput=1
\usepackage[utf8]{inputenc}
\usepackage[english]{babel}
\usepackage[T1]{fontenc}
\usepackage{amsmath}
\usepackage{hyperref}
\usepackage[numbers,sort&compress]{natbib}

\usepackage{tikz}
\usepackage{lipsum}

\usepackage{amsfonts}
\usepackage{amsthm}
\usepackage{amssymb}
\usepackage{physics}
\usepackage{float}
\usepackage{bm}
\usepackage{verbatim}
\usepackage[normalem]{ulem}
\usepackage[ruled,vlined,linesnumbered]{algorithm2e}
\usepackage{setspace}
\usepackage{stackengine}

\usepackage{enumitem,kantlipsum}

\newtheorem{theorem}{Theorem}[section]
\newtheorem{lemma}[theorem]{Lemma}

\newtheorem{definition}[theorem]{Definition}

\newtheorem{corollary}[theorem]{Corollary}
\newtheorem*{theorem*}{Theorem}
\newtheorem{problem}{Problem}
\newtheorem*{problem*}{Problem}
\newtheorem*{lemma*}{Lemma}

\newtheorem{innercustomgeneric}{\customgenericname}
\providecommand{\customgenericname}{}
\newcommand{\newcustomtheorem}[2]{%
  \newenvironment{#1}[1]
  {%
   \renewcommand\customgenericname{#2}%
   \renewcommand\theinnercustomgeneric{##1}%
   \innercustomgeneric
  }
  {\endinnercustomgeneric}
}

\newcustomtheorem{customthm}{Theorem}
\newcustomtheorem{customlemma}{Lemma}

\def\cH{\mathcal{H}}

\def\cO{\mathcal{O}}

\def\cS{\mathcal{S}}

\def\cU{\mathcal{U}}

\def\hsga{H_\mathrm{SGA}}

\def\one{{\mathchoice {\rm 1\mskip-4mu l} {\rm 1\mskip-4mu l} {\rm 1\mskip-4.5mu l} {\rm 1\mskip-5mu l}}}

\newcommand{\sket}[1]{\, | #1 \rangle}

\begin{document}

\title{Hamiltonian simulation for low-energy states with optimal time dependence}

\author{Alexander Zlokapa}
\affiliation{Center for Theoretical Physics, MIT, 02139, USA}
\affiliation{Google Quantum AI, Venice, CA 90291, USA}

\author{Rolando D. Somma}
\affiliation{Google Quantum AI, Venice, CA 90291, USA}

\maketitle

\begin{abstract}
We consider the task of simulating time evolution under a Hamiltonian $H$ within its low-energy subspace. 
Assuming access to a block-encoding of  $H':=(H-E)/\lambda$, for some $\lambda>0$ and $E \in \mathbb R$, 
the goal is to implement an $\epsilon$-approximation to the evolution operator $e^{-itH}$ when the initial state 
is confined to the subspace corresponding to eigenvalues $[-1, -1+\Delta/\lambda]$ of $H'$, for $\Delta \leq \lambda$. We present a quantum algorithm that requires $\mathcal{O}(t\sqrt{\lambda\Gamma} + \sqrt{\lambda/\Gamma}\log(1/\epsilon))$ queries to the block-encoding for any choice of $\Gamma$ such that $\Delta \leq \Gamma \leq \lambda$.
When the parameters satisfy $\log(1/\epsilon) = o(t\lambda)$ and $\Delta/\lambda = o(1)$, this result improves over generic methods with query complexity $\Omega(t\lambda)$.
Our quantum algorithm leverages spectral gap amplification and the quantum singular value transform.

For a given $H$, 
the block-encoding of its $H'$ must be prepared efficiently to achieve an asymptotic speedup in simulating the low-energy subspace; we refer to these Hamiltonians as \emph{gap-amplifiable}.
We show necessary and sufficient conditions for gap amplifiability in terms of an operationally useful decomposition of $H$ into a sum of squares. 
Gap-amplifiable Hamiltonians include physically relevant examples such as frustration-free systems, and it encompasses all previously considered settings of low-energy simulation algorithms.
Any Hamiltonian can be expressed as a gap-amplifiable Hamiltonian after simple transformations, and our algorithm retains the asymptotic improvement over generic methods as long as the conditions on the parameters are met.

We also provide lower bounds for simulating dynamics of low-energy states.
In the worst case, we show that the low-energy condition cannot be used to improve the runtime of Hamiltonian simulation methods.
For gap-amplifiable Hamiltonians, we prove that our algorithm is tight in the query model with respect to $t$, $\Delta$, and $\lambda$. In the practically relevant regime where $\log (1/\epsilon) = o(t\Delta)$ and $\Delta/\lambda = o(1)$, we also prove a matching lower bound in gate complexity (up to logarithmic factors).
To establish the query lower bounds, we consider oracular problems including search and $\mathrm{PARITY}\circ\mathrm{OR}$, and also bounds on the degrees of trigonometric polynomials. To establish the lower bound on gate complexity, we use a circuit-to-Hamiltonian reduction, where a ``clock Hamiltonian'' acting on a low-energy state can simulate any quantum circuit.
\end{abstract}

\section{Introduction}
\label{sec:intro}

A leading application for quantum computers is the simulation of the dynamics of quantum systems~\cite{Fey82}: given a Hamiltonian $H$, an initial quantum state $\ket \psi$, and an evolution time $t$, the task is to produce the state at a later time as predicted by Schr\"odinger's equation, i.e., $e^{-itH}\ket \psi$. Simulating time evolution on a quantum computer, also known as the Hamiltonian simulation problem, is relevant to numerous problems in physics~\cite{Llo96,SOGKL02}, chemistry~\cite{Llo96,aspuru2005simulated}, and many other quantum algorithms (cf., ~\cite{CKS17,holmes2022quantum,an2023linear}). Extensive effort has thus been made towards obtaining the most efficient quantum algorithms for Hamiltonian simulation in a variety of settings~\cite{BAC07,WBH+10,berry2014exponential,BCC+15,LC17,low2019qubitization,Cam19,haah2018quantum}. While such quantum algorithms are efficient in that they require polynomial resources in the evolution time and system size for interesting classes of Hamiltonians, many applications introduce considerations that necessitate careful comparison between algorithms. For this reason, recent studies examine improvements for Hamiltonian simulation when given additional knowledge about structures of the Hamiltonian or properties of the initial states. Examples include the locality of interactions~\cite{haah2018quantum,CST+19}, Lie-algebraic properties~\cite{Som16}, symmetries~\cite{tran2021faster}, or particle statistics~\cite{BBK++15,BBK+15}.

Along similar lines, 
in this article we focus on the Hamiltonian simulation problem for an important class of initial quantum states: \emph{low-energy} states. Given a Hamiltonian $H$, low-energy states are those
supported exclusively on a subspace of energies that are strictly below a given energy (eigenvalue) of $H$.
Hamiltonian simulation of low-energy states is relevant in physically motivated contexts, such as when studying zero-temperature quantum phase transitions in condensed matter systems~\cite{sachdev1999quantum}, simulating quantum field theories~\cite{JLP12}, or computing ground-state energies of molecules with quantum phase estimation~\cite{aspuru2005simulated,lanyon2010towards}.
It is also essential for adiabatic quantum computing~\cite{FGGS00,albash2018adiabatic}, 
a generic technique to solve optimization and other problems in quantum computing by preparing the ground state of a complex (interacting) Hamiltonian starting from the ground state 
of a simpler (non-interacting) one. By ground state we mean the eigenstate (eigenvector) of lowest energy (eigenvalue). 
Consequently, any enhancement of Hamiltonian simulation techniques under the low-energy consideration of the initial state is anticipated to uncover substantial applications.

Owing to its significance, some recent research has already investigated Hamiltonian simulation of low-energy states. In particular, product formulas may exploit certain structures of the problem (e.g., locality of interactions), especially if the precision requirements are not too stringent. 
Reference~\cite{csahinouglu2021hamiltonian}  
discusses improvements on Hamiltonian simulation methods based on Trotter-Suzuki product formulas when $H$ is a local Hamiltonian presented as a sum of positive semidefinite (PSD) terms. Comparable results were provided in Ref.~\cite{gong2023theory} for additional product-formula based quantum algorithms, such as qDRIFT~\cite{Cam19} or other randomized product formulas. While these quantum algorithms show an improvement with respect to the general case that allows an arbitrary initial state, they  are generally not optimal in terms of the evolution time, Hamiltonian norm, or the allowed error, which is a common feature of product formulas. Nevertheless,
they might still offer some advantages with respect 
to qubit overheads in comparison to other methods.

Unlike product formulas, methods based on a block-encoding of the Hamiltonian can achieve optimal worst-case bounds (cf.~\cite{berry2014exponential,low2017hamiltonian}). In principle,  these methods do not offer much flexibility for taking advantage of structures in the Hamiltonian or initial state, and additional ideas are needed. To this end, Ref.~\cite{gu2021fastforwarding} provides 
an approach for simulating low-energy states of PSD Hamiltonians that uses quantum phase estimation in combination with the spectral gap amplification technique of Ref.~\cite{somma2013spectral}. That algorithm, which was intended for a new demonstration of ``fast-forwarding'' quantum evolution, is not optimal either because of the considerable overhead of quantum phase estimation with the allowed error.

The scaling of our upper bound more closely resembles that of Ref.~\cite{low2017hamiltonian}. Our results improve upon their upper bounds by logarithmic factors (and we show novel lower bounds not in Ref.~\cite{low2017hamiltonian}). We summarize the setting of Ref.~\cite{low2017hamiltonian} due to its relevance. For a given Hamiltonian $H$, we let $H'=(H-E)/\lambda$ be a shifted version of $H$.  Here, $E \in \mathbb R$ and $\lambda >0$ are such that the eigenvalues of $H'$ lie in $[-1,1]$. Both $H$ and $H'$ share the same eigenstates, and
simulating $H$ for time $t$ can be done by simulating $H'$ for ``time'' $t \lambda$.
Assuming access to a unitary $T$ that block-encodes $H'$ (i.e., a unitary that contains $H'$ in one of its blocks),
Ref.~\cite{low2017hamiltonian} provides a quantum algorithm for Hamiltonian simulation when the initial state
is confined to the subspace corresponding to eigenvalues of $H'$ in $[-1, -1+\Delta/\lambda]$, where $\Delta \in [0,\lambda]$ specifies the low-energy subspace. The interesting case occurs when $\Delta/\lambda \ll 1$ --- formally, when $\Delta = o(\lambda)$ --- where the complexity of that algorithm can be an asymptotic improvement over that of more generic Hamiltonian simulation methods, like quantum signal processing (QSP)~\cite{LC17}  or related approaches~\cite{berry2014exponential,BCC+15}.
QSP approximates $e^{-itH}$ using a different block-encoding of $H$, i.e., a unitary $W$ that contains $H/\Lambda$ in one of its blocks, for some $\Lambda >0$. In general, no guarantees are made with regards to the range of eigenvalues associated with the low-energy subspace of $H/\Lambda$, other than these being in $[-1,1]$.

It is unclear how to take advantage of the result in Ref.~\cite{low2017hamiltonian} in general, or whether that algorithm can be further improved.
In particular, it is not obvious how to construct a ``good'' block-encoding $T$ from a given $H$ such that: i) the low energies of $H$ map to eigenvalues of $H'$ in $[-1, -1+\Delta/\lambda]$ for $\Delta=o(\lambda)$, ii)  the process to block-encode $H'$ and implement $T$ is efficient, and iii) the value of $\lambda$ is not prohibitively large so that the overhead from simulating the shifted Hamiltonian (being linear in $\lambda$) does not ruin any other possible improvement coming from the low-energy assumption. For example, if access to the Hamiltonian is given by a block-encoding $W$ of $H/\Lambda$ rather than $T$, which is the block-encoding of the shifted version of $H$, we want to avoid situations in which constructing or implementing $T$ requires too many uses of $W$. Otherwise methods like QSP, which readily use $W$,  might already be more efficient. Similarly, if after constructing $T$ from $W$ we satisfy $\sqrt{\lambda \Delta} \gg \Lambda$, we will see
that QSP using $W$ can again outperform the approach of Ref.~\cite{low2017hamiltonian} that uses $T$. 
Hence, if any of these three considerations 
is not satisfied, it is unclear when --- and if --- the results of Ref.~\cite{low2017hamiltonian}  can be used to successfully exploit the low-energy condition of the initial state.

The goal of this article is to answer some of these open questions. To this end, we first construct a quantum algorithm for Hamiltonian simulation for low-energy states that asymptotically outperforms the time complexity of all known previous methods. Assuming access
to the block-encoding $T$ like in Ref.~\cite{low2017hamiltonian}, our algorithm computes $e^{-i(t\lambda)H'}\ket \psi$ --- equivalently, $e^{-i t H}\ket \psi$ ---  within specified accuracy.
It uses two known primitives, spectral gap amplification~\cite{somma2013spectral} 
and the quantum singular value transform (QSVT)~\cite{gilyen2019quantum}, to implement a certain polynomial of $T$ that approximates the correct evolution when acting on $\ket \psi$.

To demonstrate the utility of our algorithm, we study when it is possible to efficiently construct the block-encoding $T$ such that our quantum algorithm is useful.
Usually, the block-encoding $T$ will have to be constructed from certain operations that provide access to $H$, like its block-encoding $W$ or a matrix oracle. In these standard access models, we find that there exists a class of Hamiltonians that we deem {\em gap-amplifiable}, for which Hamiltonian simulation can be performed more efficiently by using our method compared to other techniques like QSP, as long as the initial state is of low energy. 
Indeed, we will see that there is an equivalence between 
having efficient access to $T$ with the desired properties and 
the Hamiltonian $H$ being gap-amplifiable. 
Roughly, in an interesting regime for the parameters (where the error is not too small), the improvement replaces a cost linear in $t\lambda$ as required by QSP to linear in $t\sqrt{\lambda \Delta}$, where $\Delta \le \lambda$ is an upper bound on the largest energy. Gap-amplifiable Hamiltonians are PSD and include frustration-free Hamiltonians as a special case, which appear ubiquitously in physics, quantum computing, and beyond~\cite{schuch2008computational,BOO10,BT09,chen2012ground}. These Hamiltonians are also known to permit asymptotic improvements in other settings, such as ground state preparation, Gibbs sampling, and 
quantum linear systems~\cite{somma2013spectral,chowdhury2016quantum,orsucci2021solving,thibodeau2023nearly}. Moreover, we show that this class of Hamiltonians encompasses all previously considered settings of low-energy simulation studied in the literature~\cite{csahinouglu2021hamiltonian,gu2021fastforwarding,gong2023theory,low2017hamiltonian}.

In the asymptotic query complexity of our algorithm, 
which refers to the number of times the block-encoding $T$ and its inverse are used,
the parameters $t\Delta$, $t\lambda$, and $1/\epsilon$ can be arbitrarily large, where 
 $\epsilon >0$ is the error. 
Our results show there are only three interesting asymptotic regimes to consider: i) a regime where the evolution time dominates the cost, i.e., where $\log (1/\epsilon) \ll t\Delta \ll t\lambda$; ii) an intermediate regime where the parameters satisfy $t\Delta \ll \log (1/\epsilon) \ll t\lambda$;  iii) a regime where error dominates the cost, i.e., where $t\Delta \ll t\lambda \ll \log (1/\epsilon)$. 
The respective complexities we obtain for the first two regimes are strict improvements over QSP. 
The time-dominated regime is most natural for applications, since a large $t\Delta \gg 1$ implies that the dynamics in the low-energy sector is not trivial, 
while still $\Delta/\lambda$ could be vanishingly small. 
The second regime is distinct in that the resulting query complexity is sublinear in time (we will see that it scales as $\sqrt{t \lambda}$), and hence likely an artificial regime for most applications.
The error-dominated regime is subsumed by previous work on general quantum simulation; that is, the complexity matches that of QSP in this regime, which is logarithmic in $1/\epsilon$. 

Finally, we study whether our quantum algorithm can be further improved.
We prove matching lower bounds for all three asymptotic regimes that show that our quantum algorithm is optimal with respect to $t$, $\Delta$, and $\lambda$, in terms of query complexity. Note that the complexity of our algorithm always produces logarithmic or sublogarithmic scaling in $1/\epsilon$, and previous results readily imply optimality in the error-dominated regime~\cite{berry2014exponential}. Our lower bounds are mostly stated in terms of the {\em query} complexity when given access to the Hamiltonian in standard oracle models. Nevertheless, in the time-dominated regime, which is arguably the most relevant regime to physics and chemistry problems, we go further and also prove a matching lower bound in terms of {\em gate} complexity.

While an asymptotic benefit is achieved for gap-amplifiable Hamiltonians (equivalently, when given access to $T$), our algorithm might not offer any improvement with respect to QSP or other known methods in the worst case. That is, if $T$ is not accessible a priori and has to be constructed from some other form of access to $H$, mapping the low energies of $H$ to eigenvalues near $-1$ of $H'$ can result in a block-encoding $T$ that is inefficient in general. This observation gives rise to a question of whether our method can be further improved for simulating low-energy states of general Hamiltonians beyond those gap-amplifiable ones. Unfortunately, the answer is negative: we establish matching lower bounds in two common access oracle models --- namely, the LCU and sparse matrix model --- that show that QSP is already optimal in terms of query complexity, even if the initial state is of low energy. Hence, our results show that the low-energy condition is only useful when additional conditions hold, such as having access to $T$ or being able to construct it efficiently. This is the case for gap-amplifiable Hamiltonians, including frustration-free ones. 
We emphasize, however, that: i) our lower bounds are for worst-case instances and in the oracle model, and ii) \emph{any} Hamiltonian $H$ (e.g., a sum of Pauli strings) can be transformed into a gap-amplifiable Hamiltonian after simple manipulations (e.g., adding a constant offset), allowing the application of our simulation algorithm. 
Whether the asymptotic improvement is retained or not for any given $H$ depends on the properties of the resulting parameters 
after the transformation, especially $\Delta$ and $\lambda$. For example, our algorithm still gives an asymptotic speedup for Hamiltonians that are perturbatively far from frustration-free Hamiltonians, even if they are not originally expressed as gap-amplifiable ones.



Last, we note that Ref.~\cite{gong2023theory} recently presented a lower bound for Hamiltonian simulation also within the context of low-energy states. Their proof relies on a state with partial support on an eigenstate of energy $\Theta(\norm{H})$; this high-energy portion of the state solves PARITY, which produces the lower bound. In our work, we show a strictly stronger lower bound in terms of query complexity to a block-encoding of the Hamiltonian; our result applies to a state fully contained in a low-energy subspace of vanishingly small energies compared to $\norm{H}$.

We give more details of our contributions next.

\subsection{Summary of results}
\label{sec:summary}

Roughly, our results are four-fold: 1) an improved quantum algorithm for simulating the time evolution of low-energy states, 2) defining a class of Hamiltonians (``gap-amplifiable Hamiltonians'') for which our algorithm achieves a speedup in the simulation of low-energy states, 3) a lower bound that shows that no speedup is possible for evolving low-energy states of general Hamiltonians, and 4) a lower bound that shows that our algorithm is optimal for gap-amplifiable Hamiltonians.

To present these results more formally, we let $H$ be the Hamiltonian of an $N$-dimensional quantum system with Hilbert space $\cH \equiv \mathbb C^N$, whose evolution we seek to approximate. (We can think of this as being an $n=\log_2(N)$ qubit system.) The ground state energy of $H$ (i.e., its lowest eigenvalue) is $E_0$ and largest energy (i.e., its largest eigenvalue) is $E_{\max} \ge E_0$. Without loss of generality, we may assume $E_0 \ge 0$, which can always be achieved by shifting the Hamiltonian. 
We want to simulate the time evolution of an initial state $\ket \psi \in \cH$, and we say that $\ket \psi$ is supported in the low-energy subspace $\cS_{\Delta}$  if and only if it is supported on the subspace corresponding to eigenvalues of $H$ in $[0,\Delta]$.  The interesting case for this article is when $\Delta \ll E_{\max}$ and $\Delta > 0$. 

Our quantum algorithm uses ancillary systems and is essentially a unitary operation in an enlarged space. In general, we write $\one_M$ to denote the $M \times M$ identity matrix.  We begin by stating our results in terms of a block-encoding of $H' = (H-E)/\lambda$, for some $E \in \mathbb R$ and $\lambda>0$, such that the eigenvalues of $H'$ lie within $[-1, 1]$\footnote{The  factor $\lambda$ also appears in other Hamiltonian simulation methods and is related to an L1 norm of $H$ when presented as a linear combination of unitaries (LCU). Other authors refer to this factor as $\alpha$.}. (To be precise, $H-E$ refers to the matrix $H -E.\one_N$ in our notation.)
Moreover, we assume that the low-energy subspace $\cS_{\Delta}$ of $H$ corresponds to eigenvalues of $H'$ within $[-1, -1+\Delta/\lambda]$. This block-encoding is a unitary operator of the form
\begin{align}
\label{eq:blenc}
    T &= \begin{pmatrix}H' & \cdot \\ \cdot & \cdot\end{pmatrix}
\end{align}
that acts on an enlarged Hilbert space $\mathbb C^M$, $M \ge N$.
We also let $\Pi$ be the orthogonal projector onto $\cH$ such that $\Pi T \Pi = H'$; see Sec.~\ref{sec:sgablenc} for details. To use the techniques of the quantum singular value transform (QSVT)~\cite{gilyen2019quantum}, we will be interested in quantum circuits that act upon the enlarged Hilbert space associated with a controlled-$T$ operation and its inverse. QSVT can produce a new unitary $U$ that block-encodes a new matrix as $\Pi U \Pi$. For our problem, the new matrix is such that it approximates the time evolution of $H$ when acting on the initial state, i.e., $\Pi U \Pi \ket \psi \approx e^{-itH} \ket \psi$. Our first result is a quantum algorithm that implements this unitary.

\begin{theorem}[Quantum algorithm for time evolution of low-energy states]
\label{thm:mainalgorithm}
Let $\epsilon > 0$ be the error and $t>0$ be the evolution time. Let $T$ be the unitary of Eq.~\eqref{eq:blenc} that block-encodes the Hamiltonian $H'=(H-E)/\lambda$ for some $E\in \mathbb R$ and $\lambda > 0$.  Let $\ket{\psi} \in \cS_{\Delta}$ be any state supported exclusively on the subspace corresponding to eigenvalues of $H'$ in $[-1, -1+\Delta/\lambda]$, and assume $t \Delta \ge \epsilon$. Then, for any $\Gamma$ such that $\Delta \le \Gamma \le \lambda$, there exists a quantum algorithm that implements a unitary $U$ and satisfies $|\!\bra{\psi}\Pi U^\dagger \Pi e^{-itH}\ket{\psi}\!| \geq 1 - \epsilon$, using
\begin{align}
\label{eq:mainresult}
    \cO\left(t\sqrt{\lambda\Gamma} + \sqrt{\frac{\lambda}{\Gamma}} \log \frac{1}{\epsilon}\right)
\end{align}
queries to controlled-$T$ and its inverse.
\end{theorem}

The reason why we use the requirement on the overlap  $|\!\bra{\psi}\Pi U^\dagger \Pi e^{-itH}\ket{\psi}\!| \geq 1 - \epsilon$ to measure the performance of the quantum algorithm is because
this expression is amenable to block-encodings. 
This requirement is equivalent 
to the quantum circuit producing a state $U \ket \psi \ket 0^{\otimes q}=(1-\delta) e^{i \alpha} (e^{-itH}\ket \psi) \ket 0^{\otimes q} + \sket{\psi^\perp}$, where $\epsilon \ge \delta \ge 0$, and $\ket 0^{\otimes q}$ is the initial state of a working register of qubits. The phase $\alpha$ is irrelevant and $\sket{\psi^\perp}$ is an unnormalized state over the full system of $n+q$ qubits such that $\|\sket{\psi^\perp}\|\le \sqrt {2\epsilon}$.

In the hypothesis, we require $t \Delta \ge \epsilon$; otherwise we could simulate $e^{-itH}$ by applying $\one_N$ (or a global phase) at the cost of no queries, which is uninteresting. Depending on the joint scaling of $t\lambda$, $ t\Delta$, and $1/\epsilon$, we may choose different $\Gamma$ to minimize the query complexity. Because $\ket{\psi} \in \cS_\Delta$ also satisfies $\ket{\psi} \in \cS_\Gamma$ for any $\Gamma \geq \Delta$, such optimization is feasible. We describe the exhaustive set of possibilities for asymptotic regimes of the parameters; these regimes will be central throughout our work.
\begin{enumerate}
    \item \emph{Time-dominated regime}: $\log (1/\epsilon) = o(t\Delta)$. We choose $\Gamma = \Delta$, producing an upper bound of $\cO(t\sqrt{\lambda\Delta})$. 
    \item \emph{Intermediate regime}: $\log (1/\epsilon) = o(t\lambda)$ and $t\Delta = o(\log (1/\epsilon))$. We choose $\Gamma = \log(1/\epsilon)/t$, producing an upper bound of $\cO(\sqrt{t\lambda \log (1/\epsilon)})$.
    \item \emph{Error-dominated regime}: $t\lambda = o(\log (1/\epsilon))$. We choose $\Gamma = \lambda$, producing an upper bound of $\cO(\log 1/\epsilon)$.
\end{enumerate}

The time-dominated regime is expected to be the most relevant to applications, where $t\Delta$ must be large enough for interesting dynamics. It gives savings over standard QSP methods, which would instead time-evolve any state $\ket \psi \in \cH$ (not necessarily of low energy) with access to a block-encoding of $H/\lambda$ using
\begin{align}
\label{eq:QSPcomplexity}
    \cO\left(t\lambda + \log \frac{1}{\epsilon}\right)
\end{align}
queries (ignoring log log factors)~\cite{LC17}. Since $\Delta/\lambda$ can be asymptotically small, the time-dominated regime shows clear savings over QSP. 
The intermediate regime represents a fine-tuned region of parameter space, where a particular joint scaling between $t \Delta, 1/\epsilon$, and $\Delta/\lambda$ makes sublinear time dependence possible. For this regime we also obtain an improvement over QSP. 
In contrast, the error-dominated regime has equivalent cost to QSP. 

Our result improves upon Ref.~\cite{low2017hamiltonian} by logarithmic factors in $t \lambda/\epsilon$. 
It also improves upon the complexity of product formulas in that the asymptotic scaling is linear in $t$ and logarithmic or sublogarithmic in $1/\epsilon$; however, these approaches are not directly comparable. Later, we will also show novel lower bounds that address the time-dominated and intermediate regimes above to demonstrate optimal scaling with respect to $t$, $\Delta$, and $\lambda$. The error-dominated regime is already addressed by previous lower bounds~\cite{BCC+15}.

Theorem~\ref{thm:mainalgorithm} and subsequent results assume the initial state to be supported exclusively on the low-energy subspace. However, a similar result still applies even if the support of $\ket \psi$ outside $\cS_{\Delta}$ is small, i.e., $\|(\one_N - \cS_{\Delta})\ket \psi \|=\cO(\epsilon)$. 
This is because we are approximating a unitary transformation with unitary operations, and standard bounds for the Euclidean and operator norms can be used to show this result.

We analyze our algorithm in the standard query models considered by QSP and similar algorithms, counting queries to controlled-$T$ and its inverse. Much like QSP and similar algorithms, our quantum algorithm uses a number of other two-qubit gates essentially given by Eq.~\eqref{eq:mainresult} times the gate complexity needed to perform certain transformation on $H$ or $T$, as discussed in Sec.~\ref{sec:alg}, including the walk operators to perform polynomial transformations. (Our reference to two-qubit gates also includes one-qubit gates.)
This additional factor to the gate complexity will depend on the specific application, but it is expected to be polynomial --- and in many cases only constant or linear --- in the number of qubits or system size. Detailed analyses of the gate complexity for specific applications are outside the scope of this article.

Our quantum algorithm in Thm.~\ref{thm:mainalgorithm} is actually constructed from one that simulates a Hamiltonian of the form $H=\lambda A^\dagger A$, for $\lambda >0$ and $A$ a matrix of dimension $M \times N$
that satisfies $\|A\| \le 1$. This is an example of what we call a gap-amplifiable Hamiltonian in Sec.~\ref{sec:FrustrationFree}. The time evolution of $H$ can also be simulated if we have access to a unitary $V$ that is a block-encoding of $A$ and its inverse; that is,
\begin{align}
\label{eq:blecofA}
    V &= \begin{pmatrix} A & \cdot \\ \cdot & \cdot\end{pmatrix} \;.
\end{align}
Formally, this block-encoding satisfies $A=\Pi' V \Pi$, where $\Pi'$ is a projector onto $\cH' \equiv \mathbb C^M$, since $A$ can be a rectangular matrix. Gap-amplifiable Hamiltonians
are PSD ($H \succeq 0$) and readily satisfy $E_0 \ge 0$.

\begin{lemma}[Quantum algorithm for time evolution of low-energy states of gap-amplifiable Hamiltonians]
\label{lem:GAsimulation}
Let $\epsilon > 0$ be the error and $t>0$ be the evolution time. Consider a gap-amplifiable Hamiltonian $H=\lambda A^\dagger A$ acting on $\cH$, where $\lambda>0$ and the dimension of $A$ is $M \times N$.
Assume access to a unitary $V$ acting on an enlarged space that block-encodes $A$ as in Eq.~\eqref{eq:blecofA}.
Let $\ket \psi \in \cS_{\Delta}$ be any state supported exclusively on the subspace corresponding to energies (eigenvalues) of $H$ at most $\Delta$,  and assume $t\Delta \ge \epsilon$. Then, for any $\Gamma$ such that $\Delta \leq \Gamma \leq  \lambda$, there exists a quantum algorithm that implements a unitary $U$ and satisfies 
$|\!\bra{\psi}\Pi U^\dagger \Pi e^{-itH}\ket{\psi}\!| \geq 1 - \epsilon$, using
\begin{align}
\label{eq:gapampsimulation}
    \cO\left(t\sqrt{\lambda\Gamma} + \sqrt{\frac{\lambda}{\Gamma}} \log \frac{1}{\epsilon}\right)
\end{align}
queries to controlled-$V$ and its inverse.
\end{lemma}

The proof of this lemma is in Sec.~\ref{sec:alg} and has two steps. First, in Sec.~\ref{sec:sga} we use $A$ and $A^\dagger$ to construct a new Hamiltonian $\hsga$ that acts as the ``square root'' of $H$, borrowing ideas 
from spectral gap amplification~\cite{somma2013spectral}. This implies that $e^{-itH}$ can be obtained from $e^{-it(\hsga)^2}$. Second, in Sec.~\ref{sec:qsvt} we show 
how to simulate $e^{-it(\hsga)^2}$ using QSVT, which requires the block-encoding of $\hsga$ or $A$~\cite{gilyen2019quantum}. Hence, Lemma~\ref{lem:GAsimulation} already gives an improved Hamiltonian simulation method for low-energy states of the interesting class of  Hamiltonians $\lambda A^\dagger A$ when $\Delta = o(\lambda)$.

The main reason why Lemma~\ref{lem:GAsimulation} implies Thm.~\ref{thm:mainalgorithm} is because access to the block-encoding $T$ of Eq.~\eqref{eq:blenc} allows
us to equivalently express that Hamiltonian $H$ as a gap-amplifiable one, $\lambda A^\dagger A$, up to an irrelevant additive constant. Furthermore, access to $A$ can be efficiently obtained through $T$,  satisfying properties (i)-(iii) discussed in Sec.~\ref{sec:intro}.

\begin{lemma}[Efficient block-encoding of $A$]
\label{lem:hsgablec}
    Let $T$ be the unitary of Eq.~\eqref{eq:blenc} that block-encodes Hamiltonian $H'=(H-E)/\lambda$ acting on $\cH$, for $E \in \mathbb R$ and $\lambda >0$. Let $\ket \psi \in \cS_{\Delta}$ be any state supported exclusively on the subspace associated with eigenvalues of $H'$ in $[-1,-1+\Delta/\lambda]$ for some $\Delta > 0$.
    Then, there exists $F \in \mathbb R$ and an $M\times N$ matrix $A$ such that $H-F= 2 \lambda A^\dagger A$. A block-encoding $V$ of 
    $A$ as in Eq.~\eqref{eq:blecofA} can be implemented with one query to controlled-$T^\dagger$ and a constant number of two-qubit gates. In addition, $\ket \psi$ is supported in the subspace associated with energies at most $\Delta$ of $2\lambda A^\dagger A$.
\end{lemma}

The proof is given in Sec.~\ref{sec:sgablenc}. It is also possible to show the opposite direction of Lemma~\ref{lem:hsgablec}: any gap-amplifiable Hamiltonian $\lambda A^\dagger A$ can be efficiently converted into a Hamiltonian $H'$ with the property that, if $\ket \psi \in \cS_\Delta$, then $\ket \psi$ is supported in the subspace of eigenvalues of $H'$ in $[-1, -1+\Delta/\lambda]$. In addition, a block-encoding $T$ of $H'$ can be constructed from a constant number of queries to the block-encoding of $A$ and $A^\dagger$. This is because we can write $H'=H/\lambda-\one_N=\Pi V^\dagger (\Pi'-\one_M)V\Pi$ in this case, and a block-encoding of $H'$ follows from standard techniques (see Lemma~\ref{lem:lbhsgablec} for a proof). Hence, our quantum algorithm can give improved results {\em if and only if} the Hamiltonian can be presented as $H=\lambda A^\dagger A$ (up to additive constants) and certain access to $A$ is given or can be efficiently implemented. The result does not extend to all Hamiltonians
because creating access to $A$ can be computationally 
expensive in general.

More formally, we can write the result of Lemma~\ref{lem:GAsimulation} in terms of ``general gap-amplifiable Hamiltonians'' defined in 
Sec.~\ref{sec:FrustrationFree}. These Hamiltonians include the well-known frustration-free Hamiltonians as an example.
We will show that they can be equivalently expressed as $\lambda A^\dagger A$, and hence the same quantum algorithm for Lemma~\ref{lem:GAsimulation} gives an improved method for the time evolution of low-energy states of these Hamiltonians.
This result is captured by Thm.~\ref{thm:gapalgorithm}.
We note
a difference between Lemma~\ref{lem:GAsimulation} and Ref.~\cite{somma2013spectral}, where the latter work  was concerned on amplifying the gap of frustration-free Hamiltonians for faster quantum adiabatic state transformations rather than simulating time-evolution. 

Having devised a class of Hamiltonian simulation problems that can benefit from our quantum algorithm, it is natural to ask whether there is an even larger class of Hamiltonians for which simulating the time evolution of their low-energy states can be done more efficiently.  To this end, we show a no fast-forwarding result:  no improvement over the query complexity of QSP is possible in the worst case, even if the initial state is a low-energy state.
We prove this lower bound in two common access models:
the LCU model and the sparse matrix model~\cite{berry2014exponential,low2019qubitization}. 
These models provide a stronger form of access to the Hamiltonian than through its block-encoding, and hence the lower bound automatically applies to the case  where only access to the block-encoding is provided.
An informal statement of this result in terms of the block-encoding of $H/\lambda$ follows.

\begin{theorem*}[No fast-forwarding in the low-energy subspace of generic Hamiltonians, informal]
\label{thm:nff-inf}
Let $W$ be a unitary block-encoding of a Hamiltonian $H/\lambda$ acting on $\cH$, such that $\Pi W \Pi = H/\lambda$ for some $\lambda >0$, and assume $H \succeq 0$. Let $\ket \psi \in \cS_{\Delta}$ be any state supported exclusively on the subspace associated with energies  (eigenvalues) of $H$ at most $\Delta$, for $\lambda > \Delta >0$. Then, any quantum algorithm that implements a unitary $U$ such that $|\!\bra{\psi}\Pi U^\dagger \Pi e^{-itH}\ket{\psi}\!| \geq 2/3$ must use
\begin{align}
\label{eq:nffboundgeneral}
    \Omega\left(t\lambda\right)
\end{align}
queries to controlled-$W$ and its inverse, even when $\Delta/\lambda =o( 1)$.
\end{theorem*}

By assuming $U$ to be unitary, the result also encompasses
quantum algorithms where measurements might be interleaved with queries and gates. These can be simulated coherently by enlarging the space using standard methods in quantum computing.

A formal statement of this result that applies to the LCU and sparse matrix models is given in Sec.~\ref{sec:nff}. 
Even when $\Delta/\lambda =o(1)$, the query complexity in Eq.~\eqref{eq:nffboundgeneral} can be asymptotically worse than the complexity in Thm.~\ref{thm:mainalgorithm}.
The proof of this result relies on a reduction from multiple instances of Grover quantum search or, more specifically, from a lower bound on computing $\mathrm{PARITY}\circ\mathrm{OR}$ of certain bit strings.  Among other consequences, this result implies the inefficiency of constructing a ``good'' block-encoding $T$ from a block-encoding of $H/\lambda$ in general, as discussed in Sec.~\ref{sec:intro}. 

For gap-amplifiable Hamiltonians of the form $\lambda A^\dagger A$ and an initial state with energy at most $\Delta$, it would have been intuitive to suggest that a modification of QSP would have provided an overall query complexity  $\cO(t\Delta + \log (1/\epsilon))$ for low-energy simulation. This is because the ``relevant'' energy scale in the problem is $\Delta$ rather than $\lambda$ or $\|H\|$. However,
we show that our quantum algorithm is optimal in its dependence on $t$, $\Delta$, and $\lambda$. Informally, our result is as follows.

\begin{theorem*}[Query complexity lower bound for time evolution of low-energy states of gap-amplifiable Hamiltonians, informal]
\label{thm:alg-lb-inf}
Let $\epsilon>0$ be the error and $t>0$ be the evolution time.
Let $H$ be any Hamiltonian acting on $\cH$ and $T$ be a block encoding of $H'=(H-E)/\lambda$, for some $E \in \mathbb R$ and $\lambda >0$, i.e., $\Pi T \Pi =H'$. 
Let $\ket{\psi}\in \cS_{\Delta}$ be any state supported exclusively on the subspace of eigenvalues of $H'$  in $[-1, -1+\Delta/\lambda]$ for some $\Delta >0$, $\Delta \le  \lambda$. Then, any quantum algorithm that implements a unitary $U$ such that $|\!\bra{\psi}\Pi U^\dagger \Pi e^{-itH}\ket{\psi}\!| \geq 1-\epsilon$, must use at least $Q$ queries to controlled-$T$ and its inverse given by
\begin{enumerate}
    \item $Q = \Omega(t\sqrt{\lambda\Delta})$ if $\log (1/\epsilon) = o(t\Delta)$,
    \item $Q = \Omega(\sqrt{t\lambda})$ if $\log (1/\epsilon )= o(t\lambda)$ and $t\Delta = o(\log( 1/\epsilon))$,
    \item $Q = \Omega(\log (1/\epsilon))$ if $t\lambda = o(\log (1/\epsilon))$,
\end{enumerate}
ignoring $\log \log$ factors in $1/\epsilon$ in the third case.
\end{theorem*}

Formal statements of these results are given in Sec.~\ref{sec:nffsga}.
The above result concerns query complexity given access to $T$ in an oracle model. Since the third regime takes $\Gamma=\lambda$, both upper and lower bounds are already known~\cite{berry2014exponential,LC17}, and no improvement from the low-energy condition follows. In the second regime of intermediate scaling, the lower bound does not exactly match the upper bound $\cO(\sqrt{t \lambda \log(1/\epsilon)})$
up to a mild sublogarithmic correction in $1/\epsilon$, and
we leave open this discrepancy for error dependence. However, we do not expect a significant improvement in our upper bound in this case, which smoothly transitions to those of the other regimes as we change the asymptotic parameters. 
The first regime is likely most relevant in terms of applications; examples include quantum phase estimation
to determine ground state energies~\cite{aspuru2005simulated,lanyon2010towards} or adiabatic state transformations~\cite{BKS09}. 
Formally, the lower bounds stated above for this regime 
are proven in two different oracle models, related to the LCU model and the sparse matrix model. In the latter case, we leave open the question of a low-energy simulation algorithm that is optimal in the sparsity of the Hamiltonian; see Ref.~\cite{low2019sparse} for related work. 
For this regime, we also move outside the oracle model and show our lower bound in terms of gate complexity.

\begin{theorem}[Gate complexity lower bound for time evolution of low-energy states of gap-amplifiable Hamiltonians, informal]
\label{thm:gateinf}
Let $H = \lambda A^\dagger A$ be any $k$-local gap-amplifiable Hamiltonian acting on $\cH$, where $\lambda > 0$, the dimension of $A$ is $M \times N$, and $k$ is a constant. Let $\ket{\psi}\in \cS_{\Delta}$ be any state supported only on the subspace of eigenvalues of $H$ that are at most $\Delta$ for some $\Delta > 0$. For some projector $\Pi$, preparing a unitary $U$ satisfying $|\bra{\psi}\Pi U^\dagger \Pi e^{-itH}\ket{\psi}| \geq 2/3$, where $t > 0$, must use
\begin{align}
    \tilde \Omega\left(t\sqrt{\lambda \Delta}\right)
\end{align}
two-qubit gates, where $\tilde \Omega$ ignores logarithmic factors. The algorithm may use an unlimited number of ancilla qubits and the gates may be non-local and come from a possibly infinite gate set.
\end{theorem}

A formal statement of this result is given in Sec.~\ref{sec:gate}.
The proof relies on the Feynman-Kitaev clock Hamiltonian in an argument similar to that of Ref.~\cite{jordan2017fast}, where the low-energy condition is achieved by controlling the width and momentum of the initial wavepacket. Time evolution evaluates the computation encoded in the Hamiltonian, allowing low-energy simulation to perform arbitrary quantum computations. We show the resulting circuit lower bound using an argument based on computing distinct Boolean functions due to Ref.~\cite{haah2018quantum}.

The rest of this article is structured as follows. 
We begin by presenting our quantum algorithm in Sec.~\ref{sec:alg}, where we provide the proof of Thm.~\ref{thm:mainalgorithm}. In Sec.~\ref{sec:FrustrationFree}, we explain how to improve Hamiltonian simulation of low-energy states of general gap-amplifiable Hamiltonians, including frustration-free Hamiltonians.
In Sec.~\ref{sec:nff}, we present formal lower bounds for the complexity of simulating the time evolution of low-energy states of generic Hamiltonians, that is, a no fast-forwarding result. In Sec.~\ref{sec:nffsga} we present query lower bounds for the case of gap-amplifiable Hamiltonians. In Sec.~\ref{sec:gate} we show gate complexity lower bounds for the time-dominated regime of parameters.
We give some conclusions in Sec.~\ref{sec:conclusions}. For completeness, Appendix~\ref{app:polyapprox} presents a detailed discussion on polynomial approximations to the exponential function (corresponding to the evolution operator), and Appendix~\ref{app:por} gives a lower bound of the $K$-partition quantum search problem (i.e., computing $\mathrm{PARITY}\circ\mathrm{OR}$), which we reduced to Hamiltonian simulation to show the lower bounds in Secs.~\ref{sec:nff} and~\ref{sec:nffsga}. Appendices~\ref{app:spectrumHtilde} and~\ref{app:gates} contain portions of the proofs for query complexity lower bounds in the sparse model and for gate complexity lower bounds, respectively.

\section{Quantum algorithm for time evolution of low-energy states}
\label{sec:alg}

The proof of Thm.~\ref{thm:mainalgorithm} 
is a direct consequence of Lemma~\ref{lem:GAsimulation}, which is proven in Secs.~\ref{sec:sga} and~\ref{sec:qsvt} below, 
and the efficient construction of the block-encoding of $A$ implied by Lemma~\ref{lem:hsgablec}, which is proven in Sec.~\ref{sec:sgablenc}. 

We first sketch the quantum algorithm for the case of gap-amplifiable Hamiltonians of the form $H=\lambda A^\dagger A$,
when the initial state is supported on the subspace of energy at most $\Delta >0$. For any $\Gamma \geq \Delta$, the state is also supported in a subspace of energy at most $\Gamma$. Since we will eventually optimize $\Gamma$ to minimize the query complexity, we will write everything in terms of a state confined to low-energy subspace with energy at most $\Gamma$ in $H$, where $0 < \Delta \le \Gamma \leq \lambda$ restricts the choice of $\Gamma$.

We assume access to the block-encoding of $A$ and use it to construct a new Hamiltonian $\hsga$ acting on an enlarged space that functions as the ``square root'' of $H$. The quantum algorithm implements the evolution operator $e^{-it \lambda (\hsga/\sqrt{\lambda})^2}$ that, after tracing out the ancillary systems, gives $e^{-itH}$. The important property is that we only need to construct an approximation of  $e^{-it \lambda (\hsga/\sqrt{\lambda})^2}$ when acting on an initial state supported exclusively on a subspace associated with
eigenvalues of $\hsga/\sqrt{\lambda}$ in 
\begin{align}
\label{eq:sgainterval}
  r_\Gamma:=  \left[-\sqrt{\frac{\Gamma}{\lambda}},\sqrt{\frac{\Gamma}{\lambda}}\right]\;.
\end{align} 
The problem is analogous to finding a polynomial approximation to $e^{-it\lambda x^2}$, where $x \in r_\Gamma$.
 This allows us to use a polynomial of lower degree than that needed for the general case (i.e., when the approximation is needed for all $x \in [-1,1]$). Once the polynomial is obtained, 
 the operation to implement is simply this polynomial where we replace $x \rightarrow \hsga/\sqrt{\lambda}$. 
 QSVT provides the quantum algorithm to implement this operation, which is a sequence of controlled queries to the block-encoding of $\hsga/\sqrt{\lambda}$ interleaved with other simple two-qubit gates. This is discussed in great detail in Ref.~\cite{gilyen2019quantum}. 
The block-encoding of $\hsga/\sqrt{\lambda}$ is implemented using a constant number of queries to the block-encoding of $A$ and its inverse. The query complexity of this algorithm is then determined by the degree of the approximating polynomial. 

If access to the Hamiltonian is not initially presented in gap-amplifiable form (i.e., access to the block-encoding of $A$), but instead access is given to the block-encoding $T$ of $H'$ with low-energy subspace $[-1, -1+\Delta/\lambda]$, then the Hamiltonian can be expressed in gap-amplifiable form after simple manipulations. This is captured by Lemma~\ref{lem:hsgablec}, which states that it is possible to efficiently construct the block-encoding of $A$ using only a constant number of queries to controlled-$T$ and its inverse. Then, the above algorithm can be used for this Hamiltonian as well. 

In comparison, Corollary 2 of Ref.~\cite{low2017hamiltonian} describes a quantum algorithm for the same setting that uses
\begin{align}
\label{eq:low2017hamiltonian}
    \cO\left(t\sqrt{\lambda\Gamma}\log^{3/2}\left(\frac{t\lambda}{\epsilon}\right) + \sqrt{\frac{\lambda}{\Gamma}}\log^{5/2}\left(\frac{t\lambda}{\epsilon}\right)\right)
\end{align}
queries to controlled-$T$ and its inverse. The additional overhead of $\log^{3/2}(t\lambda/\epsilon)$ originates from Lemma 16 of Ref.~\cite{low2017hamiltonian}, which defines a polynomial filter that uniformly maps $[-1, -1+\Gamma/\lambda]$ to $[-1, 0]$. The uniformity of this spectral amplification is not strictly necessary to evaluate time evolution; we instead use a non-uniform method of spectral gap amplification, yielding a more efficient algorithm that removes some logarithmic factors in Eq.~\eqref{eq:low2017hamiltonian}.

\subsection{Spectral gap amplification and block-encoding of \texorpdfstring{$\hsga$}{SGA Hamiltonian}}
\label{sec:sga}

For the first step of the proof of Lemma~\ref{lem:GAsimulation}, we use the spectral gap amplification
technique of Ref.~\cite{somma2013spectral}, which we revisit here.
Consider a gap-amplifiable Hamiltonian
that can be expressed as $H=\lambda A^\dagger A$ for some operator $A$ that can be represented as an $M \times N$ matrix. (In Sec.~\ref{sec:FrustrationFree}
we generalize this definition.)
We define a new $(M+N)\times (M+N)$ Hamiltonian
\begin{align}
\label{eq:sgainblockform}
    H_{\rm SGA}:=\sqrt \lambda \begin{pmatrix} {\bf 0} & A^\dagger \cr A& {\bf 0}  \end{pmatrix} \;.
\end{align}
We write ${\bf 0}$ to denote the all-zero matrices whose dimensions are implied from context. 
The central reason that $H_{\rm SGA}$ is useful is the property
\begin{align}
    ( H_{\rm SGA})^2 = \begin{pmatrix} H & {\bf 0} \cr {\bf 0}  & \lambda A A^\dagger \end{pmatrix} \;.
\end{align}
The first block of $( H_{\rm SGA})^2$ contains the Hamiltonian to be simulated, i.e., $\Pi ( H_{\rm SGA})^2 \Pi=H$. Hence, we can think of $H_{\rm SGA}$ as the ``square root'' of $H$. 
We diagonalize $\hsga$ explicitly to show that its eigenvalues are the square roots of the eigenvalues of $H$\footnote{Among other consequences, this property implies that the spectral gap of $\hsga$ is larger than that of $H$, allowing for faster adiabatic state transformations, which was the problem studied in Ref.~\cite{somma2013spectral}}. Let $H$ have eigenstates $\ket{\psi_j} \in \cH$ and eigenvalues $\gamma_j \ge 0$, $j \in [N]:=\{1,\ldots,N\}$.
Let $\ket{\Psi_j} \in \mathbb C^{M+N}$ be the corresponding states in the enlarged space by
setting the last $M-N$ amplitudes in the computational basis to zero.
Define
\begin{align}
\label{eq:sgaeigs}
    \sket{\phi_j^{\pm}} = \frac 1 {\sqrt 2} \left(  \ket{\Psi_j} \pm \frac 1 {\sqrt{\gamma_j}}\hsga  \ket{\Psi_j}\right)\;.
\end{align}
These are eigenstates of $\hsga$, since
\begin{align}
\hsga \sket{\phi_j^{\pm}} &= \frac{1}{\sqrt 2}\left(\hsga \ket{\Psi_j} \pm \frac{1}{\sqrt{\gamma_j}}  \hsga^2 \ket{\Psi_j}\right) \\
&=\frac{1}{\sqrt 2}\left(\hsga \ket{\Psi_j} \pm   \sqrt{\gamma_j}\ket{\Psi_j}\right)\\
&= \pm \sqrt{\gamma_j} \sket{\phi_j^\pm}\;.
\end{align}

It is possible to simulate certain transformations 
on $H$, such as the time evolution $e^{-itH}$, using $\hsga$. For instance, for any $\ket \psi  \in \cH$, a useful identity for our quantum algorithm is
\begin{align}
\label{eq:simulationequivalence}
  \Pi e^{-it(\hsga)^2} \ket \Psi  = \Pi e^{-it\lambda (\hsga/\sqrt \lambda)^2} \ket  \Psi  = e^{-itH} \ket \psi \;,
\end{align}
where $\ket \Psi \in \mathbb C^{M+N}$ is the version of $\ket \psi  \in \cH$ in the enlarged space (after setting the last $M-N$ amplitudes in the computational basis to zero).
If access to a block-encoding of $\hsga/\sqrt \lambda$ is given, 
the operator $e^{-itH}$ can be simulated via QSVT using such a block-encoding.

A technical challenge to construct this block-encoding is that $A$ and $A^\dagger$
can be, in principle, rectangular matrices.
It is possible, however,
to pad these matrices with zeroes to make them square. Without loss of generality, we can take $M \geq N$ by adding zeros when $M < N$.
Rather than considering $A$, we can then consider the $M \times M$ matrix
\begin{align}
\label{eq:Asquareversion}
    {\bf A}:= \begin{pmatrix} A & \vline & \begin{matrix}
        0 & \ldots &0 \cr \vdots & \ldots & \vdots \cr 0 & \ldots & 0
    \end{matrix}\end{pmatrix} \;,
\end{align}
where the last $M-N$ columns are zero. 
Accordingly, instead of using $\hsga/\sqrt{\lambda}$  to simulate the time evolution of $H$ as in Eq.~\eqref{eq:simulationequivalence}, 
we can use ${\bf \hsga}/\sqrt{\lambda}:=\ket 0 \! \bra 1 \otimes {\bf A}^\dagger + \ket 1 \! \bra 0 \otimes {\bf A}$, where $\bf \hsga$ is dimension $2M \times 2M$. 
To describe the quantum algorithm as a sequence of operations acting on qubit systems, we assume without loss of generality that $M=2^m$ and $N=2^n$ are powers of two, for $m \ge n$.
Then, $\bf \hsga$ acts on the space of $m+1$ qubits and Eq.~\eqref{eq:simulationequivalence} implies
\begin{align}
\label{eq:simulationequivalence2}
      e^{-it({\bf \hsga})^2} \ket{0}^{\otimes m+1-n}\ket \psi  =   e^{-it\lambda ({\bf \hsga}/\sqrt \lambda)^2} \ket{0}^{\otimes m+1-n}\ket  \psi  =\ket{0}^{\otimes m+1-n} (e^{-itH} \ket \psi) \;.
\end{align}
Time evolution with $H=\lambda A^\dagger A$ can then be simulated via QSVT when access to a block-encoding $V_{\rm SGA}$ of 
${\bf \hsga}/\sqrt \lambda$ is allowed.
The $m+1-n$ ancillas can be discarded at the end.

Let $V$ be the block-encoding of $A$ in Eq.~\eqref{eq:blecofA} that satisfies $\Pi' V \Pi=A$, where $\Pi'$ and $\Pi$
are the orthogonal projectors onto $\cH'\equiv \mathbb C^M$ and $\cH \equiv \mathbb C^N$, respectively. Assume that $V$ acts on the space $\mathbb C^Q$ of $q=\log_2 Q \ge m$ qubits. 
Let $P':=\ketbra 0^{\otimes m-n} \otimes \one_N$
be the projector associated with $\Pi'$ that acts on $\cH'$ and $P'':= \ketbra 0^{\otimes q-n}\otimes \one_N$ a projector in $\mathbb C^Q$. Then,
\begin{align}
   { \bf A }&= \Pi' V \Pi' P' = \Pi' V P'' \Pi'\;,
\end{align}
where $\Pi' V \Pi'$ corresponds to the first $M \times M$ dimensional block of $V$ and $VP''$ is such that the first $N$ columns are those of $V$ and the last $Q-N$ columns are zero; essentially, we constructed a block-encoding of  ${ \bf A }$, which is the matrix in Eq.~\eqref{eq:Asquareversion}.
We can write
\begin{align}
    {\bf \hsga}/\sqrt{\lambda}&=\left( \ketbra 0  \otimes {\bf A}^\dagger + \ketbra 1 \otimes {\bf A}
    \right) (X \otimes \one_M) \\
    \label{eq:closetoblechsga}
    & =  \Pi''\left( \ketbra 0  \otimes P'' V^\dagger + \ketbra 1 \otimes VP''
    \right) (X \otimes \one_M) \Pi'' \;,
\end{align}
where $X$ is the one-qubit Pauli flip operator, and $\Pi''$ is now the orthogonal projector onto the space of $m+1$ qubits, i.e., $\mathbb C^{2M}$. This
is the projector associated with the subspace where 
the $m$ qubits associated with $\mathbb C^M$ are in $\ket 0^{\otimes m}$.
Equation~\eqref{eq:closetoblechsga} is not yet a proper block-encoding  of 
${\bf \hsga}/\sqrt{\lambda}$ because $P''$ is not unitary. However, $P''$ can be expressed as an LCU of the form $P''=(U_{P''}+\one_Q)/2$, where 
$U_{P''}:=2P''-\one_Q$ is a unitary acting on $\mathbb C^Q$. 
There exist standard techniques to construct a block-encoding of an LCU; see, for example,  the results in Lemma 6 of Ref.~\cite{CKS17}
or Ref.~\cite{low2019qubitization}. More specifically, we can define the unitary that acts on a space of $q+2$
qubits
\begin{align}
    V':= \ketbra {00} \otimes U_{P''}V^\dagger + \ketbra {01}  \otimes  V^\dagger + \ketbra{10} \otimes V U_{P''} +\ketbra {11} V\;,
\end{align}
and also define
\begin{align}
\label{eq:hsgablec}
    V_{\rm SGA}:= (\one_2 \otimes {\rm H} \otimes \one_Q) V' ( X \otimes {\rm H} \otimes \one_Q) \;,
\end{align}
where ${\rm H}$ is the one-qubit Hadamard gate that implements ${\rm H} \ket 0=(\ket 0 + \ket 1)/\sqrt 2$ and ${\rm H} \ket 1=(\ket 0 - \ket 1)/\sqrt 2$. Note that the dimension of $V'$ and $V_{\rm SGA}$ is $4Q \times 4Q$ in this example. Then, if $\Pi''$ is the orthogonal projector onto $\mathbb C^{2M}$, which is the space associated with ${\bf \hsga}$, we obtain
\begin{align}
   \Pi'' V_{\rm SGA}\Pi''={\bf \hsga}/\sqrt{\lambda} \;.
\end{align}
The subspace associated with
$\Pi''$ is now that where qubits 2 through 
$m+2$ are in $\ket 0^{\otimes m+1}$.

Equation~\eqref{eq:hsgablec} is the desired block-encoding
of ${\bf \hsga}/\sqrt{\lambda}$. To implement it, we need to apply $V'$ once, which requires applying controlled-$V$ and controlled-$V^\dagger$
once. The additional gate complexity is dominated by that of implementing the controlled-$U_{P''}$ operation. This gate complexity is $\cO(q)$ using standard techniques. 

Having constructed a block-encoding
of ${\bf \hsga}/\sqrt{\lambda}$, it remains 
to show: i) how to simulate the time evolution of $H=\lambda A^\dagger A$ using Eq.~\eqref{eq:simulationequivalence2}, and ii) how to equivalently express any Hamiltonian as $\lambda A^\dagger A$ from its block-encoding $T$.
These are discussed next.

\subsection{Polynomial approximation of \texorpdfstring{$e^{-itx^2}$}{transformation} and QSVT}
\label{sec:qsvt}

For the second step of the proof of Lemma~\ref{lem:GAsimulation},
we follow an approach similar to Hamiltonian simulation via QSVT. Briefly, we review the standard approach for simulating $e^{-itH}$ given access to a block-encoding of $H/\lambda$, for any Hamiltonian $H$~\cite{gilyen2019quantum}. QSVT starts from proper even and odd degree polynomial approximations to the trigonometric functions $\cos(t\lambda x)$ and $\sin(t\lambda x)$, which can be obtained using a Jacobi-Anger expansion (Thm. 57 of Ref.~\cite{gilyen2019quantum}). Here, $x \in \mathbb R$
plays the role of an eigenvalue of $H/\lambda$, and after
replacing $x \rightarrow H/\lambda$ in the polynomials, we produce approximations of $\cos(tH)$ and $\sin(tH)$, respectively. 
QSVT can then implement these polynomials using 
the block-encoding of $H/\lambda$ a number of times that depends on the degree of the polynomials. They can be later combined through an LCU, ultimately implementing an approximation of $e^{-itH}=\cos(tH)-i \sin(tH)$. See Thm. 58 of Ref.~\cite{gilyen2019quantum} for more details. 

We will now modify this procedure to accommodate a block-encoding of ${\bf H}_{\rm SGA}/\sqrt \lambda$ instead of $H/\lambda$, applicable to gap-amplifiable Hamiltonians.
For our problem, we start from even and odd polynomial approximations to $\cos(t\lambda x^2)$ and $\sin(t\lambda x^2)$ instead. Here, $x \in \mathbb R$  plays the role of a certain eigenvalue of ${\bf H}_{\rm SGA}/\sqrt \lambda$, which 
coincides with an eigenvalue of $\pm \sqrt{H/\lambda}$.
After replacing $x \rightarrow {\bf H}_{\rm SGA}/\sqrt \lambda$,
these polynomials produce approximations of $\cos(t ({\bf H}_{\rm SGA})^2)$
and $\sin(t ({\bf H}_{\rm SGA})^2)$. QSVT can then implement these polynomials using 
the block-encoding of ${\bf H}_{\rm SGA}/\sqrt \lambda$.
These can be later combined through an LCU, ultimately implementing an approximation of  $e^{-it ({\bf H}_{\rm SGA})^2} = \cos (t   ({\bf H}_{\rm SGA})^2) - i \sin (t ({\bf H}_{\rm SGA})^2)$. This operation also simulates $e^{-itH}$ due to Eq.~\eqref{eq:simulationequivalence2}. 
Importantly, since we assume the initial state
$\ket \psi$ to be supported only on the subspace corresponding to eigenvalues of $H=\lambda A^\dagger A$ at most $\Gamma>0$, we only need to consider these polynomial approximations for the domain $x \in r_\Gamma$  in Eq.~\eqref{eq:sgainterval}.
That is, for given $\epsilon \ge 0$, 
we seek polynomials $P_{\epsilon, \cos}$ and $P_{\epsilon, \sin}$ satisfying
\begin{align}
\label{eq:polyapproxcond1}
    & |P_{\epsilon, f}(x)|  \leq 1 \ \forall \ x \in [-1, 1] \;, \\
    \label{eq:polyapproxcond2}
   & \max_{x \in r_\Gamma} |P_{\epsilon,f}(x) - f(x)| \leq \epsilon \;,
\end{align}
for $f(x)=\cos(t\lambda x^2)$ and  $f(x)=\sin(t\lambda x^2)$. The techniques of piecewise polynomial constructions outlined in Ref.~\cite{gilyen2019quantum} are useful for this task. In particular, we use a slightly modified version of Corollary 66 of Ref.~\cite{gilyen2019quantum}.

\begin{lemma}[Piecewise polynomial approximations based on a local Taylor series]
\label{lem:polyapprox}
Let $r \in (0, 1]$, $\delta \in (0,r]$, and let $f:[-r-\delta, r+\delta]$ be such that $f(x) = \sum_{k=0}^\infty a_k x^k$ for all $x \in [-r-\delta, r+\delta]$. Choose any $B>0$ such that $B \geq \sum_{k=0}^\infty  |a_k| (r+\delta)^k$. For $\epsilon \in (0, 3B/2)$, there is an efficiently computable polynomial $P$ of degree $\cO(\frac{1}{\delta}\log\frac{B}{\epsilon})$ such that
\begin{align}
    \|f(x)-P(x)\|_{[-r,r]} \leq \epsilon \quad \mathrm{and} \quad \|P(x)\|_{[-1, 1]} \leq \epsilon + \|f(x)\|_{[-r-\frac{\delta}{2}, r+\frac{\delta}{2}]}\;.
\end{align}
\end{lemma}

We provide a self-contained proof of this Lemma in Appendix~\ref{app:polyapprox}. Here, we neglect the cost of classical processing (e.g., for finding the polynomial coefficients), which remains efficient; it is discussed briefly in Appendix~\ref{app:polyapprox} and in depth in Ref. \cite{gilyen2019quantum}.

For the specific case of the sine and cosine functions, Lemma~\ref{lem:polyapprox} implies the following.

\begin{lemma}
    Let $\epsilon > 0$ be the error, $t \in \mathbb R$, $\Gamma > 0$, and $\lambda \ge \Gamma$. Then, there exists polynomial approximations $P_{\epsilon,f}(x)$ to $f(x)=\sin (t \lambda x^2)$ and $f(x)=\cos (t \lambda x^2)$ of degree
    \begin{align}
         \cO\left(t\sqrt{\lambda \Gamma} + \sqrt{\frac{\lambda}{\Gamma}}\log \frac{1}{\epsilon}\right)
    \end{align}
    that satisfy Eqs.~\eqref{eq:polyapproxcond1} and~\eqref{eq:polyapproxcond2} with error $2 \epsilon$.
\end{lemma}
\begin{proof}
Let $r =\delta =  \sqrt{\Gamma/\lambda}$. The Taylor expansions of $\sin(t\lambda x^2)$ and $\cos(t\lambda x^2)$ at $x = r+\delta$ are such that
\begin{align}
    B_\mathrm{sin} &= \sinh\left[\lambda t \left( {2\sqrt{\Gamma/\lambda}}\right)^2\right] < e^{4\Gamma |t|}, \quad 
    B_\mathrm{cos} = \cosh\left[\lambda t \left( {2\sqrt{\Gamma/\lambda}}\right)^2\right] < e^{4\Gamma |t|}\;.
\end{align}
We can then choose $B = \exp(4\Gamma |t|)$ for these cases. Applying Lemma~\ref{lem:polyapprox} to polynomial approximations of $\sin(t\lambda x^2)$ and $\cos(t\lambda x^2)$ with error $\epsilon$, we obtain $\frac{1}{\delta}\log(B/\epsilon) = \sqrt{\frac{\lambda}{\Gamma}} \log \frac{e^{4\Gamma t}}{\epsilon}$. Hence, for $t > 0$, a polynomial of degree
\begin{align}
    \cO\left(t\sqrt{\lambda \Gamma} + \sqrt{\frac{\lambda}{\Gamma}}\log \frac{1}{\epsilon}\right)
\end{align}
suffices to produce the approximating polynomials  for the sine and cosine functions.

However, Lemma~\ref{lem:polyapprox} only provides normalization up to $\|P(x)\|_{[-1,1]} \leq 1+\epsilon$ since $\sin$ and $\cos$ satisfy $\|f(x)\|_{[-r-\delta/2,r+\delta/2]} \leq 1$. Since QSVT can implement polynomials normalized to 1 or less, we rescale these polynomial approximations by $\frac{1}{1+\epsilon}$. By the triangle inequality 
\begin{align}
    \left\|\frac{1}{1+\epsilon}P(x) - f(x)\right\|_{[-r, r]} \leq \frac{1}{1+\epsilon} \|P(x) - f(x)\|_{[-r,r]} + \frac{\epsilon}{1+\epsilon}\|f(x)\|_{[-r,r]} \leq \frac{2\epsilon}{1+\epsilon} < 2\epsilon\;.
\end{align}
Consequently, this rescaling only increases the approximation error from $\epsilon$ to $2\epsilon$ in the worst case. 
The result is two polynomials that approximate $\cos(t\lambda x^2)$ and $\sin(t\lambda x^2)$ within error $2 \epsilon$ in the domain $x \in r_\Gamma$.
\end{proof}

We now have suitable polynomial approximations of $\sin(t\lambda x^2)$ and $\cos(t\lambda x^2)$. 
We can then use QSVT to simulate the Hamiltonian as explained above.
This requires access to the block-encoding $V_{\rm SGA}$
of  ${\bf H}_{\rm SGA}/\sqrt \lambda$; see Eq.~\eqref{eq:hsgablec}.
To apply the corresponding polynomials, QSVT uses the controlled-$V_{\rm SGA}$ and its inverse
a number of times determined by the degree of the polynomials,
which was shown to be $\cO(t\sqrt{\lambda \Gamma} + \sqrt{\frac{\lambda}{\Gamma}}\log \frac{1}{\epsilon})$.
Together with Sec.~\ref{sec:sga}, this completes the proof of Lemma~\ref{lem:GAsimulation}.
We refer to Ref.~\cite{gilyen2019quantum} for more details on QSVT.

\subsection{Efficient implementation of \texorpdfstring{$V$}{V}}
\label{sec:sgablenc}

The main result in Thm.~\ref{thm:mainalgorithm} is stated in terms of the block-encoding $T$, but our quantum algorithm is essentially that for gap-amplifiable Hamiltonians, i.e., Lemma~\ref{lem:GAsimulation}. Now
we prove Lemma~\ref{lem:hsgablec}, which relates $T$ to the block-encoding $V$ of a corresponding $A$.

The block-encoding $T$ is an $M \times M$ unitary matrix, where we assume $M \ge N$ without loss of generality. It satisfies $\Pi T \Pi = H'=(H-E)/\lambda$, for some $E \in \mathbb R$ and $\lambda >0$; see Eq.~\eqref{eq:blenc}. Here, $\Pi$ is the projector onto $\cH=\mathbb C^N$.  Let 
\begin{align}
\label{eq:Achoice}
    A:=\left( \frac {T^\dagger+ \one_M}2 \right) \Pi \;
\end{align} 
be of dimension $M \times N$.
It follows that, 
\begin{align}
    2 \lambda A^\dagger A &=2 \lambda \Pi \left( \frac {T+ \one_M}2 \right)\left( \frac {T^\dagger + \one_M}2 \right) \Pi \\
    & = \lambda (\one_N + (H-E)/\lambda) \\
    & = H - F\;,
\end{align}
where $F=E-\lambda$. This proves that, when access to $T$ and $T^\dagger$ is given, the Hamiltonian can be equivalently expressed as a gap-amplifiable Hamiltonian (up to an additive constant). 

Then, to simulate $H$, or equivalently $2\lambda A^\dagger A$, using the approach for gap-amplifiable Hamiltonians of Lemma~\ref{lem:GAsimulation} (also discussed in 
Secs.~\ref{sec:sga} and~\ref{sec:qsvt}), we need to construct the block-encoding $V$ of $A$ defined in Eq.~\eqref{eq:Achoice}. This satisfies
$\Pi' V \Pi=A$ for corresponding projectors $\Pi'$ and $\Pi$.
The only remaining challenge is that the term in brackets in Eq.~\eqref{eq:Achoice} is an LCU rather than unitary. Nevertheless, we can again use standard techniques to block-encode an LCU. For example, we can define
\begin{align}
\label{eq:blecAfromT}
    V:= ({\rm H} \otimes \one_M) \left( \ketbra 0 \otimes T^\dagger + \ketbra 1 \otimes \one_M \right) ({\rm H} \otimes \one_M) \;,
\end{align}
where H is the Hadamard gate, which implies
\begin{align}
    \bra 0 V \ket 0 = \left(\frac{T^\dagger+\one_M} 2 \right) \;.
\end{align}
Equivalently, $\Pi' V \Pi = A$, which is the desired result. In the notation of Sec.~\ref{sec:sga}, the unitary $V$ of Eq.~\eqref{eq:blecAfromT} acts on the space $\mathbb C^{Q}$ of $q=m+1$ qubits, and $\Pi'$ is the projector that takes $\mathbb C^{2M}$ to $\mathbb C^M$.

Implementing $V$ requires using a controlled-$T^\dagger$ operation once, in addition to a constant number of two-qubit gates.
Also, if the initial state $\ket \psi$ is supported on the subspace associated with eigenvalues of $H'$ in $[-1,-1+\Delta/\lambda]$, then it is supported on the subspace associated with energies at most $\Delta$ of the Hamiltonian $2 \lambda A^\dagger A$. This proves Lemma~\ref{lem:hsgablec} that, together with
Lemma~\ref{lem:GAsimulation}, imply Thm.~\ref{thm:mainalgorithm}. We note that the factor of $2$ that appears in the block-encoding $2\lambda A^\dagger A$ occurs because $\Delta$ is restricted to $[0, 2\lambda]$ to ensure the block-encoding $T$ has eigenvalues in $[-1, 1]$. Although we took $H=\lambda A^\dagger A$ in the previous subsections, the constant factor does not affect the reported query complexity of $\cO(t\sqrt{\lambda \Gamma} + \sqrt{\lambda/\Gamma}\log(1/\epsilon))$ in Thm.~\ref{thm:mainalgorithm}.

\section{Quantum algorithm for time evolution of low-energy states of general gap-amplifiable Hamiltonians}
\label{sec:FrustrationFree}

Our quantum algorithm for simulating time evolution of low-energy states is based on that for simulating Hamiltonians that can be expressed as $H=\lambda A^\dagger A$, for $\lambda >0$. 
In this case, if the initial state is of low-energy $\Delta \ll \lambda$ (e.g., $\Delta=o(\lambda)$), Lemma~\ref{lem:GAsimulation} can give an asymptotic improvement over generic methods like QSP.
We called these ``gap-amplifiable Hamiltonians'', but this definition can be extended to also address a more general case where $H=\sum_{l=1}^L \lambda_l A_l^\dagger A_l^{}$, $\lambda_l >0$ for all $l \in [L]$, and $A_l$ are matrices of dimension $M \times N$. 
Examples of such Hamiltonians are frustration-free Hamiltonians~\cite{schuch2008computational,BOO10,BT09,chen2012ground}, where each term $\lambda_l A_l^\dagger A_l^{}$ might correspond to a local interaction between spins or other physical systems. 
In this section we show how to use the quantum algorithm of Lemma~\ref{lem:GAsimulation} to also give an improved simulation method for these ``general gap-amplifiable Hamiltonians''.

\subsection{General gap-amplifiable Hamiltonians}
\label{sec:gap}

We begin by providing a general definition of a $(\lambda,L)$-gap-amplifiable Hamiltonian.

\begin{definition}[General gap-amplifiable Hamiltonian]
\label{def:h}
A Hamiltonian $H \succeq 0$ acting on $\cH \equiv \mathbb C^N$ is said to be 
$(\lambda,L)$-\emph{gap-amplifiable} if there exists 
$L \ge 1$ such that: i) $H$ can be expressed as
$H =\sum_{l=0}^{L-1} \lambda_l A_l^\dagger A_l^{}$, where $\lambda_l>0$ for all $l \in [L]$, ii) efficient
 access to block-encodings of each $A_l$ are provided, and iii) $\sum_{l=0}^{L-1} \lambda_l \le \lambda$.
\end{definition}

The case $L=1$ was  analyzed in Sec.~\ref{sec:alg}.
In general, the operators $A_l$ can be expressed as
rectangular matrices of dimension $M \times N$. 
A block-encoding $V_l$ of $A_l$ might involve two different projectors $\Pi$ and $\Pi'$ onto $\mathbb C^N$ and $\mathbb C^M$, respectively, such that $\Pi' V_l \Pi = A_l$.
The unitaries $V_l$ might act on a larger space $\mathbb C^Q$, where $Q \ge \max(M,N)$.

Each block-encoding $V_l$ will in practice incur some cost to construct. This cost can be measured as a query complexity to an oracle that provides some other form of access to $A_l$. Later, in Secs.~\ref{sec:accessmodels} and~\ref{sec:accessGAcase}, we will discuss two common access models for quantum algorithms: 
\begin{enumerate}[label=\roman*),wide, labelwidth=0pt, labelindent=0pt]
\item the  LCU model, which provides coefficients $\beta_{l,j} > 0$ and access to unitaries $U_{l,j}$, $l \in [L]$ and $j \in [J]$, such that $A_l = \frac 1 {\beta_l }\sum_{j=1}^J \beta_{l,j} U_{l,j}$, for $\beta_l:=\sum_{j=1}^J \beta_{l,j}$, and
\item the sparse matrix model, which provides access to the positions and values of nonzero matrix elements of each sparse $A_l$, $l \in [L]$.
\end{enumerate}
In our definition, we assume \emph{efficient access} to the block-encoding of $A_l$; this refers to requiring only a constant number of queries to the oracles in the LCU or sparse matrix model to construct each $V_l$.

\subsection{Quantum algorithm and complexity}
\label{sec:complexity}

To use the same algorithm as that of Lemma~\ref{lem:GAsimulation},
we begin by showing a simple reduction
from general $(\lambda,L)$-gap-amplifiable Hamiltonians to $(\lambda,1)$-gap-amplifiable Hamiltonians, i.e., $\lambda A^\dagger A$.

\begin{lemma}[Reduction to $(\lambda,1)$-gap amplifiable Hamiltonians] 
\label{lem:GAreduction}
Consider a $(\lambda,L)$-gap-amplifiable Hamiltonian 
$H$ acting on $\cH\equiv \mathbb C^N$ and assume access to the  unitary $W=\sum_{l=0}^{L-1} \ketbra l \otimes V_l$ and its inverse. Each unitary $V_l$ 
is the block-encoding of  the $M\times N$ dimensional $A_l$ and satisfies $A_l=\Pi' V_l \Pi$, where $\Pi'$ and $\Pi$ are orthogonal projectors onto $\cH'\equiv\mathbb C^M$
and $\cH \equiv \mathbb C^N$, respectively. 
Then, there exists a matrix $A$ of dimension $(LM)\times N$ such that $H$ can be alternatively expressed as $H=\lambda A^\dagger A$. A block-encoding $V$ of $A$  can be constructed using   controlled-$W$ 
once, and additional $\cO(L)$ two-qubit gates. Moreover, if $\ket \psi \in \cS_{\Delta}$ is any state supported exclusively on the subspace corresponding to energies (eigenvalues) of $H$ in $[0, \Delta]$, then $\ket \psi$ is supported exclusively on a subspace corresponding to energies at most $\Delta$ of $\lambda A^\dagger A$.
\end{lemma}

\begin{proof}
    The statement on the support of $\ket \psi$
    follows directly from the fact that $\lambda A^\dagger A$ is also an expression for $H$, which is PSD by assumption. 
    Since this is $(\lambda,L)$-gap-amplifiable, it can be expressed as 
    $H=\sum_{l=0}^{L-1} \lambda_l A_l^\dagger A_l^{}$, where $\lambda_l>0$ and each $A_l$ is of dimension $M \times N$. We define
    \begin{align}
        A:= \sum_{l=0}^{L-1} \sqrt{\frac{\lambda_l}{\lambda}}\ket l \otimes A_l \;,
    \end{align}
    which is of dimension $M'\times N$, $M':=LM$. Here,  $\lambda=\sum_{l=0}^{L-1} \lambda_l$ and note that
    \begin{align}
        A^\dagger A = \frac 1 {\lambda} \sum_{l=0}^{L-1} \lambda_l A^\dagger_l A^{}_l \;.
    \end{align}
   This implies $\|A\|\le \frac 1\lambda \sum_l \lambda _l \|A^\dagger_l A^{}_l\|\le 1$ and 
   \begin{align}
       H = \lambda A^\dagger A \;.
   \end{align}
   We seek a block-encoding of this $A$, which would allow us to simulate $H$ as explained in Sec.~\ref{sec:sga}. Assume $L=2^l$ without loss of generality and 
    let $G$ be a unitary that performs the map
    \begin{align}
      G:\ket{0} \to  \sum_{l=0}^{L-1} \sqrt{\frac{\lambda_l}{\lambda}}\ket{l}\;,
\end{align}
where $\ket 0 = \ket 0^{\otimes l}$ in this case.
This can be implemented with $\cO(L)$ two-qubit gates using standard methods. We also assume that each block-encoding $V_l$ of $A_l$ acts on the space of $q$ qubits, i.e., on $\mathbb C^{Q}$ for $Q=2^q$. Note that
\begin{align}
   W (G \ket 0 \otimes \one_Q) = \sum_{l=0}^{L-1} \sqrt{\frac{\lambda_l}{\lambda}} \ket l \otimes V_l \;.
\end{align}
This implies that
\begin{align}
    \Pi' W (G \otimes \one_Q) \Pi = \sum_{l=0}^{L-1} \sqrt{\frac{\lambda_l}{\lambda}} \ket l \otimes A_l =A \;.
\end{align}
The block-encoding is then $V=W (G \otimes \one_Q)$,
which requires one use of $W$ and $G$.
\end{proof}

The ability to simulate low-energy states of a general $(\lambda,L)$-gap-amplifiable Hamiltonian follows as a consequence of the above results. While Thm.~\ref{thm:mainalgorithm} is phrased in terms of queries to a Hamiltonian block-encoding with low-energy subspace $[-1, -1+\Delta/\lambda]$, the phrasing below is more natural for settings such as frustration-free systems.

\begin{theorem}[Quantum algorithm for time evolution of low-energy states of general gap-amplifiable Hamiltonians]
\label{thm:gapalgorithm}
Let $\epsilon > 0$ be the error and $t>0$ be the evolution time. Consider a $(\lambda,L)$-gap-amplifiable Hamiltonian $H=\sum_{l=0}^{L-1} \lambda_l A^\dagger_l A_l$ acting on $\cH \equiv \mathbb C^N$, where $\lambda_l >0$, and each $A_l$ has dimension $M \times N$. Assume access to unitaries $V_l$ block-encoding each $A_l$. Let $\ket \psi \in \cS_{\Delta}$ be any state supported exclusively on the subspace associated with energies (eigenvalues) of $H$ at most $\Delta > 0$,  and assume $t\Delta \ge \epsilon$. Then, for any $\Gamma$ such that $\Delta \le \Gamma \le \lambda$, there exists a quantum algorithm that implements a unitary $U$ and satisfies 
$|\!\bra{\psi}\Pi U^\dagger \Pi e^{-itH}\ket{\psi}\!| \geq 1 - \epsilon$, using
\begin{align}
    \cO\left(t\sqrt{\lambda\Gamma} + \sqrt{\frac{\lambda}{\Gamma}} \log \frac{1}{\epsilon}\right)
\end{align}
queries to each controlled-$V_l$ and its inverse.
\end{theorem}
\begin{proof}
The result is an immediate consequence of Lemma~\ref{lem:GAreduction}, to construct a block-encoding of $A$ when expressing the Hamiltonian as $H=\lambda A^\dagger A$, and Lemma~\ref{lem:GAsimulation} to simulate the latter. The number of queries to $V=W (G \otimes \one_Q)$ appearing in Lemma~\ref{lem:GAreduction} is given by Eq.~\eqref{eq:gapampsimulation}. Each such query uses $W$ once, and this uses each controlled-$V_l$ once. 
\end{proof}

 When the initial state is such that $\Delta \ll \lambda$, or $\Delta=o(\lambda)$,
and the Hamiltonian is gap-amplifiable, the result is again an improvement over QSP.
Because any Hamiltonian in $\cH$ can be expressed as a sum of PSD terms $\sum_{l=0}^{L-1} \lambda_l A_l^\dagger A_l^{}$ up to a constant, the algorithm of Thm.~\ref{thm:gapalgorithm} can always be applied. However, in the following section, we will see that the low-energy assumption does not result in an improvement over QSP for {\em all} Hamiltonians, even in the LCU and sparse  matrix access models. Our low-energy simulation algorithm does not provide any benefit in general, because offsetting each term in a given $H$ to make it PSD may place the low-energy sector of $H$ to a relatively large energy sector of the shifted Hamiltonian. In other words, $\Delta/\lambda$ can become a constant after the shift, preventing any fast-forwarding from the low-energy assumption alone. 

Nonetheless, physically relevant classes of Hamiltonians are already $(\lambda,L)$-gap-amplifiable with reasonably small $\Delta/\lambda$. We consider the example of a frustration-free Hamiltonian $H_{\rm FF}=\sum_{l=0}^{L-1} \pi_l$. Each $\pi_l \succeq 0$ might refer to a local term in the Hamiltonian, e.g., involving a few qubits.
The ground state energy of $H_{\rm FF}$ is $E_0=0$. With some preprocessing, it is possible to find the $\lambda_l$'s and $A_l$'s such that $\pi_l=\lambda_l A^\dagger_l A^{}_l$ and construct the block-encodings $V_l$, which are also local operators in this example.  
Implementing the operation $W$ that applies the controlled-$V_l$ will incur in some gate cost, which is $\cO(L)$ for the case of local Hamiltonians. Hence, the total gate complexity to simulate the time evolution of a low-energy state of $H_{\rm FF}$ using our quantum algorithm is $\cO(L t \sqrt{\lambda \Delta})$, assuming constant precision. Applying QSP instead
to this problem would have resulted in gate complexity
$\cO(L t \lambda)$ for this case.

Similarly, if a Hamiltonian is nearly frustration-free, some benefit can be achieved.
To illustrate this, we consider the example where $H = \delta P + H_{\rm FF}$, for some Pauli string $P$, and $0 \le \delta \ll 1$. Hence, $H+\delta = \delta(P+\one_N) + H_{\rm FF}$ is a sum of PSD terms and a $(\lambda+2\delta,L+1)$-gap-amplifiable Hamiltonian. For example, we can choose
$A_0=(P+\one_N)/2$ for the first operator, so that   $A_0^\dagger A_0^{} = (P+\one_N)/2$. For an initial state $\ket{\psi} \in \cS_{\Delta_0}$ of $H$, the low-energy cutoff must now satisfy $\Delta \geq \Delta_0 + \delta$ due to the offset. Since $\lambda \approx \lambda +2 \delta$, we observe that this nearly frustration-free system can still benefit from our quantum algorithm. The additional query complexity  would only be $\cO(t\delta\sqrt{\lambda \Delta_0})$ for $\delta \ll 1$. This observation
hints at improvements over QSP for Hamiltonians that are not necessarily gap-amplifiable; although the 
asymptotic complexity might not change for a particular example, the actual gate complexity by using our algorithm could still give an improvement.

\section{No fast-forwarding theorem for simulating time evolution of low-energy states in general}
\label{sec:nff}

In general, the no fast-forwarding theorems of Refs.~\cite{BAC07,atia2017fast,haah2018quantum} establish a lower bound of $\Omega(t\|H\|)$ queries or gates needed to simulate the time evolution $e^{-itH}$. The result of Ref.~\cite{BAC07}, in particular,  follows from a reduction of the parity problem to Hamiltonian simulation. It is shown that simulating the dynamics of a particular quantum system in one dimension (i.e., a quantum walk) allows one to compute the parity of a bit string of increasing length $n$ after time $t=\Theta(n)$.
A lower bound on the query complexity of parity implies
a lower bound on the number of queries needed for Hamiltonian simulation. 
However, that proof cannot directly extend to simulating low-energy states, since the initial state is not fully confined to a low-energy subspace.

Without relying on the reduction from parity, we may gain some intuition for the simulation lower bound as follows. Hamiltonian simulation methods like QSP implement a polynomial of a unitary, and the query complexity coincides with the degree of this polynomial. In the worst case, the state $\ket \psi$ can be supported on states of
energy as high as $\|H\|$ or as small as $-\|H\|$.
If we use QSP, we would need to 
start from a polynomial approximation of $e^{-i\phi}$, where $|\phi|$ can be as large as $t \|H\|$. For constant error, this is achieved by a polynomial of degree $\Omega(t\|H\|)$. Then, QSP  must query the Hamiltonian
$\Omega(t\|H\|)$ times. 
When confined to the low-energy subspace with largest eigenvalue $\Delta$, this analysis suggests a Hamiltonian simulation lower bound of $\Omega(t\Delta)$ instead.
However, we will show that this intuition is incorrect.

As discussed in the previous section, our quantum algorithm can apply to any Hamiltonian. However, Thm.~\ref{thm:mainalgorithm} improves upon QSP when the block-encoding $T$ of $H'=(H-E)/\lambda$ can be accessed and has the property that the low-energy subspace of $H$ corresponds to the subspace associated with eigenvalues of $H'$ near 
$-1$. From simple inspection, it is not obvious how to efficiently construct such $T$ (and $H'$) when query access to the original $H$ is given (e.g., in form of a block-encoding of $H/\lambda$). Indeed, we find that no speedup is possible for simulating the time evolution of low-energy states of general Hamiltonians, and that the asymptotic complexity of QSP in Eq.~\eqref{eq:QSPcomplexity} remains optimal. Our result applies to any Hamiltonian simulation method that assumes access to $H$ in the LCU model or in the sparse matrix model, which we now formally introduce.

\subsection{Access models}
\label{sec:accessmodels}

We seek to prove that when access to $H$ is given through standard access models it is not generically possible to improve
the query complexity of methods like QSP under the low-energy assumption. To this end, we discuss two of the most common models used in Hamiltonian simulation to provide access to $H$: the LCU and sparse matrix models.

The LCU model (cf.,~\cite{BCC+15}) describes Hamiltonians of the form $H=\sum_{l=0}^{L-1} \beta_l U_l$, where $\beta_l > 0$ and $U_l$ is unitary acting on $\mathbb C^N$ for all $l$. It assumes access to SELECT and PREPARE oracles defined by
\begin{align}
\label{eq:LCUaccess}
    \mathrm{SELECT}(H) = \sum_{l=0}^{L-1} \ketbra l \otimes U_l, \quad 
    \mathrm{PREPARE}(H) : \ket{0} \to \sum_{l=0}^{L-1} \sqrt{\frac{\beta_l}{\beta}}\ket{l}\;,
\end{align}
where $\beta = \sum_{l=0}^{L-1} \beta_l$.
Without loss of generality, we can assume $L=2^l$, so that
$\ket 0$ above refers to the $l$-qubit state $\ket 0^{\otimes l}$. A useful property used by some quantum algorithms is
\begin{align}
\label{eq:LCUblec}
  \frac 1 \beta H=  \bra{0}^{\otimes l} (\mathrm{PREPARE} (H))^\dagger{\rm SELECT}(H) \mathrm{PREPARE}(H) \ket{0}^{\otimes l} \;,
\end{align}
which gives essentially a block-encoding of $H/\beta$.
The query complexity in the LCU model is measured by the number of times these oracles are used, together with their inverses and controlled versions.

In contrast, the sparse matrix model gives access to the  elements of a $d$-sparse matrix $H$~\cite{BC12}. It uses an oracle $O_H$ that allows to perform the transformations
\begin{align}
    \ket{j}\ket{k}\ket{z} \rightarrow \ket{j}\ket{k}\ket{z \oplus H_{jk}}, \quad  \ket{j}\ket{\ell} \rightarrow \ket{j}\ket{\nu(j,\ell)}\;.
\end{align}
Here, $j,k \in [N]:=\{1,\ldots,N\}$ label a row and column of $H$, and $H_{jk}$ is the corresponding matrix entry (specified in binary form). Also, $\ell \in [d]$ and $\nu: [N] \times [d] \rightarrow [N]$ computes the row index of the $\ell^{\rm th}$ nonzero entry of the $j^{\rm th}$ row in place. 
Like in the LCU model, having oracle access to the matrix allows one to construct a block-encoding of the Hamiltonian --- for example, a block-encoding of $H/(d\|H\|_{\max})$
(cf. Sec. 4.1 of Ref.~\cite{low2019qubitization} or Ref.~\cite{CKS17}).
The query complexity in the sparse matrix model is measured by the number of times $O_H$ is used, together with its inverse and controlled version.

A block-encoding of $H/\lambda$, for some $\lambda>0$, can be constructed using either the LCU or the sparse matrix model. 
Compared to simply accessing the block-encoding directly, both the LCU and sparse access models provide stronger access to $H$. Our lower bounds for these models are accordingly stronger than restricting access to the block-encoding model. Since the LCU and sparse matrix access models are incomparable, we must provide lower bounds in both models.

\subsection{Lower bound: formal statement}
\label{sec:nfflb}

Given a Hamiltonian $H$, 
we wish to show that even if the initial state is supported
on a subspace associated with energies that are vanishingly small, e.g., $\Delta = o(\norm{H})$ or $\Delta=o(\lambda)$, time evolution of low-energy states requires $\Omega(t\norm{H})$ queries in the worst case. 
We will construct a Hamiltonian $H$ that can be efficiently block-encoded as $H/\lambda$ for $\lambda = \Theta(\norm{H})$. In our algorithm (Thm.~\ref{thm:mainalgorithm}), the query complexity can be written in terms of independent and asymptotically large parameters $t\Delta, \lambda/\Delta$, and $1/\epsilon$. Here, we only show a lower bound for the term that depends linearly on $t$, and we disregard the dependence on $1/\epsilon$; that is, we are assuming constant error or instances where $t \lambda =\Omega(\log(1/\epsilon))$. A lower bound in terms of $1/\epsilon$ in the opposite regime $t \lambda =\cO(\log(1/\epsilon))$ for general Hamiltonians is already given in Ref.~\cite{berry2014exponential}.

In order to vary the parameters $t\Delta$ and $\lambda/\Delta$ independently of each other, we introduce a Grover-like problem defined by two integer parameters, $K \ge 1$ and $N \ge 1$. To show the general no fast-forwarding theorem, we will evolve a Hamiltonian $H_{K,N}$ and examine its low-energy subspace such that $t\Delta = \Theta(K)$ and $\lambda/\Delta = \Theta(\sqrt{N})$. In the following section, to show the optimality of our quantum algorithm for gap-amplifiable Hamiltonians, we will use the same task but instead take $t\Delta = \Theta(K)$ and $\lambda/\Delta = \Theta(N)$.

Before stating the no fast-forwarding theorem, we briefly introduce the Grover-like problem, which we refer to as $\mathrm{PARITY}\circ\mathrm{OR}_{K,N}$. Informally, the problem consists of performing $K$ instances of quantum search, each over $N$ items. There is a single marked element in each instance. The output is the parity of a bit string of size $K$, where the $k^{\rm th}$ bit is determined by the $k^{\rm th}$ instance, and is 0 or 1 whether 
the marked element is in the first $N/2$ items or not, respectively.  As shown in Appendix~\ref{app:por} via the adversary method, this problem requires $\Omega(K\sqrt{N})$ queries to a Grover-like oracle. 
We defer the formal definition of this problem 
to Appendix~\ref{app:por} as it is not essential for the analysis in this section. However, we note that this problem is related to a quantum state preparation problem, which we now describe; see Lemma~\ref{cor:state-lb} for its proof.

\begin{lemma}[Quantum state preparation lower bound]
\label{lem:state-lb}
For all $k \in [K]$ define functions $f_k : [N] \to \{0,1\}$ such that $f_k(x_k) = 1$ for some $x_k \in [N]$ and $f_k(x\neq x_k)=0$ otherwise. Then, at least $\Omega(K\sqrt{N})$ queries to an oracle $O_f:\ket{k}\ket{x} \to (-1)^{f_k(x)}\ket{k}\ket{x}$ acting on $\mathbb C^{K \times N}$ are required to prepare a quantum state
\begin{align}
    \ket{\phi} = \bigotimes_{k=1}^K \ket{\phi_k}
\end{align}
satisfying $|\!\bra{\phi_k}\ket{x_k}\!|^2 \ge \alpha^2$ for constant $\alpha \in (0, 1)$. Here, $\ket{\phi_k} \in \mathbb C^N$ and $\ket \phi \in \mathbb C^{N^K}$. 
\end{lemma}

Here, $x_k$ is the single marked element of the $k^{\rm th}$ instance of Grover search with $N$ items. Accordingly, preparation of $\ket \phi$ would allow us to solve $\cO(K)$ instances of Grover search with high probability (determined by $\alpha$). In this sense, the lower bound $\Omega(K \sqrt N)$ is not surprising, but the technical analysis of Appendix~\ref{app:por} is required to formally prove it. 

The above state preparation problem reduces to a Hamiltonian simulation problem, which we will use to prove lower bounds on time evolution. These apply to the LCU and sparse matrix models discussed in Sec.~\ref{sec:accessmodels}. The main result of this section follows.

\begin{theorem}[No fast-forwarding in the low-energy subspace of generic Hamiltonians, formal]
\label{thm:no-gen}
Let $K \ge 1$ and $N\ge 1$ be integers. Then, there exists a sequence of PSD Hamiltonians $\{H_{K,N}\}_{K,N}$, each acting on a $N^K$-dimensional Hilbert space $\cH_{K,N}:=\mathbb C^{N^K}$, with the following properties.
\begin{itemize}
    \item In the LCU model, each Hamiltonian can be expressed as $H_{K,N} = \sum_{l=0}^{L_K-1} \beta_l U_l$, where $\beta_l>0$ and $U_l$ is unitary. Also, 
    $\lambda_{K} = \sum_{l=0}^{L_K-1} \beta_l$ satisfies $\lambda_{K}=\Theta(K)$ and $L_K =\Theta(K)$.
    \item In the sparse matrix model, each Hamiltonian can be expressed a $d_K$-sparse matrix, with $d_K=\Theta(K)$. Also, $\lambda_{K}=d_K \|H_{K,N}\|_{\max}$ satisfies $\lambda_{K}=\Theta(K)$.
\end{itemize}
Assume access to the corresponding oracles for each model. Then, for each $H_{K,N}$, there exists $\Delta_{K,N}=\cO(K/\sqrt N)$, a low-energy state $\ket{\psi_{K,N}} \in \cS_{\Delta_{K,N}}$, and a time $t_{N} = \cO\left(\sqrt N\right)$, such that preparing a unitary $U$ satisfying $|\!\bra{\psi_{K,N}}\Pi U^\dagger \Pi e^{-itH_{K,N}}\ket{\psi_{K,N}}\!| \geq 2/3$ requires at least
\begin{align}
    \Omega(t_{N}\lambda_{K})
\end{align}
queries to either the $\mathrm{SELECT}$ and $\mathrm{PREPARE}$ oracles in the LCU model, or the $O_{H_{K,N}}$ oracle in the sparse matrix model. In addition,  $\Delta_{K,N}/\lambda_K =\cO(1/\sqrt N)$.
\end{theorem}

The result proves that only having standard methods of access to $H$ might prevent one from using the low-energy condition of the initial state to improve upon methods like QSP. 
Independently, it gives a no fast-forwarding result 
for Hamiltonian simulation in the LCU model, complementing
the result of Ref.~\cite{BAC07} that applies to the sparse matrix model. (In general, Hamiltonians in the LCU model are not sparse.)

Given the result of no fast-forwarding in the LCU model, we can already conclude that access to a block-encoding of a general Hamiltonian does not allow fast-forwarding of time evolution in the low-energy subspace either. Otherwise there would be a contradiction, since Eq.~\eqref{eq:LCUblec} details how to efficiently build the block-encoding in this model.

We  give the proof of Thm.~\ref{thm:no-gen}, first in the LCU model and then in the sparse matrix model, below.

\subsection{No fast-forwarding in the low-energy subspace for the LCU model}
\label{sec:lcu-nff}

Consider an $n$-qubit system in $\mathbb C^N$, where $N=2^n$. For some $x \in \{0,1\}^n$, let $\ket{x}$ be a marked quantum state and $\ket s:=\frac 1 {\sqrt N} \sum_{x'=0}^{N-1} \ket{x'}$ be the uniform superposition state over all computational basis states (bitstrings).
The Hamiltonian
\begin{align}
    H_x = (1+1/\sqrt{N})\one_N - \ketbra{x} - \ketbra{s}
\end{align}
solves one instance of quantum search as follows~\cite{nielsen2010quantum}: i) initially prepare $\ket s$, ii) let it evolve for time $t_N=\Theta(\sqrt N)$, and iii) perform a measurement to obtain $x$ with high probability. The Hamiltonian acts non-trivially in a subspace of dimension 2
spanned by $\{\ket x, \ket s\}$. It has two distinct eigenvalues $\{0, \frac{2}{\sqrt{N}}\}$ in that subspace and the eigenvalue $1+\frac{1}{\sqrt{N}}$ is of multiplicity $N-2$. It follows that $H_x \succeq 0$,
 $\|H_x\|=1 + \frac{1}{\sqrt{N}} \leq 2$ and the spectral gap is $\Delta = \frac{2}{\sqrt{N}}$.

As described above, we will consider $K$ independent copies of this search problem; choosing arbitrary $K$ and $N$ permits the low-energy and evolution time to be  set independently. The total Hamiltonian is then
\begin{align}
    H_{K,N} := \sum_{k=1}^K \one_N^{\otimes k-1} \otimes H_{x_k} \otimes \one_N^{\otimes K-k} \;,
\end{align}
where $x_k \in \{0,1\}^n$ for all $k \in[K]$, referring to a marked item over a set of $N$ items in the $k^{\rm th}$ instance of quantum search.

We briefly comment on the assumed LCU access for $H_{K,N}$.
Observe that each $H_{x_k}$ can be easily expressed as an LCU since $H_x$ can be expressed as 
\begin{align}
    H_x = \frac 1 {\sqrt N} \one_N + \left(\frac {\one_N-2 \ketbra x } 2 \right) + \left( \frac {\one_N-2\ketbra s}2 \right) \;.
\end{align}
Each term above is associated with a unitary; for example, the first term in brackets is a reflection over the state $\ket x$ which corresponds to the Grover oracle that applies $O_x \ket y=(-1)^{\delta_{xy}}\ket y$. Then, if we let ${\rm U}_1=\one_N$, ${\rm U}_2(x)=\one_N-2 \ketbra x$, ${\rm U}_3=\one_N-2 \ketbra s$, we have $H_x= \frac 1 {\sqrt N}{\rm U}_1 + \frac 1 2 {\rm U}_2(x) + \frac 1 2 {\rm U}_3$ and
\begin{align}
\label{eq:HKNLCU}
    H_{K,N}=\sum_{k=1}^K \one_N^{\otimes k-1} \otimes  \left(\frac 1 {\sqrt N}{\rm U}_1 + \frac 1 2 {\rm U}_2(x_k) + \frac 1 2{\rm U}_3 \right)\otimes \one_N^{\otimes K-k} \;,
\end{align}
which is itself an LCU. That is, $H_{K,N}=\sum_{l=0}^{L_K-1} \beta_l U_l$, where each $U_l$ uses ${\rm U}_1$, ${\rm U}_2(x_k)$, or ${\rm U}_3$ once, $L_K=3K$, $\beta_l =\cO(1)$ for all $l$,
and $\lambda_{K} =\sum_{l=0}^{L_K-1} \beta_l =\cO(K)$.
We then assume access to the ${\rm SELECT}(H_{K,N})$ and ${\rm PREPARE}(H_{K,N})$ oracles of Eq.~\eqref{eq:LCUaccess}, where 
the $\beta_l$'s and $U_l$'s are specified from Eq.~\eqref{eq:HKNLCU}. 

Importantly, the corresponding ${\rm SELECT}(H_{K,N})$ 
can be implemented with a single call to an oracle that implements $O_f \ket k \ket x = (-1)^{x_k} \ket k \ket x$.
This is the oracle $O_f$ used in Lemma~\ref{lem:state-lb}.
To see this, note that ${\rm U}_2(x_k) = \sum_x (-1)^{f_k(x)} \ketbra x$ and then the oracle can be expressed as 
$O_f = \sum_{k}  \ketbra k \otimes {\rm U}_2(x_k)$. The ${\rm SELECT}(H_{K,N})$ oracle implements 
each ${\rm U}_2(x_k)$ on a different subsystem conditional on $\ket k$. This can be achieved 
by conjugating one $O_f \otimes \one_N^{K-1}$ with corresponding controlled-SWAP operations, which do not use the oracle.

Each Grover instance can be solved within the low-energy subspace $[0, \Delta_{K,N} = \frac{2K}{\sqrt{N}}]$ of the Hamiltonian $H_{K,N}$. 
Let
\begin{align}
\ket{\psi_{K,N}}:= \ket{S} = \ket{s}^{\otimes K}
\end{align}
be now the equal superposition state in $\mathbb C^{N^K}$. Each $\ket s$ is confined to the low-energy subspace of each $H_{x_k}$, and hence $ \ket{S}$ is confined to the 
low-energy subspace of $H_{K,N}$.
\begin{align}
    \sket{\phi_{k,0}} &= \sqrt{\frac{1+\sqrt{\frac{1}{N}}}{2}}\ket{x_k} + \sqrt{\frac{1-\sqrt{\frac{1}{N}}}{2}}\ket{\perp_k} \; ,\\
    \sket{\phi_{k,1}} &= \frac{1}{\sqrt{2\left(1 + \frac{1}{\sqrt{N}-1}\right)}}\ket{x_k} - \sqrt{\frac{1+\sqrt{\frac{1}{N}}}{2}}\ket{\perp_k} \; ,
\end{align}
where we introduced quantum state
\begin{align}
    \ket{\perp_k} = \frac{1}{\sqrt{N-1}} \sum_{j\neq x_k} \ket{j} \;.
\end{align}
The low-energy eigenstates of $H_{K,N}$
are then tensor products of the eigenstates above. In particular, we decompose $\ket{S}$ over its eigenbasis as
\begin{align}
    \ket{S} = \bigotimes_{k=1}^K \sqrt{\frac{1+\frac{1}{\sqrt{N}}}{2}}\sket{\phi_{k,0}} - \sqrt{\frac{1 - \frac{1}{\sqrt{N}}}{2}}\sket{\phi_{k,1}}  \;.
\end{align}

To apply the state preparation lower bound of Lemma~\ref{lem:state-lb}, it suffices to consider the preparation of
\begin{align}
    \ket{\phi} = \bigotimes_{k=1}^K \ket{x_k} \;,
\end{align}
 which corresponds to $\alpha=1$ in Lemma~\ref{lem:state-lb}. We show that $\ket \phi$ can be prepared {\em exactly} via time evolution with $H_{K,N}$ from the initial state $\ket S$. 
For simplicity, we isolate the analysis to a single subsystem $H_{x_k}$ with a marked element $x_k$, since the subsystems do not interact:
\begin{align}
    e^{-itH_{x_k}}\ket{s} &= \sqrt{\frac{1+\frac{1}{\sqrt{N}}}{2}}\ket{\phi_{k,0}} - e^{-\frac{2it}{\sqrt{N}}}\sqrt{\frac{1 - \frac{1}{\sqrt{N}}}{2}}\ket{\phi_{k,1}}\;.
\end{align}
The probability of producing the marked element after time $t$ is given by
\begin{align}
    |\bra{x_k}e^{-iH_{x_k}t}\ket{s}|^2 &= \left|\frac{1+\sqrt{\frac{1}{N}}}{2} - e^{-2it/\sqrt{N}}\frac{1-\sqrt{\frac{1}{N}}}{2}\right|^2 \geq \sin \left(\frac{t}{\sqrt{N}}\right)^2 \;,
\end{align}
for $0 \le t \leq \frac{\pi}{2}\sqrt{N}$. Hence, at time $t_{N} = \frac{\pi}{2}\sqrt{N}$, we have $e^{-it_{N}H_{K,N}}\ket{\psi_{K,N}}=e^{-it_{N}H_{K,N}}\ket{S} = \ket{\psi}$. By Lemma~\ref{lem:state-lb}, this requires
$\Omega (K \sqrt N)$ queries to the oracle, and hence 
$\Omega (K \sqrt N)$ queries to ${\rm SELECT}(H_{K,N})$.
Written in terms of $\lambda_{K}$ and $t_{N}$, this lower bound is $\Omega (t_{N} \lambda_{K})$, which applies even when $\Delta_{K,N}/\lambda_{K}=\cO(1/\sqrt N)$ is asymptotically small. Thus, we obtain a lower bound of $\Omega(t\lambda)$ at time $t = \Theta(\lambda/\Delta)$
in general, when access to the LCU oracles is given.

\subsection{No fast-forwarding in the low-energy subspace for the sparse matrix model}
\label{sec:sparse-nff}

We now consider the sparse matrix access model, which is incomparable to the LCU model. 
We also consider a Hamiltonian that uses the $(\mathrm{PARITY}\circ\mathrm{OR})_{K,N}$ lower bound, but the Hamiltonian is made sparse following Ref.~\cite{childs2004spatial} on spatial quantum walks. Each $H_{x_k}$ is now defined over a $d$-dimensional  periodic lattice with $N$ vertices, where we choose $d$ to be a constant larger than 4. As shown in Ref.~\cite{childs2004spatial} (see their Eq. 96), choosing constant $\gamma >0$ and $c=\cO(1/\sqrt{N})$ yields a PSD $N \times N$ Hamiltonian
\begin{align}
\label{eq:dlatticeHamiltonian}
    H_{x_k} = \gamma (D-A) - \ketbra{x_k} + c
\end{align}
with diagonal matrix $D$, adjacency matrix $A$, and marked state $x_k \in \{0,1\}^n$. Moreover,
\begin{align}
\label{eq:sparse-evolve}
    |\bra{x_k}e^{-iH_{x_k}t}\ket{s}|^2 = \alpha^2 \sin\left( \frac{\alpha t}{\sqrt{N}}\right)^2 \;,
\end{align}
for some constant $\alpha>0$ and $\ket{s} = \frac{1}{\sqrt{N}}\sum_{x'=0}^{N-1} \ket{x'}$. To apply a similar argument as in Sec.~\ref{sec:lcu-nff}, we need to show that $\ket{s}$ is (almost) in the low-energy subspace of $H_{x_k}$, and that the relevant energies are   $\cO(1/\sqrt{N})$ when the Hamiltonian is normalized so that $\|H_{x_k}\|=\cO(1)$. Define the integral $I_{j,d}$ by
\begin{align}
    I_{j,d} = (2d)^{-j} \int_0^\infty dx \frac{x^{j-1}e^{-x}}{\Gamma(x)}[\mathcal{I}_0(x/d)]^d
\end{align}
for modified Bessel function of the first kind $\mathcal{I}_0$. For a lattice with dimension $d > 4$, the constant $\alpha$ in the success probability is given by $I_{1,d}/\sqrt{I_{2,d}}$; moreover, $3/4 < \alpha < 1$ for all choices of $d$. It was also shown in Ref.~\cite{childs2004spatial} that $H_{x_k}$ with
\begin{align}
    \gamma = I_{1,d}, \quad c = \frac{\alpha}{\sqrt{N}}\;,
\end{align}
has ground and first excited state energies
\begin{align}
    E_0 &= 0, \quad E_1 = \frac{2\alpha}{\sqrt{N}}\;,
\end{align}
while the spectral norm of $H_{x_k}$ is upper-bounded by a constant. Hence, $H_{x_k}\succeq 0$. The $N$ eigenstates $\sket{\phi_{k,j}}$  of $H_{x_k}$ satisfy, for all $j=0,\dots,N-1$,
\begin{align}
    \left| \langle s\sket{\phi_{k,j}}\right|^2 = \frac{1}{N}\frac{1}{(E_j-c)^2 F'(E_j-c)}\;,
\end{align}
where in the large-$N$ limit,
\begin{align}
    F'(E) = \frac{1}{NE^2} + \frac{I_{2,d}}{\gamma^2}\;.
\end{align}
Evaluating the support of $\ket s$ on the two eigenstates of lowest energy $E_0$ and $E_1$, we obtain
\begin{align}
\label{eq:phi-support}
    \left|\bra{s_k}\ket{\phi_{k,0}}\right|^2 = \left|\bra{s_k}\ket{\phi_{k,1}}\right|^2 = \frac{1}{N}\left(\frac{1}{N} + \frac{I_{2,d}}{\gamma^2}\frac{I_{1,d}^2}{I_{2,d}N}\right)^{-1} \rightarrow \frac{1}{2}
\end{align}
to leading order. Consequently, $\ket{s}$ mostly is supported in a subspace associated with eigenvalues at most $2\alpha/\sqrt{N}$ of $H_{x_k}$ and has vanishingly small support outside this subspace. Formally, we can introduce the low-energy state
\begin{align}
    \ket{\psi} := \frac{\Pi_\Delta \ket{s}}{\norm{\Pi_\Delta \ket{s}}}
\end{align}
for projector $\Pi_\Delta$ onto the low-energy subspace of the Hamiltonian defined by $\Delta = 2\alpha/\sqrt{N}$. Eq.~\eqref{eq:phi-support} implies that $\ket{\psi}$ satisfies $|\bra{\psi}\ket{s}| \geq 1-\delta$ for any (arbitrarily small) $\delta > 0$ that is constant with respect to $N$. Consequently, a lower bound on the query complexity for preparing a state of the form $e^{-itH_{x_k}}\ket{s}$ with constant fidelity also implies a lower bound on the query complexity for preparing the state $e^{-itH_{x_k}}\ket{\psi}$ with constant fidelity. We continue our analysis in terms of $\ket{s}$ and apply it to the low-energy state $\ket{\psi}$ at the end.

These properties of $H_{x_k}$ and $\ket \psi$ are similar to the case studied in Sec.~\ref{sec:lcu-nff}. The rest of the proof is similar to that of the LCU model. Time evolution of the initial state $\ket{S} = \ket{s}^{\otimes K}$ for some time $t_{N} = \Theta(\sqrt{N})$ under the $\cO(Kd)$-sparse Hamiltonian 
\begin{align}
    H_{K,N} = \sum_{k=1}^K \one_N^{\otimes k-1} \otimes H_{x_k} \otimes \one_N^{\otimes K-k} \;,
\end{align}
will produce a state $\ket \phi = \bigotimes_{k=1}^K \ket{\phi_k}$ that satisfies
\begin{align}
    |\bra{\phi_k}\ket{x_k}| ^ 2 = \alpha^2 \;,
\end{align}
for the constant $\alpha$ in Eq.~\eqref{eq:sparse-evolve}. Since $d=\cO(1)$, note that the sparsity of $H_{K,N}$ is $d_K=\cO(K)$ and also $\|H_{K,N}\|_{\max}=\cO(1)$, since $\|H_{x_k}\|_{\max}=\cO(1)$
for Eq.~\eqref{eq:dlatticeHamiltonian}. Then, $\lambda_K =d_K \|H_{K,N}\|_{\max}=\cO(K)$.
Furthermore, the oracle $O_{H_{K,N}}$ that provides access 
to the nonzero matrix elements of $H_{K,N}$ and their positions
can be constructed with a single call to an oracle that implements $O_f \ket k \ket x= (-1)^{x_k} \ket k \ket x$.
This is similar to the oracle in Lemma~\ref{lem:state-lb},
which now implies a lower bound $\Omega(K \sqrt N)$ for the number of uses of $O_{H_{K,N}}$. To see this, note that each entry of $H_{K,N}$ can be labeled by $(k,x)$, and the positions of the nonzero entries are readily obtained from those of $H_{x_k}$, which are known due to Eq.~\eqref{eq:dlatticeHamiltonian}.
In terms of $\lambda_{K}$  and $t_{N}$, the lower bound is $\Omega(t_{N}\lambda_{K})$. Since the state $\ket{s}$ has vanishingly small support outside the low-energy subspace $\cS_{\Delta}$, the state $\ket S$ similarly has vanishingly small support outside the low-energy subspace $\cS_{\Delta_{K,N}}$, for $\Delta_{K,N} =\cO(K/\sqrt N)$.
The lower bound also holds for the state $\ket{\psi_{K,N}}:=\Pi_{\Delta_{K,N}} \ket S/\|\Pi_{\Delta_{K,N}} \ket S\|$, which is entirely contained within the low-energy subspace, and obtained by projecting each $\ket s$
into their low-energy subspaces.

Hence, we proved a lower bound $\Omega(t\lambda)$ at time $t=\Theta(\lambda/\Delta)$ in general, when access oracle access to a sparse matrix $H$ is given. Note that our result takes constant sparsity $d$; a more careful analysis would be needed to obtain a lower bound in terms of the sparsity of $H$.

\section{Query lower bounds for simulating time evolution of low-energy states of gap-amplifiable Hamiltonians}
\label{sec:nffsga}

We  show lower bounds on the query complexity of simulating low-energy states of gap-amplifiable Hamiltonians $H=\lambda A^\dagger A$, given some form of oracle access to $A$. 
As shown by Thm.~\ref{thm:mainalgorithm} and Sec.~\ref{sec:nff}, this proves to be more powerful than having access to $H$ alone.

Lemma~\ref{lem:hsgablec} and Lemma~\ref{lem:lbhsgablec} below show an equivalence between having access to a gap-amplifiable Hamiltonian $H = \lambda A^\dagger A$ and having access to a block-encoding of $H'$ with low-energy subspace corresponding to eigenvalues in $[-1, -1+\Delta/\lambda]$.
For this reason, we will perform the analysis for  gap-amplifiable Hamiltonians here. 
For a state $\ket{\psi}$ with energy at most $\Delta$ of such $H$, we summarize all possible asymptotic regimes and results:
\begin{enumerate}
    \item \emph{Time-dominated regime}: $\log (1/\epsilon) = o(t\Delta)$, producing an upper bound of $\cO(t\sqrt{\lambda\Delta})$. In this section, we will show matching query lower bounds in two different oracle models: the LCU model and the sparse matrix model. In the following section, we will show matching gate complexity lower bounds for this case.
    \item \emph{Intermediate regime}: $\log (1/\epsilon) = o(t\lambda)$ and $t\Delta = o(\log (1/\epsilon))$, producing an upper bound of $\cO(\sqrt{t\lambda \log 1/\epsilon})$. In this section, we will show a query lower bound in the LCU model of $\Omega(\sqrt{t\lambda})$.
    \item \emph{Error-dominated regime}: $t\lambda = o(\log (1/\epsilon))$, producing an upper bound of $\cO(\log (1/\epsilon))$. A matching lower bound (up to log log factors) is known~\cite{berry2014exponential}, based on parity.
\end{enumerate}

Hence, we will obtain query lower bounds that show our algorithm is optimal with respect to $t$, $\Delta$, and $\lambda$; we do not prove optimality with respect to $1/\epsilon$ in the intermediate scaling regime but note that the algorithm has a mild sublogarithmic dependence on error.

\subsection{Access models}
\label{sec:accessGAcase}

Similar to Sec.~\ref{sec:accessmodels}, we assume that access to $M\times N$ matrix $A$ is given through the LCU or sparse matrix model. 
Since $A$ is not necessarily a square matrix, some 
small modifications are needed. 

The matrix $A$ can be expressed as $A= (\sum_{l=0}^{L-1} \beta_l U_l) \Pi$,
where each $\beta_l>0$ and each $U_l$ is unitary acting on $\mathbb C^M$ for all $l$, and $\Pi$ is the projector onto $\mathbb C^N$.
The LCU model then assumes access to a SELECT and a PREPARE oracle defined by 
\begin{align}
\label{eq:LCUaccessA}
    \mathrm{SELECT}(A) = \sum_{l=0}^{L-1} \ketbra l \otimes U_l, \quad 
    \mathrm{PREPARE}(A) : \ket{0} \to \sum_{l=0}^{L-1} \sqrt{\frac{\beta_l}{\beta}}\ket{l}\;,
\end{align}
where $\beta = \sum_{l=0}^{L-1} \beta_l$. 
As before, we can assume $L=2^l$, so that
$\ket 0$ above refers to the $l$-qubit state $\ket 0^{\otimes l}$.

The sparse matrix model  uses an oracle $O_A$ that allows to perform the transformations
\begin{align}
    \ket{j}\ket{k}\ket{z} \rightarrow \ket{j}\ket{k}\ket{z \oplus A_{jk}}, \quad  \ket{j}\ket{\ell} \rightarrow \ket{j}\ket{\nu(j,\ell)}\;.
\end{align}
Here, $j\in [M]$ and $k \in [N]$ label a row and column of $d$-sparse $A$, and $A_{jk}$ is the corresponding matrix entry (specified in binary form). Also, $\ell \in [d]$ and $\nu: [M] \times [d] \rightarrow [N]$ computes the row index of the $\ell^{\rm th}$ nonzero entry of the $j^{\rm th}$ row in place.

A block-encoding $V$ of $A$ can be obtained with one use of ${\rm SELECT}(A)$, one use of $ \mathrm{PREPARE}(A)$, and one use of  $\mathrm{PREPARE}(A)^\dagger$ following a similar procedure to the one given in Sec.~\ref{sec:accessmodels}. It can also be constructed using $O_A$ and other two-qubit gates (cf. Lemma 6 of Ref.~\cite{low2019qubitization} or Ref.~\cite{CKS17}). In both models, the block-encoding is a unitary acting on an enlarged space that contains $A/\alpha$ in one of its blocks. The factor $\alpha>0$ can be different for the LCU and sparse models. In the results below, $\alpha$ is defined appropriately for either model.

In Lemma~\ref{lem:hsgablec}, we showed that a constant number of queries to a block-encoding $T$ of a Hamiltonian $H'=(H-E)/\lambda$ with the low-energy subspace contained in $[-1, -1+\Delta/\lambda]$ provides access to a block-encoding $V$ of $A$ such that $H-F=2\lambda A^\dagger A$, where $E, F \in \mathbb{R}$ and $\lambda >0$. Since our algorithm uses the block-encoding of $A$, Lemma~\ref{lem:hsgablec} implies that our algorithm can be applied when access to a Hamiltonian is given through the block-encoding $T$. Here, we will show our lower bounds for gap-amplifiable Hamiltonians. For these lower bounds to apply when given access to $T$ instead of $V$, we need to prove the opposite direction compared to Lemma~\ref{lem:hsgablec}: we must show that a constant number of queries to $V$ suffice to prepare $T$.
This result follows from standard techniques. 

\begin{lemma}[Efficient block-encoding of $H'$]
\label{lem:lbhsgablec}
Let $V$ be the unitary that block-encodes an $M\times N$ matrix $A$ as in Eq.~\eqref{eq:blecofA} and let the Hamiltonian acting on $\cH$ be $H = \lambda A^\dagger A$, for some $\lambda > 0$. Let $\ket \psi \in \cS_{\Delta}$ be any state supported exclusively on the subspace associated with eigenvalues of $H$ at most $\Delta$, for $\lambda \ge \Delta > 0$. Then, a block-encoding $T$ of $H' = (H-\lambda)/\lambda$ can be implemented with a constant number of queries to controlled-$V$ and controlled-$V^\dagger$. In addition, $\ket{\psi}$ is supported in the subspace associated with eigenvalues of $H'$ in $[-1, -1+\Delta/\lambda]$.
\end{lemma}
\begin{proof}
The statement on the support of $\ket \psi$ immediately follows from the definition of $H'$.
We can alternatively write
\begin{align}
    H' = A^\dagger A - \one_N = \Pi V^\dagger \Pi' V \Pi \;.
\end{align}
Let $M=2^m$ and $N=2^n$. Above, the projector $\Pi'$ can be replaced by 
\begin{align}
    P'=\ketbra 0^{\otimes m-n} \otimes \one_N \; ;
\end{align}
see Sec.~\ref{sec:sga}. The projector $P'$ can be expressed as $(U_{P'}+ \one_M)/2$, where $U_{P'}:=2P'-\one_M$ is unitary. Then,
\begin{align}
      H' = \Pi V^\dagger \left( \frac{U_{P'}+ \one_M} 2 \right) V \Pi \;.
\end{align}
Standard techniques give a block encoding of $V^\dagger \left( \frac{U_{P'}+ \one_M} 2 \right) V$, which is a sum of two unitaries. For example, if 
\begin{align}
    \tilde V:= \ketbra + \otimes U_{P'} + \ketbra - \otimes \one_M \;,
\end{align}
where $\ket{\pm}=(\ket 0 \pm \ket 1)/\sqrt 2$, we obtain
\begin{align}
     H' = \Pi V^\dagger \tilde V V \Pi \;,
\end{align}
and hence
\begin{align}
    T=V^\dagger \tilde V V \;.
\end{align}
To construct it, we needed $V$ and $V^\dagger$ once, 
and $\cO(m)$ two-qubit gates to implement $U_{P'}$.
\end{proof}
We can now show lower bounds for simulating general gap-amplifiable Hamiltonians and conclude that the same lower bounds hold when given access to a block-encoding of $H'$ instead.

\subsection{Lower bounds: formal statements}

We first address the lower bounds in the time-dominated scaling regime. Our argument is once again based on the $\mathrm{PARITY}\circ\mathrm{OR}_{K,N}$ lower bound of $K$ partitions of independent Grover searches presented in Sec.~\ref{sec:nff}. We will construct a Hamiltonian such that the asymptotic parameters $t\Delta$ and $\lambda/\Delta$ are independently controlled as $t\Delta = \cO(K)$ and $\lambda/\Delta=\cO(N)$, allowing $t$ and $\Delta$ to be arbitrarily set with respect to $\lambda$. We present our lower bounds formally in terms of these parameters. 

\begin{theorem}[$t\sqrt{\lambda\Delta}$ lower bound on the query complexity of simulating low-energy states of gap-amplifiable Hamiltonians]
\label{thm:t-bound}
Let $K \ge 1$ and $N \ge 1$ be integers. Then, there exists a sequence of gap-amplifiable Hamiltonians 
$\{H_{K,N}=\lambda_K A^{\dagger}_{K,N} A^{}_{K,N}\}_{K,N}$,  each acting on an $N^K$-dimensional Hilbert space $\cH_{K,N}:=\mathbb C^{N^K}$, $\|A_{K,N}\|\le 1$, with the following properties.
\begin{itemize}
    \item In the LCU model, each term can be expressed as $A_{K,N}=(\sum_{l=0}^{L_K-1} \beta_l U_l) \Pi$, where $\beta_l>0$, $U_l$ is unitary, and $\Pi$ is a projector onto $\cH_{K,N}$. Also, $\sum_{l=0}^{L_K-1} \beta_l \le 1$, $L_K=\Theta(K)$, and $\lambda_K = \Theta(K)$.
    \item In the sparse matrix model, each term $A_{K,N}$ and $H_{K,N}$ can be expressed as a $d_K$-sparse matrix, with $d_K=\Theta(K)$. Also, $\lambda_K = d_K \|H_{K,N}\|_{\max}$ satisfies $\lambda_K = \Theta(K)$.
\end{itemize}
Assume access to the corresponding oracles for each model. Then, for each $H_{K,N}$, there exists 
$\Delta_{K,N}=\cO(K/N)$,
a low-energy state $\ket{\psi_{K,N}} \in \cS_{\Delta_{K,N}}$, and a time $t_N = \cO\left(N\right)$,
such that preparing a unitary $U$ satisfying $|\!\bra{\psi_{K,N}}\Pi U^\dagger \Pi e^{-it_N H_{K,N}}\ket{\psi_{K,N}}\!| \geq 2/3$   requires at least
\begin{align}
    \Omega\left(t_N\sqrt{\lambda_{K}\Delta_{K,N}}\right)
\end{align}
queries to each of either the $\mathrm{SELECT}(A_{K,N})$ and $\mathrm{PREPARE}(A_{K,N})$ oracles in the LCU model, or the $O_{A_{K,N}}$ oracles in the sparse matrix model.
In addition, $\Delta_{K,N}/\lambda_K=\cO(1/N)$.
\end{theorem}

In Sec.~\ref{sec:summary}, we report our lower bound in terms of access to a block-encoding of a Hamiltonian $H'$ with low-energy subspace $[-1, -1+\Delta/\lambda]$. That result is implied by Theorem~\ref{thm:t-bound}: the LCU model can provide stronger access to $A$ or $A_l$ compared to only providing access to their block-encodings, since these can be efficiently constructed from the LCU oracles. 
In addition, providing access to the block-encoding of $A$ is equivalent to providing access to the block-encoding of $H'$ (due to Lemmas~\ref{lem:hsgablec} and~\ref{lem:lbhsgablec}).

To show the lower bound in the intermediate regime where $\log (1/\epsilon) = o(t \lambda)$ and $t\Delta = o(\log (1/\epsilon))$, we instead use a version of the polynomial method applicable to trigonometric polynomials. This allows us to show that time-dependence in the LCU model cannot be improved from $\sqrt{t\lambda}$.

\begin{theorem}[$\sqrt{t\lambda}$ lower bound on the query complexity of simulating low-energy states of gap-amplifiable Hamiltonians]
\label{thm:sqrtt-bound}
Let $t > 0$, $\Delta > 0$, and $\epsilon > 0$ be such that $t\Delta = o(\log (1/\epsilon))$ and $\log (1/\epsilon )= o(t\lambda)$, where $\lambda=1$.
Let $\theta \in [0,\theta_M]$, where $\theta_M:= \arcsin \sqrt \Delta$. Then, there exists a
continuous family of gap-amplifiable Hamiltonians
$\{H_\theta = A^\dagger_\theta A^{\;}_\theta \}_\theta$,
each acting on a $2$-dimensional Hilbert space $\mathbb{C}^2$, $A_\theta=(A_\theta)^\dagger$, $\|A_\theta \| \le 1$, with the following properties. In the LCU model, $A_\theta=(U_\theta+(U_\theta)^\dagger)/2$ is a linear combination of two unitaries $U_\theta$ and $U^\dagger_\theta$. Assume access to these unitaries and their controlled versions.  Then,  for each $H_\theta$, there exists a low-energy state $\ket{\psi_\theta} \in \cS_{\Delta}$, such that preparing a unitary $U$ satisfying $|\!\bra{\psi}\Pi U^\dagger \Pi e^{-itH_{\theta}}\ket{\psi}\!| \geq 1-\epsilon$   requires at least
\begin{align}
    \Omega\left(\sqrt{t\lambda}\right)
\end{align}
uses of the unitaries. 
Similarly, if access to $A_\theta$ is given 
through its block-encoding, preparing $U$ requires $\Omega(\sqrt{t \lambda})$ queries to this block-encoding. 
\end{theorem}

\subsection{\texorpdfstring{$t\sqrt{\lambda\Delta}$}{Time-dominated regime} lower bound for the LCU model}
\label{sec:LCUboundGA}

We provide the proof of Thm.~\ref{thm:t-bound}, which concerns the time-dominated regime, for the LCU model.
The proof in this case is similar to that
in Sec.~\ref{sec:lcu-nff}, with the modification
that we use the square of the Hamiltonians $H_{x}$ to obtain a gap-amplifiable Hamiltonian.
This modifies that analysis slightly, as we discuss below.

We consider the gap-amplifiable Hamiltonian acting on $\cH_{K,N}=\mathbb C^{N^K}$
\begin{align}
    H_{K,N}  = \sum_{k=1}^K \one_N^{\otimes k-1} \otimes (H_{x_k})^2 \otimes \one_N^{\otimes K-k} \;,
\end{align}
where $H_x = (1+1/\sqrt N) \one_N - \ketbra x - \ketbra s$, $\ket s=\frac 1 {\sqrt N}\sum_{x'=0}^{N-1}\ket{x'}$.
Each of the $K$ bit strings $x_k \in \{0,1\}^n$ is associated with an instance of Grover search over $N$ items.
Since $\|H_{x_k}\|=\cO(1)$, the $H_{K,N}$ above is a $(\lambda_K,K)$-gap-amplifiable Hamiltonian with $\lambda_K = \cO(K)$. That is, using Lemma~\ref{lem:GAreduction}, the Hamiltonian can also be expressed as $H_{K,N} = \lambda_K A_{K,N}^\dagger A_{K,N}^{}$. The operators $A_{K,N}^{}$ of dimension $(KN) \times N$ can be constructed
from the individual operators $A_k^\dagger = A_k^{}=H_{x_k}/(1+1/\sqrt N)$ of dimension $N \times N$:
\begin{align}
    A_{K,N}:= \frac 1 {\sqrt K}\sum_{k=1}^K \ket k \otimes ( \one_N^{\otimes k-1} \otimes A_k \otimes \one_N^{\otimes K-k})\;.
\end{align}
This can also be expressed as
\begin{align}
   A_{K,N} = \tilde A_{K,N} \Pi \;,   
\end{align}
where $\tilde A_{K,N}:= \sum_{k=1}^K \ketbra k \otimes ( \one_N^{\otimes k-1} \otimes A_k \otimes \one_N^{\otimes K-k})$ is of dimension $(KN)\times (KN)$ and $\Pi$ is 
the projector onto $\cH_{K,N}$; explicitly, we can project onto $\frac 1 {\sqrt K} \sum_{k=1}^K \ket k$. 
Note that $\tilde A_{K,N}$ can be expressed as an LCU with $L_K=\Theta(K)$ unitaries that use ${\rm U}_1$, ${\rm U}_2(x_k)$, and ${\rm U}_3$ in Sec.~\ref{sec:lcu-nff}.

Importantly, the corresponding ${\rm SELECT}(A_{K,N})$
can be implemented with a single call to an oracle that implements $O_f \ket{k}\ket{x} = (-1)^{x_k}\ket{k}\ket{x}$.
This is the oracle $O_f$ of Lemma~\ref{lem:state-lb}. 
${\rm SELECT}(A_{K,N})$ requires applying ${\rm U}_2(x_k)$
on a different subsystem conditioned in $\ket k$. However,
like  in Sec.~\ref{sec:lcu-nff}, this can be achieved 
by using $O_f$ and conjugating it with corresponding controlled-SWAP gates.

The low-energy state $\ket{\psi_{K,N}}:=\ket S=\ket s^{\otimes K}$ is now exclusively supported on the subspace of eigenvalues $[0,\Delta_{K,N}=\frac {4 K}N]$ of $H_{K,N}$.
The eigenstates of $H_{K,N}$ are the same as those discussed in 
Sec.~\ref{sec:lcu-nff}, and to analyze the time evolution of $\ket S$  we consider a single system to obtain
\begin{align}
    |\bra{x_k}e^{-it (H_{x_k})^2}\ket{s}|^2 &= \left|\frac{1+\sqrt{\frac{1}{N}}}{2} - e^{-4it/N}\frac{1-\sqrt{\frac{1}{N}}}{2}\right|^2 \geq \sin\left( \frac{2t}{N}\right)^2\;.
\end{align}
At time $t_{N}=\frac \pi 4 N$, each $\ket s$ is evolved
exactly to $\ket {x_k}$. That is, if $\ket \phi = \bigotimes_{k=1}^K \ket{x_k}$, then $e^{-it_{N}H_{K,N}}\ket S= \ket \phi$.

By Lemma~\ref{lem:state-lb}, simulating this time evolution requires
$\Omega (K \sqrt N)$ queries to the oracle $O_f$, and hence 
$\Omega (K \sqrt N)$ queries to ${\rm SELECT}(A_{K,N})$.
Written in terms of $\lambda_{K}$, $\Delta_{K,N}$, and $t_{N}$, this lower bound is $\Omega (t_{N} \sqrt{\lambda_{K} \Delta_{K,N}})$, which applies even when $\Delta_{K,N}/\lambda_K=\cO(1/N)$ is asymptotically small. Then, we proved a lower bound of $\Omega(t\sqrt{\lambda \Delta})$ at time $t = \Theta(\lambda/\Delta)$ for gap-amplifiable Hamiltonians, when access to the LCU oracles for $A_{K,N}$ is given.

\subsection{\texorpdfstring{$t\sqrt{\lambda\Delta}$}{Time-dominated regime} lower bound for the sparse matrix model}

We provide the proof of Thm.~\ref{thm:t-bound}, which concerns the time-dominated regime, for the sparse matrix model.
We introduce a gap-amplifiable Hamiltonian $H_{K,N}$ acting on $\cH_{K,N}=\mathbb C^{N^K}$ built from an expander graph similar to a construction in Ref.~\cite{somma2013spectral}. Once again, time evolution under this Hamiltonian solves $(\mathrm{PARITY}\circ\mathrm{OR})_{K,N}$. The Hamiltonian is also a sum of gap-amplifiable Hamiltonians acting on $K$ different subsystems of dimension $N$. Each of these Hamiltonians can be expressed as $H_{x_k} = \sum_{l=1}^L A_{l,x_k}^\dagger A_{l,x_k}^{}=\lambda A^\dagger_k A^{}_k$ for constant $L$, $k \in [K]$, and $\lambda=\cO(1)$, and $\|A_{l,x_k}\|_{\max} \le 1$. 
Explicitly, the gap-amplifiable Hamiltonian under consideration is
\begin{align}
\label{eq:fullGAsparse}
    H_{K,N} = \sum_{k=1}^K   \one_N^{\otimes k-1} \otimes \left(\sum_{l=1}^{L} A^\dagger_{l,x_k} A^{}_{l,x_k} \right) \otimes \one_N^{\otimes K-k} = \lambda_K A^\dagger_{K,N} A^{}_{K,N} \;.
\end{align}
In particular, in this construction, if each $H_{x_k}$ has constant sparsity then $H_{K,N}$ is $d_K=\cO(K)$-sparse. Moreover, if $\lambda= \cO(1)$, then 
 $\lambda_K = \cO(K)$. 
 
 To show that $H_{K,N}$ produces the required state of Lemma~\ref{lem:state-lb}, we can largely rely on our proof for the LCU model. It suffices to analyze the individual gap-amplifiable Hamiltonians $H_x$ and show that the state $\ket{s}$ is mostly confined to the low-energy subspace corresponding to $\Delta = \cO(1/N)$. We also require that evolving this low-energy state for time $t_N=\cO(N)$ solves an instance of quantum search.

We introduce a suitable gap-amplifiable Hamiltonian with these properties. The marked element is $x$. For organizational purposes, we divide the bulk of the proof into three pieces: we introduce a Hamiltonian $H_x$ and show it is gap-amplifiable; we show that $\ket{s}$ is largely confined in the low-energy in $H_x$ and replace it with a similar state $\ket{\psi}$ that is fully confined in the low-energy subspace; finally, we show that time evolution $e^{-it_NH_x}\ket{\psi}$ produces the marked state with constant probability in time $t_N=\cO(N)$.

\begin{lemma}
\label{lem:Hxgapamplifiable}
Let $x \in \{0,1\}^n$ and $O_x:\ket y \rightarrow (-1)^{\delta_{x,y}} \ket y$ be Grover's oracle. Then, there exists a $d$-sparse $(\lambda,L)$-gap-amplifiable Hamiltonian $H_x=\sum_{l=0}^{L-1} A^\dagger_{l,x} A^{}_{l,x}$ acting on $\mathbb C^N$, such that $d=\cO(1)$, $\lambda=\cO(1)$, and $L=d+1=\cO(1)$. The operators $A^{}_{l,x}$ are sparse and of dimension $N \times N$, and satisfy $A^{}_{l,x}=A^{\dagger}_{l,x}=\Pi_{l,x}$, for some projectors $\Pi_{l,x}$.
Moreover, access to the nonzero entries of $A_{l,x}$
and their positions can be performed with one use of $O_x$. 
Finally, $H_x$ has lowest eigenvalue $0$ and the first nonzero eigenvalue is at least $1/(4N)$.
\end{lemma}

The proof is in Appendix A of Ref.~\cite{somma2013spectral} and also contained in Appendix~\ref{app:spectrumHtilde}.
To provide a similar proof as we did for the Hamiltonians in Sec.~\ref{sec:LCUboundGA}, we would like $\ket{s}$ to reside within a low-energy subspace of $H_x$ associated with $\Delta = \cO(1/N)$. Unfortunately, $\ket{s}$ includes vanishingly small support outside this low-energy subspace. We accordingly introduce a state $\ket{\psi}$ that has high overlap with $\ket{s}$ and is fully confined in $\cS_\Delta$.

\begin{lemma}
\label{lem:Hxspectrum}
Let $\Pi_\Delta$ be the projector onto the low-energy subspace of the Hamiltonian $H_x$ defined in Lemma~\ref{lem:Hxgapamplifiable}, and let $\Delta = \cO(1/N)$. The state $\ket{\psi}$ defined by
\begin{align}
    \ket{\psi} = \frac{\Pi_\Delta \ket{s}}{\norm{\Pi_\Delta \ket{s}}}
\end{align}
satisfies $|\bra{\psi}\ket{s}| \geq 1-\delta$ for any (arbitrarily small) constant $\delta > 0$. Moreover, for eigenstates $\ket{\phi_j}$ of $H_x$, for $j=0,\dots,N-1$, with ground state $\ket{\phi_0}$, the following holds:
\begin{align}
\label{eq:gs}
    \ket{\psi} = c_0 \ket{\phi_0} + \sum_{j\geq 1}c_j \ket{\phi_j}, \quad c_0 > \frac{1}{2}, \quad \ket{\phi_0} = \frac{1}{\sqrt{2}}\ket{x} + \frac{1}{\sqrt{2(N-1)}}\sum_{y\neq x}\ket{y}.
\end{align}
\end{lemma}

The proof of this lemma is also contained in Appendix~\ref{app:spectrumHtilde}.

We seek to show that time evolution of $\ket{\psi}$ under $H_x$ produces the marked state $\ket x$ with constant probability if the evolution time is $t_N=\cO(N)$. As shown in Eq.~\eqref{eq:gs}, the ground state $\ket {\phi_0}$ has constant overlap with both the marked state $\ket{x}$ and the state $\ket{\psi}$. Hence, we can obtain the marked state with high probability by projecting from $\ket \psi$ onto $\ket{\phi_0}$. We implement this measurement with random time evolution, as given by the following known result.

\begin{lemma}[Ground state projection by random time evolution, from Ref.~\cite{BKS09}]
\label{lem:randomtimemeasurement}
Let $H$ be a Hamiltonian with eigenstates $\ket{\phi_0},\ket{\phi_1},\ldots$ of eigenvalues $\lambda_0 \le \lambda_1 \le \ldots$. The time evolution channel $\mathcal{N}$ defined by drawing evolution times $T$ randomly from a distribution with characteristic function $\Phi$ satisfies
\begin{align}
    \|(\mathcal{N} - \mathcal{M})(\ketbra{s}{s})\|_1   \leq 2\sqrt{\sum_{j \geq 1}|\Phi(\lambda_j) c_0 c_j^*|^2} \leq \sup_j |\Phi(\lambda_j)|\;,
\end{align}
where $\|\cdot \|_1$ indicates trace norm, and
\begin{align}
    \mathcal{M}(\rho) = P_0 \rho P_0 + \mathcal{E}((1-P_0)\rho(1-P_0))
\end{align}
is the channel that performs a perfect measurement of the ground state. Here, $P_0 =\ketbra {\phi_0}$ is the projector onto the ground state and $\mathcal{E}$ is some quantum channel. 
\end{lemma}

Since for $j \ge 1$ we have $\lambda_{j} > 1/(4N)$ for $H_x$, choosing the distribution to be the uniform interval $[0, \alpha N]$ for constant $\alpha>0$ gives the characteristic function
\begin{align}
    \epsilon := \sup_j |\Phi(\lambda_j)| = \max_{1/(4N) \leq \lambda \leq d+1} \left|\frac{e^{i\lambda \alpha N}-1}{i \lambda \alpha N}\right| \leq \frac{8}{\alpha}\;.
\end{align}
The corresponding channel $\mathcal{M}$ outputs a state
\begin{align}
    \sigma = \mathcal{M}(\ketbra{\psi}{\psi}) = c_0^2 \ketbra{\phi_0}{\phi_0} + (1-c_0^2) \rho_\perp\;,
\end{align}
where $\bra{\psi}\rho_\perp\ket{\psi}=0$.  If the state $\tilde\sigma = \mathcal{N}(\ketbra{\psi}{\psi})$ is $\epsilon$-close to $\sigma$, then
\begin{align}
    |\bra{\phi_0}(\tilde\sigma-\sigma)\ket{\phi_0}| \leq \|\tilde\sigma-\sigma\| \leq \|\tilde\sigma-\sigma\|_1 \leq \epsilon\;.
\end{align}
Hence, by the triangle inequality
\begin{align}
    |\bra{\phi_0}\tilde\sigma\ket{\phi_0}| \geq c_0^2 - \epsilon \geq \frac{1}{4} - \epsilon\;.
\end{align}
Due to the definition of the ground state, evolving for random times of at most $\alpha N$ will yield the marked string $x$ with probability at least
\begin{align}
    |\bra{x}\mathcal{N}(\ketbra{\psi}{\psi})\ket{x}|^2 \geq \frac{1}{2}\left(\frac{1}{4}-\frac{8}{\alpha}\right)\;.
\end{align}
By choosing, for example $\alpha = 64$, we conclude that $t = \cO(N) = \cO(1/\Delta)$ is sufficient to solve quantum search with constant probability using the Hamiltonian $H_x$.

The full Hamiltonian over $K$ subsystems that solves $(\mathrm{PARITY}\circ\mathrm{OR})_{K,N}$ is the one in Eq.~\eqref{eq:fullGAsparse}. Each $H_{x_k}$
is given in Lemma~\ref{lem:Hxgapamplifiable}, that is,
\begin{align}
    H_{x_k} = \sum_{l=0}^{d} A_{l,x_k}^\dag A_{l,x_k}^{}=\lambda A_k^\dagger A_k^{} \;.
\end{align}
The sparsity of $H_{x_k}$ is $d = \cO(1)$ and
\begin{align}
    \lambda= d+1, \quad A_k := \frac{1}{\sqrt{d+1}} \left(\sum_{l=0}^{d} \ket{l} \otimes A_{l,x_k}\right) 
    \;.
\end{align}
The operators $A_{K,N}$ are defined via
\begin{align}
    \lambda_K = K(d+1), \quad A_{K,N} =  \frac{1}{\sqrt K} \sum_{k=1}^{K} \ket{k} \otimes A_k \;,
\end{align}
which allow us to express the Hamiltonian as $H_{K,N}=\lambda_K A^{\dagger}_{K,N}A^{}_{K,N}$.
The operators $A_{K,N}$ are $\cO(K)$-sparse.
Each entry can be labeled by $(k,x)$, and constructing
the matrix oracle that provides access to these entries and their positions can be  done with one use of the oracle of Lemma~\ref{lem:state-lb} that implements $O_f \ket{k}\ket{x} = (-1)^{x_k}\ket{k}\ket{x}$.
The initial state is $\ket{\psi_{K,N}}:=\Pi_{\Delta_{K,N}}\ket S/ \|\Pi_{\Delta_{K,N}}\ket S\|$ for $\Delta_{K,N}=\cO(K/N)$ by Lemma~\ref{lem:Hxspectrum}.

Hence, using the query complexity lower bound 
$\Omega(K\sqrt{N})$ for this problem in Lemma~\ref{lem:state-lb}, 
it implies a lower bound $\Omega(t_N \sqrt{\lambda_K \Delta_{K,N}})$ on the number of queries to the matrix oracle for $A_{K,N}$. As in the LCU case, this applies even when $\Delta_{K,N}/\lambda_K=\cO(1/N)$ is asymptotically small.

\subsection{\texorpdfstring{$\sqrt{t\lambda}$}{Intermediate regime} lower bound for the LCU model}

We now address the intermediate regime of parameters satisfying $\log (1/\epsilon) = o(t\lambda)$ and $t\Delta = o(\log (1/\epsilon))$, and provide the proof to Thm.~\ref{thm:sqrtt-bound}. In this regime, the optimization over $\Gamma \in [\Delta, \lambda]$ in Lemma.~\ref{lem:GAsimulation} produces $\Gamma = \log(1/\epsilon)/t$, and the query complexity of our algorithm is
\begin{align}
    \cO\left(\sqrt{t\lambda \log ({1}/{\epsilon})}\right).
\end{align}
We will show a lower bound of $\Omega(\sqrt{t\lambda})$   applicable to the LCU model.

Specifically, to produce problem instances in the intermediate regime, we choose parameters such that $\lambda=1$, $1 > \Delta > 0$, and $1\ge \sqrt{\epsilon}\geq t\Delta \geq 8\epsilon >0$. For example, we can let $\Delta = \epsilon$ and $t = 1/\sqrt{\epsilon}$, and satisfy $\log (1/\epsilon) = o(1/\sqrt{\epsilon})$ and $\sqrt{\epsilon} = o(\log 1/\epsilon)$, which are the desired conditions since $t \lambda = 1/\sqrt \epsilon$ and $t \Delta= \sqrt{\epsilon}$ in this example. Simple modifications of our proofs below would allow one to consider other classes of instances in this regime. 
We will show that simulating certain instances in our setting is related to realizing a trigonometric polynomial approximation to $t\sin^2\theta$, for sufficiently large $t>0$ and sufficiently small $|\theta|$. The proof of Thm.~\ref{thm:sqrtt-bound} uses the following result on the degree for such a polynomial.

\begin{lemma}[Lower bound on the degree of trigonometric-polynomial approximations]
\label{lem:sqrtt-poly}
Let $t > 0$, $\epsilon > 0$, and $\Delta < 1$ satisfy $1 \geq \sqrt\epsilon \geq t\Delta \geq 8\epsilon$, and define $\theta_M:=\arcsin(\sqrt\Delta)$. Then, the lowest degree $K$ of a complex trigonometric polynomial $P_K(\theta) := \sum_{k=-K}^K a_k e^{ik\theta}$, $a_k \in \mathbb C$, satisfying
\begin{align}
    |P_K(\theta)| &\leq 1, \; \forall \; \theta \in (-\pi,\pi] \;,\\
    |e^{-it\sin^2\theta} - P_K(\theta)| &\leq \epsilon, \; \forall \; \theta \in [-\theta_M,\theta_M] \; ,
\end{align}
is $K \geq \frac{1}{\pi}\sqrt{2t}$.
\end{lemma}
\begin{proof}
 In the Taylor series expansion of the exponential
\begin{align}
    e^{-it\sin^2\theta} = 1 - it\sin^2\theta + r( \theta),
\end{align}
the remainder satisfies, for $|\theta| \leq \theta_M$,
\begin{align}
    |r(\theta)| \leq \sum_{l=2}^\infty \frac{(t\sin^2\theta)^l}{l!} \leq (t\Delta)^2 \leq \epsilon \;,
\end{align}
where the last inequality follows from the assumption $\sqrt\epsilon \geq t\Delta$. By the triangle inequality, the condition $|e^{-it\sin^2\theta} - P_K(\theta)| \leq \epsilon$ implies that
\begin{align}
    |1 - it\sin^2\theta - P_K(\theta)| \leq 2\epsilon, \; \forall \; \theta \in [-\theta_M,\theta_M].
\end{align}
Consider instead the trigonometric polynomial $Q_K(\theta) := i(1 - P_K(\theta))$.  Applying the triangle inequality, this polynomial is required to satisfy
\begin{align}
\label{eq:qconds}
|Q_K(\theta)| &\le 2, \; \forall \; \theta \in (-\pi,\pi] \;, \\
\label{eq:qconds2}
|t \sin^2 \theta - Q_K(\theta)| &\le 2 \epsilon, \; \forall \; \theta \in [-\theta_M,\theta_M] \;.
\end{align}
We will show a lower bound on $K$ for $Q_K(\theta)$, which implies the same lower bound on the degree of $P_K(\theta)$. 

Let $f'(\theta):= \frac{{\rm d} f(\theta)}{{\rm d}\theta}$. Bernstein's inequality  implies (cf.~\cite{borwein2012polynomials})
\begin{align}
    \sup_{\theta  \in (-\pi,\pi]} \left | Q''_K(\theta)  \right| \leq K^2 \sup_{\theta  \in (-\pi,\pi]} |Q_K(\theta)| \le 2 K^2 \;.
\end{align}
Then,
\begin{align}
  \left|Q_K(\theta_M)\right| &= \left|Q_K(0) + \int_0^{\theta_M} d \theta \left (Q_K'(0) + \int_0^\theta d \theta ' Q_K''(\theta') \right ) \right|\\ &\leq \left|Q_K(0)\right| + \theta_M \left|Q_K'(0)\right| + K^2\theta_M^2.
\end{align}
Due to our condition in Eq.~\eqref{eq:qconds2}, we can impose $|Q_K(0)| \leq 2\epsilon$. Also, $Q_K(\theta)$ 
can in principle be a combination of even and odd degree polynomials. However, we note that the even part $\frac 1 2 (Q_K(\theta)+Q_K(-\theta))$ readily satisfies Eqs.~\eqref{eq:qconds} and~\eqref{eq:qconds2} if $Q_K(\theta)$ does. Hence, we can assume that $Q_K(\theta)$
is an even polynomial and place a lower bound on its degree. This allow us to set $Q_K'(0) = 0$. Since $|Q_K(\theta_M)|$ must be within $2\epsilon$ of $t\sin^2\theta_M$, we obtain
\begin{align}
    t\sin^2 \theta_M - 2\epsilon &\leq 2\epsilon + K^2\theta_M^2.
\end{align}
The condition $\Delta \leq 1$ implies $\theta_M \leq \frac{\pi}{2} \sin \theta_M=\frac \pi 2 \sqrt{\Delta}$, and the condition $t\Delta \ge 8 \epsilon$ implies $t \Delta-4 \epsilon \ge t\Delta/2$, giving the desired result on the polynomial degree:
\begin{align}
    K \geq \frac{\sqrt{t\Delta - 4\epsilon}}{\theta_M} \geq \frac{2}{\pi}\sqrt{\frac{t\Delta}{2\Delta}} \geq \frac{1}{\pi}\sqrt{2t}.
\end{align}
\end{proof}
The lower bound $\Omega(\sqrt t)$ automatically generalizes to bound the degree of a trigonometric polynomial that approximates, for example, $\frac 1 2 e^{-it \sin^2\theta}$ for $\theta \in [-\theta_M, \theta_M]$.

We now describe how this trigonometric polynomial degree bound  can be adapted to lower bound the query complexity of low-energy simulation in the LCU model. To this end,
we need the following result, which is similar to Claim 5.3 in Ref.~\cite{mande2023tight}.

\begin{lemma} [Adapted from Ref.~\cite{mande2023tight}]
\label{lem:queryboundtrigonometric}
Let $U_\theta$ be a unitary acting on $\mathbb C^2$
of the form
\begin{align}
    U_\theta := \begin{pmatrix}
        e^{i \theta} & 0 \cr 0 & e^{i \theta_M} 
    \end{pmatrix}\;,
\end{align}
where $\theta_M$ is given and fixed, and $\theta \in (-\pi, \pi]$. 
Let $K>0$ be a positive integer. 
Consider a quantum circuit on $q \ge 2$ qubits that
has starting state $\ket 0^{\otimes q}$, uses an arbitrary number of $\theta$-independent two-qubit gates, uses $K$ applications of controlled-$U_\theta$ and its inverse in total, and performs no intermediate measurements. Then the
amplitudes of basis states before the final measurement are degree-$K$ trigonometric polynomials in $\theta$.
\end{lemma}
We provide the proof for completeness.
\begin{proof}
This is shown by induction.
Let $V_0,\ldots,V_K$
be the $\theta$-independent unitaries interleaved with the controlled-$U_\theta$'s and their inverses. Let $Q=2^q$ be the dimension of the Hilbert space associated with the $q$ qubits. 
For $K=0$, we have $V_0 \ket 0^{\otimes q}=\sum_{i=0}^{Q-1} V_0^{i,0}\ket i$, where $V_k^{i,j}$ is the $(i,j)$-th matrix entry of $V_k$. These entries do not depend on $\theta$, and then the claim is readily true for $K=0$.

Let $\ket{\psi_k}$ be the state before the aplication of the $(k+1)$-th application of controlled-$U_\theta$ (or its inverse). Assume it can be written in the computational basis as
\begin{align}
\label{eq:statestepk}
    \ket{\psi_{k}} = \sum_{j\in\{0,1\}} \sum_{b \in \{0,1\}} \sum_{w\in\{0,1\}^{q-2}} P^{j,b,w}(\theta) \ket{j}\ket{b}\ket{w} \;,
\end{align}
where the first register is the $\mathbb C^2$ space $U_\theta$ acts on, the second register is the controlled qubits, and the third register corresponds to the remaining $q-2$ qubits used as workspace. Each $P^{j,b,w}(\theta)$ is a trigonometric polynomial of degree at most $k$ by assumption. 

We now apply the next controlled-$U_\theta$. Note that
\begin{align}
    {\rm controlled-}U_\theta \ket j \ket b \ket w \left \{ \begin{matrix}
        e^{i \theta} \ket 0 \ket 1 \ket w &\ {\rm if} \ j=0,b=1 \;, \\
        e^{i \theta_M} \ket 1 \ket 1 \ket w & \ {\rm if} \ j=1,b=1  \; , \\
        \ket j \ket 0 \ket w & \ {\rm if} \ b= 0 \;.
    \end{matrix}\right.
\end{align}
Using this in Eq.~\eqref{eq:statestepk} increases the degree of each
$P_{j,b,w}(\theta)$ by at most 1. Applying $V_{k+1}$
does not increase the degree of the polynomial because the transformation is $\theta$-independent. A similar result applies when acting with the inverse of controlled-$U_{\theta}$. Hence, we conclude that the state output by this quantum circuit takes the form 
\begin{align}
    \ket{\psi_{K}} = \sum_{j\in\{0,1\}} \sum_{b \in \{0,1\}} \sum_{w\in\{0,1\}^{q-2}} P_K^{j,b,w}(\theta) \ket{j}\ket{b}\ket{w} \;,
\end{align}
where each $P_K^{j,b,w}(\theta)$ is a trigonometric polynomial of degree at most $K$.
\end{proof}

We are now ready to prove Thm.~\ref{thm:sqrtt-bound}.
Define the gap-amplifiable Hamiltonians
\begin{align}
    H_\theta:= \left( \frac {U_\theta + (U_{\theta})^\dagger} {2i}\right)^2 \;,
\end{align}
where $\theta_M$ is given, $0< \theta_M \le \pi/2$, and $0 \le \theta \le \theta_M$. Note that we can write $H_\theta= A^\dagger_\theta A^{\;}_\theta$ (i.e., $\lambda=1$), where
\begin{align}
    A^\dagger_\theta = A^{\;}_\theta = \begin{pmatrix}
        \sin \theta & 0 \cr 0 & \sin \theta_M 
    \end{pmatrix} \;.
\end{align}
Consider, for example, the initial state $\ket \psi =(\ket 0 + \ket 1)/\sqrt 2$. This is supported on the subspace of eigenvalues at most $\Delta = \sin^2 \theta_M$ of $H_\theta$.
Time evolution implies
\begin{align}
\label{eq:timeevolutiontrigonometric}
    \bra 0 e^{-it H_\theta}\ket \psi = \frac 1 2   e^{-it \sin^2 \theta}  \;.
\end{align}

In the LCU model we assume access to $U_\theta$, its inverse, and controlled versions.
If the operator $e^{-it H_\theta}$ acting on $\ket \psi$
is approximated using controlled-$U_\theta$ and its inverse $K$ times, Lemma~\ref{lem:queryboundtrigonometric} implies
that the approximation to Eq.~\eqref{eq:timeevolutiontrigonometric} is given by a trigonometric polynomial of degree at most $K$. 
If the approximation is done within error $\epsilon$,
and assuming the parameters of Lemma~\ref{lem:sqrtt-poly}, then 
this degree must satisfy $K=\Omega(\sqrt t)$.
Note that while $\theta_M$ was given and fixed (similarly for $\Delta$), 
the phase $\theta$ is unknown a priori and the simulation
method is required to produce an $\epsilon$-approximation for all $\theta \in [0,\theta_M]$.
Also, note that the SELECT and PREPARE oracles in the LCU model for this example, which allow to block-encode $A_\theta$, are simply obtained from a single controlled-$U_\theta$ and its inverse, and a Hadamard gate.
This gives the desired result on the query complexity in the LCU model stated in Thm.~\ref{thm:sqrtt-bound}.

\section{Gate complexity lower bounds for simulating time evolution of low-energy states of gap-amplifiable Hamiltonians}
\label{sec:gate}

We depart from the oracle model and instead show a gate complexity lower bound that addresses the time-dominated regime. In this regime, a state $\ket{\psi}$ has energy at most $\Delta$ in a $(\lambda, L)$-gap-amplifiable Hamiltonian, and the simulation parameters satisfy $\log (1/\epsilon) = o(t\Delta)$ and $t\Delta = o(t\lambda)$. We present a more formal version of Theorem~\ref{thm:gateinf}.

\begin{theorem}[Gate complexity lower bound for time evolution of low-energy states of gap-amplifiable Hamiltonians]
\label{thm:gate}
For any $n$-qubit quantum circuit with $G$ two-qubit gates such that $1 \leq n \leq G/n \leq 2^n$, there exists a sequence of 9-sparse and 4-local $(\lambda=4, L=1)$-gap-amplifiable Hamiltonians $\{H_G\}_G$ on  the space $\cH_G$ of $\Theta(G)$ qubits and a state $\ket{\psi_G} \in \cS_\Delta$ for $\Delta/\lambda = \cO(1/G^2)$, such that preparing a unitary $U$ satisfying
\begin{align}
\label{eq:thmgate}
    |\bra{\psi_G}\Pi U^\dagger \Pi e^{-itH_G}\ket{\psi_G}| \geq 2/3 \;,
\end{align}
for time $t = \cO(G^2)$ and some projector $\Pi$ onto $\cH_G$,  requires at least
\begin{align}
    \tilde \Omega\left(t\sqrt{\lambda \Delta}\right)
\end{align}
two-qubit gates, where $\tilde \Omega$ ignores logarithmic factors. The unitary $U$ may use an unlimited number of ancilla qubits and the gates may be non-local and come from a possibly infinite gate set.
\end{theorem}

To show Thm.~\ref{thm:gate}, we use ideas developed in Refs.~\cite{jordan2017fast} and~\cite{haah2018quantum}. We time-evolve a Gaussian wavepacket (i.e., an initial state with Gaussian-like amplitudes in the computational basis), whose width and initial momentum are chosen such that it is a low-energy state of the so-called Feynman-Kitaev clock Hamiltonian. This Hamiltonian can be shown to be gap-amplifiable. After a sufficient amount of time, the wavepacket spreads and propagates sufficiently far to encode the result of a computation that performs $G$ two-qubit gates. Using an argument based on counting distinct Boolean functions, we show that this implies that simulating the time evolution itself must have required at least $G$ two-qubit gates. In terms of the parameters of the problem instances, this coincides with the claimed lower bound of $\tilde\Omega(t\sqrt{\lambda\Delta})$.

We schematically describe the argument here before introducing modifications required to realize the necessary technical details. In standard circuit-to-Hamiltonian reductions, the Feynman-Kitaev clock Hamiltonian reads
\begin{align}
    H &= \sum_{x=1}^G \left(-U_x \otimes \ketbra{x}{x-1} - U_x^\dagger \otimes \ketbra{x-1}{x} + \one \otimes \ketbra{x}{x} + \one \otimes\ketbra{x-1}{x-1}\right) \;.
\end{align}
(We assume periodic boundary conditions were $\ket {x=G} \equiv \ket {x=0}$.)
It encodes a circuit $\cU=U_GU_{G-1}\cdots U_1$ consisting of $G$ two-qubit gates acting on $n$ qubits ($N=2^n$), and $H$ acts on the space $\mathbb C^{2^{n+G}}$, where we choose a clock encoding implemented over $G$ qubits.
Up to discretization, the dynamics of $H$ can be understood by considering a free particle on a line, since
\begin{align}
\label{eq:tight-bind}
   H_{\rm FP}:= W^\dagger H W &= \one_N \otimes \sum_{x=1}^G \left(-\ketbra{x}{x-1} - \ketbra{x-1}{x} + \ketbra{x}{x} + \ketbra{x-1}{x-1}\right)
\end{align}
for unitary
\begin{align}
    W &:= \sum_{x=1}^G U_x\cdots U_1 \otimes \ketbra{x}{x}\;.
\end{align}
Note that $H_{\rm FP}$ is a translationally invariant tight-binding model; it can be interpreted as the discretized version of $\tilde H_{\rm FP}=-\partial_x^2/2$, the free-particle Hamiltonian in the continuum (with mass $m=1$).
We will show that there exists an initial state $\ket{\psi_\Delta} \in \cS_\Delta$ within the low-energy subspace of $H$ given by $\Delta =\Theta( 1/G^2)$, such that evolution for time $t =\Theta( G^2)$ with constant error produces large support on the state that encodes the output of the circuit $\cU$, after tracing out the clock space spanned by $\{\ket x\}_x$. We will then show a lower bound requiring $\tilde \Omega(G)$ gates to simulate the evolution of $\ket {\psi_\Delta}$.

To show that this result is tight with our upper bound, we first rewrite $H$ in the form of a gap-amplifiable Hamiltonian to apply our algorithm. Using a binary encoding for $\ket x$ with $g=\log_2 G$ bits,
note that
\begin{align}
    H_{\rm FP}= \one_N \otimes (\one_G-X)^\dagger(\one_G-X) \;,
\end{align}
where $X$ is the $G\times G$ matrix that implements the cyclic permutation $X \ket x \rightarrow \ket{x-1}$.
(We assumed that $G$ is a power of two without loss of generality.) The lowest eigenvalue of $H_{\rm FP}$ is 0 and the largest is bounded by 2. We can define $A_{\rm FP}:=\one_N \otimes (\one_G-X)/ 2$ so that $H_{\rm FP}=4 A^\dagger_{\rm FP} A^{\!}_{\rm FP}$ is a $(\lambda=4,L=1)$-gap-amplifiable Hamiltonian.  Note that our quantum algorithm can simulate $e^{-it H_{\rm FP}}$ acting on $\ket{\psi_{\rm FP}}$, which is a state of supported in the subspace of energies at most $\Delta$, using $O(t\sqrt{\Delta} + 1/\sqrt{\Delta}) = \cO(G)$ queries to the block-encoding of $A_{\rm FP}$. This block-encoding can be constructed with ${\rm polylog}(G)$ two-qubit gates because of the simple structure of the operator $A^{\!}_{\rm FP}$.  For example, applying $X$ can be done using the quantum Fourier transform with $\cO(g^2)$ two-qubit gates.  Since $H$ is a similarity transformation of $H_{\rm FP}$ (Eq.~\eqref{eq:tight-bind}), then $H$ is also a $(\lambda=4,L=1)$-gap-amplifiable Hamiltonian: $H=4 A^\dagger A$, with $A:=W A_{\rm FP} W^\dagger$. To simulate $e^{-itH}$ on $\ket{\psi_\Delta}:=W \ket{\psi_{\rm FP}}$, we simply compute $e^{-itH_{\rm FP}}\ket{\psi_{\rm FP}}$ and conjugate by $W$, since $e^{-itH} = We^{-it H_{\rm FP}} W^\dagger$. The gate complexity to simulate $e^{-itH}\ket \psi$ remains $\tilde \cO(G)$, since the circuit to realize $W$ requires only an additional $\tilde \cO(G)$ two-qubit gates using standard techniques. This coincides with the lower bound, up to mild polylogarithmic corrections. The same result applies for other encodings of $\ket x$, like a unary encoding, as long as transforming from one encoding to the other can be done with $\tilde \cO(G)$ two-qubit gates. Note that QSP and other generic methods that do not use the low-energy condition would have produced an algorithm of gate complexity $\cO(G^2)$ or worse in this case, since $t=\Theta(G^2)$. 

Next, we present some technical details required for the proof of Thm.~\ref{thm:gate}.

\subsection{Clock Hamiltonian and low-energy states}

We now provide details of the clock Hamiltonian used in our construction, define the initial state, and show that this is confined to low-energy subspace. As before, we let $\cU=U_G U_{G-1} \cdots U_1$ be a quantum circuit constructed from a universal gate set of two-qubit gates (which includes one-qubit gates), where $\cU$ acts on $n$ qubits. To obtain a local Hamiltonian, we encode the clock register in unary as
\begin{align}
    \ket{0} = \ket{100\dots0}, \quad \ket{1} = \ket{010\dots 0},  \dots, \quad \ket{G'-1} = \ket{000\dots 1}
\end{align}
using $G'$ qubits, for some $G' \geq G$. This encoding will also allow us to show the result for a sparse Hamiltonian. We will specify the overhead $G'-G$ later; for now, it suffices to introduce integers $c_1, c_2 \geq 2$ constrained such that $G'=(c_1+c_2)G$. Define unitaries $V_x$ acting on the $n$ qubits that pad the original circuit with identity gates, so that
\begin{align}
    V_x = \begin{cases}
        \one_N &{\rm if} \ \ 1 \leq x \leq c_1 G \;,\\
        U_{x-c_1 G}&{\rm if} \ \ c_1 G+1 \leq x \leq (c_1 + 1)G \;,\\
        \one_N&{\rm if} \ \ (c_1+1)G + 1 \leq x \leq (c_1+c_2)G \;,
    \end{cases}
\end{align}
with periodic boundary conditions (i.e., $x=0$ is equivalent to $x=G'$). We consider the PSD clock Hamiltonian $H$ acting on Hilbert space $\cH_G:=\mathbb{C}^{2^{n+G'}}$ of $n+G'$ qubits, of periodic boundary conditions, defined by
\begin{align}
\label{eq:clock-h}
    H &:= \sum_{x=1}^{G'} \left(-V_x \otimes \ketbra{x}{x-1} - V_x^\dagger \otimes \ketbra{x-1}{x} + \one_N \otimes \ketbra{x}{x} + \one_N \otimes\ketbra{x-1}{x-1}\right)\;.
\end{align}
This Hamiltonian is at most 9-sparse because the $V^{}_x$'s and $V^\dagger_x$'s are 4-sparse in the worst case (two-qubit gates), and $H$ contains these unitaries in the off-diagonal blocks.  It is also 4-local because the elements $\ketbra{x}$ and  $\ketbra{x}{x-1}$ can be replaced by one-qubit and two-qubit interactions, respectively. 
The padding of identities $V_1= \dots= V_{c_1G}=\one_N$ allows us to initialize a wavepacket (state) over the clock register without interfering with the computation. That is, if the computation on the $n$ qubits starts in $\ket{0}^{\otimes n}$, we can choose {\em any} initial superposition over the clock register like
\begin{align}
\label{eq:wavepacket-def}
    \ket{\psi} = \ket{0}^{\otimes n} \otimes \sum_{x=0}^{c_1 G} \psi(x) \ket{x} \;,
\end{align}
without requiring any information from $\cU$. To additionally restrict $\ket \psi$ to be of low energy, we choose a Gaussian wavepacket centered at site $x_0 = (c_1/2) G$ and with initial rightwards momentum $p_0>0$, i.e.,
\begin{align}
\label{eq:wavepacketGaussian}
    \psi(x) &= \eta\exp\left[-\frac{(x - x_0)^2}{2\sigma^2} + i p_0 x\right], \quad \eta = \left(\sum_{x=0}^{c_1 G} \exp\left[-\frac{(x-x_0)^2}{\sigma^2}\right]\right)^{-1/2}\;.
\end{align}
Below, we will choose the wavepacket to have width $\sigma$ linear in $G$ and initial rightwards momentum $p_0$ linear in $1/G$. We wish to show that these choices imply that $\ket{\psi}$ is mostly confined to a low-energy subspace with cutoff $\Delta =\cO(1/G^2)$. (The choice of an initial momentum being nonzero is not necessary, but we use it to relate our results with Ref.~\cite{jordan2017fast} more directly.)

In the continuum limit, where $x \in \mathbb R$, these facts are apparent. The Hamiltonian becomes that of a free particle, i.e., $\tilde H_{\rm FP}=p^2/2m=-\partial_x^2/2$, with $m=1$. Its eigenfunctions are $e^{-ipx}$. The spatial width $\sigma \sim G$ of a Gaussian wavepacket implies, under a Fourier transform, a width of $\sim 1/G$ in momentum space. The initial momentum $p_0 \sim 1/G$ scales identically with respect to $G$.  Hence, the initial state of Gaussian amplitudes in the continuum is easily shown to be mostly supported in the subspace of energies at most $\Delta =\cO( 1/G^2)$ of $\tilde H_{\rm FP}$.

In practice, our initial state $\ket{\psi}$ is both defined on a discrete finite-dimensional system and has truncated tails compared to a true Gaussian (due to Eq.~\eqref{eq:wavepacket-def}). Nevertheless, we show that it has high overlap with a state $\ket{\psi_\Delta}$ that is fully confined to a low-energy subspace of $H$ specified by some $\Delta= \cO(1/G^2)$, by explicitly evaluating the energy of the wavepacket under the clock Hamiltonian and applying Markov's inequality.

\begin{lemma}[Low-energy state of clock Hamiltonian]
\label{lem:clockstate}
Let $H$ be the clock Hamiltonian for a $G$-gate quantum circuit $\cU=U_G \ldots U_1$ as defined in Eq.~\eqref{eq:clock-h}; it contains $c_1G$ initial identity gates and $(c_2-c_1-1)G$ final identity gates, for even positive integers $c_1, c_2$. Let $\ket{\psi}$ be the state defined in Eq.~\eqref{eq:wavepacket-def}, with amplitudes determined by Eqs.~\eqref{eq:wavepacketGaussian}, and centered at $x_0=(c_1/2)G$. Fix the width of the wavepacket to $\sigma = \hat\sigma x_0$ and the momentum to $p_0 = \hat p_0/x_0$ for constants $\hat\sigma>0$, $\hat p_0>0$. Then, for any (arbitrarily small) constant $\delta > 0$, there exists $\Delta = \cO(1/G^2)$ such that the state
\begin{align}
\label{eq:psiproj}
    \ket{\psi_\Delta} := \frac{\Pi_\Delta \ket{\psi}}{\norm{\Pi_\Delta \ket{\psi}}},
\end{align}
where $\Pi_\Delta$ is the projector onto the low-energy subspace specified by $\Delta$, satisfies
\begin{align}
    |\!\bra{\psi_\Delta}\ket{\psi}\!| \geq 1-\delta.
\end{align}
\end{lemma}
\begin{proof}
Direct computation gives
\begin{align}
    \bra{\psi}H\ket{\psi} &= \sum_{x=1}^{2x_0} |\psi(x) - \psi(x-1)|^2\\
    &= \eta^2 \sum_{x=1}^{2x_0} \exp\left[-\frac{\left(x-(x_0+1/2)\right)^2}{\sigma^2} - \frac{1}{4\sigma^2}\right] \nonumber\\
    &\quad \times \Bigg(\exp\left[\frac{x-x_0-1/2}{\sigma^2}\right] + \exp\left[-\frac{x-x_0-1/2}{\sigma^2}\right] - 2\cos(p_0)\Bigg)\;.
\end{align}
For $x_0 = c_1G/2 = \Theta(G)$, we choose wavepacket width $\sigma = \hat\sigma x_0 = \Theta(G)$ and momentum $p_0 = \hat p_0 / x_0 = \Theta(1/G)$ for constants $\hat \sigma, \hat p_0$. The terms of the sum are independent of $x$ up to $1/G$ corrections, i.e.,
\begin{align}
    \bra{\psi}H\ket{\psi} &= \eta^2 \sum_{x=1}^{2x_0} \left(\exp\left[-\frac{1}{\hat\sigma^2}\right] + \cO(G^{-1})\right) \left(\frac{\hat p_0^2}{x_0^2} + \frac{1}{\hat\sigma^2 x_0^2} + \cO(G^{-3})\right)\\
    &= \cO\left(\frac{1}{G^2}\right)\;,
\end{align}
where in the last step we used the fact that $\eta = \Theta(G^{-1/2})$. 
Let $H$ have eigenstates $\ket{\varphi_k}$ with corresponding eigenvalues $\gamma_k \geq 0$. By Markov's inequality, the support of $\ket{\psi}$ on eigenstates with energy at least $a/G^2$ is bounded by
\begin{align}
    \sum_{k \;:\; \gamma_k \geq a/G^2} |\bra{\varphi_k}\ket{\psi}|^2 = \cO\left(\frac{1}{a}\right) \;,
\end{align}
for any $a > 0 $. Consequently, for any $\delta > 0$, one can always choose $\Delta = \cO(1/G^2)$ such that the state $\ket{\psi_\Delta}$ defined in Eq.~\eqref{eq:psiproj} satisfies
\begin{align}
    |\!\bra{\psi_\Delta}\ket{\psi}\!| \geq 1-\delta.
\end{align}
\end{proof}

The state $\ket{\psi_\Delta}$ corresponds to $\ket{\psi_G}$ in Thm.~\ref{thm:gate}. Nevertheless,
we will continue our analysis to obtain a gate complexity lower bound on time-evolving $\ket{\psi}$ with constant fidelity; due to Lemma~\ref{lem:clockstate}, this implies an equivalent gate complexity lower bound on time-evolving $\ket{\psi_\Delta}$.

\subsection{Evaluating \texorpdfstring{$\cU$}{U} via time evolution of low-energy states}
\label{sec:circuittimeevolution}

We now show that time evolution of the low-energy state $\ket{\psi}$ or $\ket{\psi_\Delta}$ for time $t = \Theta(G^2)$ produces large support on the state $U_G\cdots U_1\ket{0}^{\otimes n}$ of the $n$ qubits.
We sketch this by first considering
the continuum limit, where $x \in \mathbb R$. The result in this case is readily available by the dynamics under the free particle Hamiltonian $\tilde H_{\rm FP}=-\partial_x^2/2$. A Gaussian wavepacket at time $t=0$, centered at $x_0$ and with momentum $p_0$,
\begin{align}
    \psi(x, 0) \propto \exp\left[-\frac{(x-x_0)^2}{2\sigma^2} + ip_0x\right] \;,
\end{align}
evolves under Schr\"odinger's equation $i\frac{\partial \psi(x,t)}{\partial t} = -\frac{1}{2}\frac{\partial^2 \psi(x,t)}{\partial x^2}$ to
\begin{align}
    \psi(x, t) \propto \exp\left[-\frac{(x-x_0(t))^2}{2\sigma(t)^2} + i\phi(x, t)\right]\;,
\end{align}
for phase $\phi(x, t) \in \mathbb R$ and
\begin{align}
    x_0(t) &= x_0 + {p_0 t}, \quad \sigma(t) = \sqrt{\sigma^2 + \left(\frac{t}{\sigma}\right)^2}\;.
\end{align}
That is, the center of the Gaussian wavepacket is spatially translated from $x_0$ to $x_0+{t p_0}$, and its width increases in time to $\sigma(t)$.
As discussed above, our choice of initial state $\ket{\psi}$ is, in the corresponding continuum limit, a Gaussian wavepacket with spatial width $\sigma\sim G$ centered at $x_0 \sim G$ and with initial momentum $p_0 \sim 1/G$. The clock Hamiltonian is a discretized version of  $\tilde H_{\rm FP}$. To perform a computation that evaluates $\cU=U_GU_{G-1}\cdots U_1$, we need $t$ to be sufficiently large to apply the next $G$ terms in the clock Hamiltonian; that is, we choose the parameters so that $x_0(t) - x_0 \sim G$. Additionally, if we choose $t p_0 \sim t/G$ to ``advance'' $G$ steps, we obtain a time $t \sim G^2$. This widens the wavepacket by a constant factor (i.e., $\sigma(G^2)/\sigma(0) >1$), allowing the evolved state to be supported in a sufficiently large subspace related to having applied the corresponding non-trivial $G$ gates. 

To formally prove that time evolution evaluates $U_G \ldots U_1\ket{0}^{\otimes n}$ with high probability, we can start from the analysis in the continuum limit, and then bound errors due to the discretization and truncation of the initial state $\ket{\psi}$. The proof is based on techniques established in Ref.~\cite{somma2016quantum} for analyzing the dynamics of one-dimensional systems. Due to the length of this technical analysis, we provide the proof in Appendix~\ref{app:gates}.

\begin{lemma}[Time evolution of initial Gaussian-like state under the clock Hamiltonian]
\label{lem:1d}
Let $H$ be the clock Hamiltonian for the $G$-gate unitary $\cU=U_G U_{G-1}\ldots U_1$ acting on $n$ qubits as defined in Eq.~\eqref{eq:clock-h}; it contains $c_1G$ initial identity gates and $(c_2-c_1-1)G$ final identity gates, for even integers $c_1, c_2$. Let $\ket{\psi}$ be the state defined in Eq.~\eqref{eq:wavepacket-def}, with amplitudes determined by Eqs.~\eqref{eq:wavepacketGaussian}, and centered at site $x_0=(c_1/2)G$. Fix the width of the wavepacket to $\sigma = \hat\sigma x_0$ and the momentum to $p_0 = \hat p_0/x_0$ for positive constants $\hat\sigma$, and $\hat p_0$. Consider the time-evolved state $e^{-itH}\ket{\psi}$ for time $t = \hat t x_0^2 = \cO(G^2)$ for constant $\hat t>0$. Then, for any $G\ge 1$, there exists a choice of constants $c_1$, $c_2$, $\hat \sigma$, $\hat p_0$, and $\hat t$, such that
\begin{align}
    \frac{1}{2}\norm{\cU\ketbra{0}{0}^{\otimes n} \cU^\dagger - \Pi_n e^{-itH} \ketbra{\psi}  e^{itH} \Pi_n}_1 \leq \frac{1}{4} \;,
\end{align}
where $\Pi_n$ is the projector onto the space of the $n$ qubits where the computation occurs, and $\frac{1}{2}\norm{\cdot}_1$ indicates the trace distance.
\end{lemma}

In particular, this implies
\begin{align}
    \bra 0^{\otimes n} \cU^\dagger \rho(t) \cU \ket 0^{\otimes n} \ge 3/4 \;,
\end{align}
for $\rho(t)=\tr_{\rm clock} (e^{-itH} \ketbra{\psi}  e^{itH})$, i.e., the state of the $n$ qubits after evolving with $H$ and tracing out the clock register.
By making this overlap to be larger than $2/3$,
we can allow other errors to be small constants and 
still satisfy Eq.~\eqref{eq:thmgate}.

\subsection{Gate complexity lower bound}

We summarize the analysis and results thus far. For an arbitrary circuit $\cU = U_GU_{G-1}\cdots U_1$ acting on $n$ qubits with $G$ two-qubit gates, we introduced a gap-amplifiable clock Hamiltonian $H=\lambda A^\dagger A$, with $\lambda = 4$. We defined a state $\ket{\psi_\Delta} \in \cS_\Delta$, for $\Delta = \cO(1/G^2)$, that has overlap $|\!\bra{\psi_\Delta}\ket{\psi}\!| \geq 1-\delta$ for arbitrarily small constant $\delta >0$, where $\ket \psi$
is the state with Gaussian-like amplitudes in Eqs.~\eqref{eq:wavepacket-def} and~\eqref{eq:wavepacketGaussian}. Consequently, a gate complexity lower bound on time-evolving $\ket{\psi}$ under $H$ with constant fidelity implies a lower bound on time-evolving the low-energy state $\ket{\psi_\Delta}$. We showed that evolution for time $t = \Theta(G^2)$ produces a state close to $\cU\ket{0}^{\otimes n}$. We must now show that this implies a lower bound on the gate complexity of time evolving low-energy states of gap-amplifiable Hamiltonians. Specifically, we will show that preparing $\cU\ket{0}^{\otimes n}$ with constant probability requires $\tilde \Omega(G) = \tilde \Omega(t\sqrt{\lambda\Delta})$ gates, implying the same lower bound for time-evolving $\ket{\psi_\Delta}$,
which is the state $\ket{\psi_G}$ in Thm.~\ref{thm:gate}.

We follow a similar argument to Ref.~\cite{haah2018quantum}. To this end, we constrain the circuits $\cU$ to those that compute distinct Boolean functions $f:\{0,1\}^n \to \{0,1\}$. Suppose briefly that $\cU$ acts on $n$ qubits and an unlimited number of ancilla qubits. We say that $\cU$ computes a Boolean function $f$ with high probability if measuring the first output qubit of $\cU\ket{x_1 x_2 \cdots x_n 0 \cdots 0}$ yields $f(x)$ with probability at least 2/3.
Restricting only to $\cU$ that act on a total of $n$ qubits, we count the number of distinct Boolean functions that $\cU$ can encode in terms of its depth $T$.

\begin{lemma}[Lower bound on distinct Boolean functions computed by $\cU$; Lemma 8 of Ref.~\cite{haah2018quantum})]
\label{lem:bool-low}
For any integers $n$ and $T$ such that $2 \leq n \leq T \leq 2^n$, the number of distinct Boolean functions $f:\{0,1\}^n\to\{0,1\}$ that can be computed by depth-$T$ quantum circuits on $n$ qubits that use geometrically local two-qubit gates on a constant-dimensional lattice from a finite get set is at least $2^{\tilde \Omega(Tn)}$.
\end{lemma}

Even if we allow an unlimited number of ancilla qubits, we have an upper bound on the number of Boolean functions that such $\cU$'s with $G$ two-qubit gates can compute with high probability.

\begin{lemma}[Upper bound on computing Boolean functions with high probability; Lemma 9 of Ref.~\cite{haah2018quantum}]
\label{lem:bool-up}
The number of Boolean functions $f:\{0,1\}^n \to \{0,1\}$ that can be computed with high probability by quantum circuits with unlimited ancilla qubits using $G$ two-qubit gates from any gate set is at most $2^{\tilde O(G \log n)}$.
\end{lemma}

We can now prove Thm.~\ref{thm:gate}. In Lemma~\ref{lem:1d}, we showed that time evolution of the state $\ket{\psi}$ under the clock Hamiltonian for a circuit $\cU$ with $G$ gates can evaluate $\cU\ket{0}^{\otimes n}$ in time $t=\cO(G^2)$ with high probability. Since $Tn \geq G$, Lemma~\ref{lem:bool-low} implies that the number of distinct Boolean functions
that can be computed by time-evolving with the Hamiltonians for corresponding $\cU$'s of $G$ gates, and with high probability, is at least $2^{\tilde \Omega(G)}$.
However, by Lemma~\ref{lem:bool-up}, a circuit that uses $G$ gates can compute at most $2^{\tilde O(G\log n)}$ Boolean functions with high probability. This implies that $\tilde \Omega(G)$ two-qubit gates are required to perform the time evolution with constant error in the worst case. Equivalently, replacing the state $\ket{\psi}$ with its low-energy projection $\ket{\psi_\Delta} \in \cS_\Delta$, 
we conclude that
\begin{align}
    \tilde \Omega\left(t\sqrt{\lambda \Delta}\right)
\end{align}
two-qubit gates are required to simulate the time evolution of $\ket{\psi_\Delta}$ with constant error and in the worst case. Note that in Thm.~\ref{thm:gate}, we report this gate complexity under the equivalent condition $|\!\bra{\psi_G}\Pi U^\dagger \Pi e^{-itH}\ket{\psi_G}\!| \geq 2/3$, where $\ket{\psi_G}\equiv \ket{\psi_\Delta}$.

\section{Conclusions}
\label{sec:conclusions}

We described a quantum algorithm for simulating the time evolution of states supported in a low-energy subspace of a Hamiltonian. Under certain conditions, the quantum algorithm provides improved runtimes over alternative Hamiltonian simulation methods that apply generally, like QSP. Specifically, the improvement arises when the Hamiltonian is ``gap-amplifiable'', i.e., it is PSD, can be expressed as $H=\lambda A^\dagger A$, and its ``square root'' $A$ can be efficiently accessed. These Hamiltonians include frustration-free Hamiltonians that appear ubiquitously in physics and quantum computing. In applications of these examples,
creating access to $A$ (e.g., in the form of an LCU or a block-encoding) may incur a similar cost as creating access to $H$. Hence, our improved 
bounds on query complexity are anticipated to translate
into improved bounds on gate complexities, provided that the initial state is mostly supported on the low-energy subspace of such Hamiltonians.
In general, we observe that it is always possible to shift a Hamiltonian by a constant to become PSD. While this shift might not yield any improvement in the worst-case scenario, it is conceivable that some moderate improvement could still be achieved by employing our algorithm in this case; we provided an example of this using a nearly frustration-free Hamiltonian.

We also proved several important lower bounds for the query complexity of this Hamiltonian simulation problem.
First, we showed that the low-energy condition of the initial state cannot be exploited to improve the runtime 
of Hamiltonian simulation methods that apply generally, like QSP, in the worst case (e.g., when we do not have efficient access to $A$). This answers a long-standing
question in Hamiltonian simulation; for example Ref.~\cite{csahinouglu2021hamiltonian} also required the Hamiltonians to be PSD and noted that it was not clear how to improve the runtime of Hamiltonian simulation in the low-energy subspace in general. Our results represent strengthened lower bounds compared to previous no-fast-forwarding theorems for generic Hamiltonians.

Second, we proved matching lower bounds in query complexity to show that our algorithm has optimal dependence on time and the size of the low-energy subspace. In the asymptotic (time-dominated) regime that is most interesting for applications, our quantum algorithm for gap-amplifiable Hamiltonians is a strict improvement over generic methods like QSP. Beyond the oracle model, we also proved a matching lower bound on the gate complexity in this regime. We note that in the asymptotic (intermediate) regime where time-dependence in sublinear, the cost of our algorithm also depends on the precision; we proved a lower bound in terms of time in this regime but did not prove a lower bound in in terms the precision, however,  we note that our upper bound only has sublogarithmic dependence on $1/\epsilon$.

Our quantum algorithm builds on a polynomial approximation to $e^{-it x^2}$ as a function of $x$. Since we are interested in a domain where $|x|\ll 1$, the degree of the polynomial is smaller than that needed for approximating the function for all $x \in [-1,1]$. For gap-amplifiable Hamiltonians, $x$ is associated with $A$ (or $A^\dagger$), and the degree of the polynomial with the query complexity. Related results and improvements are thus expected for implementing other functions of gap-amplifiable Hamiltonians
when the initial state is low-energy.

Since low-energy states appear ubiquitously in physics simulation, quantum chemistry, optimization and beyond, we expect our results to be broadly useful.

\section{Acknowledgements}
We thank Robin Kothari for discussions, feedback on the presentation, and for pointing out Ref.~\cite{haah2018quantum}, 
and Thomas E. O'Brien for reviewing an earlier draft of this article and comments. We also thank Ryan Babbush, William J. Huggins, Stephen Jordan, Robbie King, and Ramis Movassagh  at Google's Quantum AI, and Isaac Chuang and Aram Harrow at MIT, for discussions.

\bibliographystyle{plainnat}

\onecolumn\newpage
\appendix

\section{Proof of Lemma~\ref{lem:polyapprox}}
\label{app:polyapprox}
We prove Lemma~\ref{lem:polyapprox}, reproduced here for convenience. It is a simplified form of Corollary 66 of Ref.~\cite{gilyen2019quantum}, with a minor correction in the bounds of $\epsilon$ from $\epsilon\in(0,1/2B)$ to $\epsilon\in(0,3B/2)$~\cite{gilyencorrect}. Consistent with our notation throughout this work, we will ignore $\log\log$ factors associated with the Jacobi-Anger expansion in Lemma~\ref{lem:jacobi}.

\begin{customlemma}{2.1} [Piecewise polynomial approximations based on a local Taylor series]
Let $r \in (0, 1]$, $\delta \in (0,r]$, and let $f:[-r-\delta, r+\delta]$ be such that $f(x) = \sum_{k=0}^\infty a_k x^k$ for all $x \in [-r-\delta, r+\delta]$. Choose any $B>0$ such that $B \geq \sum_{k=0}^\infty  |a_k| (r+\delta)^k$. For $\epsilon \in (0, 3B/2)$, there is an efficiently computable polynomial $P$ of degree $\cO(\frac{1}{\delta}\log\frac{B}{\epsilon})$ such that
\begin{align}
    |f(x)-P(x)|_{[-r,r]} \leq \epsilon \quad \mathrm{and} \quad |P(x)|_{[-1, 1]} \leq \epsilon + |f(x)|_{[-r-\frac{\delta}{2}, r+\frac{\delta}{2}]}\;.
\end{align}
\end{customlemma}

\begin{proof}
Our proof largely follows that of Ref.~\cite{gilyen2019quantum}. Let $L(x) = \frac{x}{r+\delta}$ be the linear transformation from $[-r - \delta, r + \delta]$ to $[-1, 1]$. Let $g(y) = f(L^{-1}(y))$; this looks like $f(x)$ on $[-r, r]$ expanded onto $[-1, 1]$. Since $y = x/(r+\delta)$, we have $g(y) = \sum_{k=0}^{\infty} (a_k(r+\delta)^k) y^k = \sum_{k=0}^\infty b_k y^k$ for $b_k := a_k (r+\delta)^k$.
For any $K \ge 1$ and $x : |x|\le r$, we obtain
\begin{align}
    \left|g(y) - \sum_{k=0}^{K-1} b_k y^k\right| = \left|\sum_{k=K}^\infty b_k y^k\right| \leq \sum_{k=K}^\infty \left|b_k\left(1 - \frac{\delta}{2(r+\delta)}\right)^k\right| \;,
\end{align}
since we only need to bound the error for $|x| \leq r$, so that
\begin{align}
    |y| \leq \frac{r}{r+\delta} = 1 - \frac{\delta}{r+\delta} \leq 1 - \frac{\delta}{2(r+\delta)}\;.
\end{align}
In particular, we will take
\begin{align}
    K = \left\lceil \frac{2(r+\delta)}{\delta} \log \left(\frac{12B}{\epsilon}\right)\right\rceil,
\end{align}
where 
$\frac{12B}{\epsilon} > 1$, i.e., $\epsilon < 12B$. Continuing, we have
\begin{align}
    \sum_{k=K}^\infty \left|b_k\left(1 - \frac{\delta}{2(r+\delta)}\right)^k\right| & \leq \left(1 - \frac{\delta}{2(r+\delta)}\right)^K \sum_{k=K}^\infty |b_k|\\
    & \leq \exp\left[-K\frac{\delta}{2(r+\delta)}\right] \sum_{k=K}^\infty |a_k| (r+\delta)^k \\ & \le
   \exp\left[-K\frac{\delta}{2(r+\delta)}\right] \sum_{k=0}^\infty  |a_k| (r+\delta)^k
    \\
    & \le \exp\left[-K\frac{\delta}{2(r+\delta)}\right] B \;,
\end{align}
since $\left(1 - \frac{\delta}{2(r+\delta)}\right)^K \leq \exp\left[-K\frac{\delta}{2(r+\delta)}\right]$.
Finally, we plug in $K$ to obtain
\begin{align}
    \left|g(y) - \sum_{k=0}^{K-1} b_k y^k\right| \leq \exp\left[-\left\lceil\frac{2(r+\delta)}{\delta} \log \left(\frac{12B}{\epsilon}\right)\right\rceil\frac{\delta}{2(r+\delta)}\right] B \le \frac{\epsilon}{12}\;.
\end{align}
We now require the following result (adapted from Lemma 65 of Ref.~\cite{gilyen2019quantum}).
\begin{lemma}[Low-weight approximation by Fourier series]
\label{lem:fourier}
Let $\delta'\in (0, 1)$, $\epsilon \in (0, 1)$, and $g : \mathbb{R}\to\mathbb{C}$ be such that $\left|g(y)-\sum_{k=0}^{K-1} b_k y^k\right| \leq \epsilon/12$ for all $y \in [-1+\delta', 1-\delta']$. Then, there exists ${\bf c}:=(c_{-M},\ldots,c_M) \in \mathbb{C}^{2M+1}$ such that
\begin{align}
    \left|g(y) - \sum_{m=-M}^M c_m e^{i\pi m y/2}\right| \leq \epsilon/3
\end{align}
for all $y \in [-1+\delta', 1-\delta']$, where
\begin{align}
    M = \max\left(2\left\lceil\frac{1}{\delta'} \log \frac{12\|{\bf b}\|_1}{\epsilon}\right\rceil, 0\right) \;,
\end{align}
${\bf b}:=(b_0,\ldots,b_{K-1}) \in \mathbb C^{K}$,
and $\|{\bf c}\|_1 \leq \|{\bf b}\|_1$. Moreover ${\bf c}$ can be efficiently calculated on a classical computer in time $\mathrm{poly}(K, M, \log 1/\epsilon)$.
\end{lemma}

We apply this result directly to the function $g(y)=\sum_{k=0}^\infty b_k y^k = \sum_{k=0}^\infty a_k (r+\delta)^k y^k$ of Lemma~\ref{lem:polyapprox} to obtain
\begin{align}
    \left|\sum_{k=0}^{K-1}b_k y^k - \sum_{m=-M}^M c_m e^{i\pi m y/2}\right| \leq \frac{\epsilon}{3}\;,
\end{align}
where $y \in \left[-1+\frac{\delta}{2(r+\delta)},1-\frac{\delta}{2(r+\delta)}\right]$ and $\|{\bf c}\|_1  \leq \|{\bf b}\|_1 \leq \sum_{k=0}^\infty |b_k| \le B$. 
Let $\delta':=\frac{\delta}{2(r+\delta)}$. Then, Lemma~\ref{lem:fourier} gives
\begin{align}
    M = \max\left(2\left\lceil\frac{2(r+\delta)}{\delta}\log \frac{12\|{\bf b}\|_1}{\epsilon}\right\rceil, 0\right) = \cO\left(\frac{r}{\delta}\log\left( \frac{B}{\epsilon}\right)\right)\;.
\end{align}
The restriction on $\epsilon$ here is that $\epsilon/3 \in (0, 1)$. Changing variables back to $x$, the Fourier approximation of $f(x)=\sum_{k=0}^\infty a_k x^k$ can be expressed as
\begin{align}
    \tilde f(x) = \sum_{m=-M}^M c_m e^{i\pi m x/2(r+\delta)}\;.
\end{align}
This is an $\epsilon/3$-approximation to $f$ on $[-r-\delta/2, r+\delta/2]$, i.e., $|f(x)-\tilde f(x)|\le \epsilon/3$ in the domain.

To return to polynomials, we require the following two results (Lemmas 57 and 59 of Ref.~\cite{gilyen2019quantum}).
\begin{lemma}[Polynomial approximations of trigonometric functions]
\label{lem:jacobi}
Let $t \in \mathbb{R}\setminus \{0\}, \epsilon \in (0, 1/e)$, and
\begin{align}
    R = \left\lfloor\frac{1}{2}r\left(\frac{e|t|}{2}, \frac{5}{4}\epsilon\right)\right\rfloor\;,
\end{align}
where $r(u, \xi)$ is defined as the solution to
\begin{align}
    \xi = \left(\frac{u}{r}\right)^r : r \geq u\;,
\end{align}
which satisfies
\begin{align}
    r(u, \xi) = \Theta\left(u + \frac{\log(1/\xi)}{\log(e + \log(1/\xi)/u)}\right) = \cO\left(u + \log \frac{1}{\xi}\right)
\end{align}
for all $u > 0$ and $\xi \in (0, 1)$. Then, for all $x \in [-1, 1]$,
\begin{align}
    \left |\cos(tx) - J_0(t) + 2\sum_{k=1}^R (-1)^k J_{2k}(t) T_{2k}(x)\right |_{[-1, 1]} &\leq \epsilon \;,\\
    \left |\sin(tx) - 2 \sum_{k=0}^R (-1)^k J_{2k+1}(t) T_{2k+1}(x)\right |_{[-1, 1]} &\leq \epsilon \;,
\end{align}
where $J_m(t)$ are the Bessel functions  of the first kind and 
$T_m(x)$ are the Chebyshev polynomials of the first kind.
\end{lemma}

Hence, this result shows that $\cos(tx)$ and $\sin(tx)$
can be approximated by polynomials of even degree up to $2R$ and odd degree up to $2R+1$, respectively.

\begin{corollary}
There exist polynomials $P_\mathrm{sin}, P_\mathrm{cos}$ of degree
\begin{align}
    \cO\left(t + \log \frac{1}{\epsilon}\right)
\end{align}
such that for all $x \in [-1, 1]$,
\begin{align}
    \left |\cos(tx) - P_\mathrm{cos}(x)\right |_{[-1, 1]} \leq \epsilon, \quad \left |\sin(tx) - P_\mathrm{sin}(x)\right |_{[-1, 1]} \leq \epsilon
\end{align}
and
\begin{align}
    |P_\mathrm{cos}(x)|_{[-1, 1]} \leq 1, \quad |P_\mathrm{sin}(x)|_{[-1, 1]} \leq 1
\end{align}
for $\epsilon \in (0, 2/e)$.
\end{corollary}
\begin{proof}
We rescale the polynomial approximation by $\frac{1}{1+\epsilon}$ to ensure it is confined to magnitude 1. By the triangle inequality, for any function $h(x)$ that is $\epsilon$-approximated by $P(x)$,
\begin{align}
\left |\frac{1}{1+\epsilon}P(x) - h(x)\right |_{[-1, 1]} \leq \frac{1}{1+\epsilon} |P(x) - h(x)|_{[-1,1]} + \frac{\epsilon}{1+\epsilon}|h(x)|_{[-1,1]} \leq \frac{2\epsilon}{1+\epsilon} < 2\epsilon\;.
\end{align}
Hence, this rescaling only takes $\epsilon$ to $2\epsilon$. Accordingly, we rescale the bounds of $\epsilon$ to $(0, 2/e)$.
\end{proof}

We need to approximate $\sin(\pi m x/2(r+\delta))$ and $\cos(\pi m x/2(r+\delta))$, taking
\begin{align}
    t = \frac{\pi m}{2(r+\delta)} = \cO\left(\frac{1}{\delta}\log\frac{B}{\epsilon}\right)
\end{align}
and applying the above lemma. Hence, a normalized polynomial approximation $P_{\tilde f}(x)$ to $\tilde f(x)$ with error $\|P_{\tilde f}(x) - \tilde f(x)\|_{[-1,1]} \leq \xi$ has degree
\begin{align}
    \cO\left(\frac{1}{\delta}\log\frac{B}{\epsilon} + \log\frac{1}{\xi}\right)\;.
\end{align}
We choose $\xi = \epsilon/3B \leq 2/e$ to give a degree of $\cO\left(\frac{1}{\delta}\log\frac{B}{\epsilon}\right)$ and $\epsilon \leq 6B/e$. Summarizing the properties of $P_{\tilde f}$, we have
\begin{align}
    |P_{\tilde f}(x)\|_{[-1, 1]} \leq B, \quad\; |P_{\tilde f}(x) - \tilde f(x)|_{[-1,1]} \leq \frac{\epsilon}{3}, \quad\; |P_{\tilde f}(x) - f(x)|_{[-r-\delta/2, r+\delta/2]} \leq \frac{2\epsilon}{3}
\end{align}
where the first property follows from $\|c\|_1 \leq B$; the second property follows from the $\epsilon/3B$-approximation of each trigonometric function and the fact $\|c\|_1 \leq B$; and the third property adds the $\epsilon/3$ error of the original Fourier approximation.

Since $B$ may exceed 1, the first condition does not guarantee that $P_{\tilde f}$ is properly normalized to use as a QSVT polynomial. To obtain the final polynomial, we need a polynomial approximation of the rectangle function (Lemma 29 of Ref.~\cite{gilyen2019quantum}).

\begin{lemma}[Polynomial approximation of the rectangle function]
Let $\delta, \epsilon \in (0, 1/2)$ and $t \in [-1, 1]$. There exists an even polynomial $P$ of degree $\cO\left(\frac{1}{\delta}\log\frac{1}{\epsilon}\right)$ such that
\begin{align}
    |P(x)|_{[-1, 1] \setminus [-t-\delta, t+\delta]} \leq \epsilon, \qquad |P(x)-1|_{[-t+\delta, t-\delta]} \leq \epsilon, \qquad |P(x)|_{[-1, 1]} \leq 1\;.
\end{align}
\end{lemma}
We define polynomial $Q$ as the product of $P_{\tilde f}$ and a polynomial approximation $P_\mathrm{rect}$ of a rectangle function that satisfies
\begin{align}
    |P_\mathrm{rect}(x) - 1|_{[-r, r]} \leq \frac{\epsilon}{3B}, \quad |P_\mathrm{rect}(x)|_{[-1, 1] \setminus [-r-\delta/2, r+\delta/2]} \leq \frac{\epsilon}{3B}, \quad |P_\mathrm{rect}(x)|_{[-1, 1]} \leq 1\;.
\end{align}
The rectangle function must have $\epsilon/3B < 1/2$, i.e., $\epsilon < 3B/2$. This is the tightest condition on $\epsilon$, providing the condition in the final theorem statement. Since the rectangle function is $\epsilon/3B$-close to 1 on $[-r, r]$, the total error in $[-r, r]$ is $(\epsilon/3B + \epsilon/3B)(B) + \epsilon/3 = \epsilon$. The three terms come from the rectangle function error, the polynomial approximation of the Fourier expansion expansion, and the Fourier expansion respectively. Hence, $Q$ satisfies
\begin{align}
    |Q(x) - f(x)|_{[-r, r]} \leq \epsilon\;.
\end{align}
Due to the rectangle function approximating 0 outside $[-r-\delta/2, r+\delta/2]$, we similarly have
\begin{align}
    |Q(x)|_{[-1, 1] \setminus [-r-\delta/2, r+\delta/2]} \leq \epsilon\;.
\end{align}
Because the rectangle function is at most 1, and $P_{\tilde f}$ approximates the original function $f$ over $[-r-\delta/2, r+\delta/2]$, we have
\begin{align}
    |Q(x)|_{[-1, 1]} \leq \epsilon + |f(x)|_{[-r-\delta/2, r+\delta/2]}\;.
\end{align}
Finally, the degree of $Q$ is $\cO\left(\frac{1}{\delta}\log\frac{B}{\epsilon}\right)$ for $\epsilon \in (0, 3B/2)$, completing the Lemma.

\end{proof}

\newpage
\section{Quantum search task for lower bounds}
\label{app:por}

To show lower bounds for both gap-amplifiable Hamiltonians and generic Hamiltonians, we shall reduce to lower bounds for a particular search task. We define the following search problem; heuristically, it consists of $K$ instances of a typical Grover search, with one marked element per instance.
\begin{problem}[$K$-partition quantum search]
\label{def:ksearch}
We denote the $K$-partition quantum search over a length-$KN$ bitstring by $\mathrm{SEARCH}_{K,N}$, where $N=2^n$, and $K \ge 1$, $n \ge 1$ are integer. Each partition $k \in [K]$ consists of $N$ bits, where one bit is marked $f_k(x_k)=1$ for some $x_k \in \{0,\ldots,N-1\}$, and all other bits $x \neq x_k$ are marked $f_k(x) = 0$.  
Access to the  functions $f_1, \dots, f_K:\{0,1\}^N\rightarrow \{0,1\}$ is provided through a phase oracle $O_f$ that implements $\ket{k}\ket{x} \mapsto (-1)^{f_k(x)}\ket{k}\ket{x}$. The $K$-partition quantum search problem is to return the marked element in each partition, i.e., $x_1, \dots, x_K$, with probability at least $2/3$.
\end{problem}

Briefly, we comment on the connection between quantum search and Hamiltonian simulation, and
note that the use of multiple repetitions of quantum search is critical to obtaining a meaningful lower bound.
For a single instance of Grover's quantum search in dimension $N$, there exists an example of an $n$-qubit Hamiltonian with spectral gap $\Delta=\Theta (1/\sqrt N)$ such that a quantum walk for time $t \sim 1/\Delta$ can find the marked item; see Sec.~\ref{sec:lcu-nff} and Ref.~\cite{nielsen2010quantum}. In our case, we want to vary different parameters independently, and the only independent parameter in that example is $N$. Showing a lower bound for times independent of, and larger than, $1/\Delta$ will become important when we must separate time lower bounds from error lower bounds that also depend on $1/\Delta$. By including $K$ copies of that quantum system, each associated with an instance of quantum search of dimension $N$, the energy scale associated with the low-energy subspace increases with $K$ to include multiple energy levels, while the evolution time depends solely on $N$. A lower bound on $\mathrm{SEARCH}_{K,N}$ is then carried to a lower bound
on simulating such quantum systems. 

In practice, we consider the related decision problem $(\mathrm{PARITY}\circ\mathrm{OR})_{K,N}$, which can be solved by $\mathrm{SEARCH}_{K,N}$.

\begin{problem}[$K$-partition PARITY $\circ$ OR decision problem]
\label{def:KParityOr}
We denote the $K$-partition decision problem over a length-$KN$ bitstring by $(\mathrm{PARITY}\circ\mathrm{OR})_{K,N}$, where $N=2^n$, and $K \ge 1$, $n \ge 1$ are integer. Each partition $k \in [K]$ consists of $N$ bits, where one bit is marked $f_k(x_k)=1$ for some $x_k \in \{0,\ldots,N-1\}$, and all other bits $x \neq x_k$ are marked $f_k(x) = 0$. Access to the  functions $f_1, \dots, f_K: \{0,1\}^N \rightarrow \{0,1\}$ is provided through a phase oracle $O_f$ that implements $\ket{k}\ket{x} \mapsto (-1)^{f_k(x)}\ket{k}\ket{x}$. Let $g:\{0,1\}^N \rightarrow \{0,1\}$ be such that $g(x)$, $x \in \{0,1\}^N$, is the OR of the first $N/2$ bits of $x$. The $K$-partition {\rm PARITY $\circ$ OR} decision problem is to return $\mathrm{YES}$ if $\mathrm{PARITY}(g(x_1), \dots, g(x_K)) = 1$ and $\mathrm{NO}$ otherwise, with probability at least $2/3$.
\end{problem}

Known lower bounds for PARITY and for OR are sufficient to place a lower bound on $(\mathrm{PARITY}\circ\mathrm{OR})_{K,N}$ using the adversary method.

\begin{lemma}[Query complexity of $K$-partition PARITY $\circ$ OR]
\label{lem:por}
Deciding $(\mathrm{PARITY}\circ\mathrm{OR})_{K,N}$ requires $\Omega(K\sqrt{N})$ queries to $O_f$.
\end{lemma}
\begin{proof}
Consider the Boolean function $\mathrm{PARITY}(z)$, where $z \in \{0,1\}^K$, and define the Boolean function $g(x) : \{0,1\}^{N} \to \{0, 1\}$ to be the OR of the first $N/2$ bits of $x$. The function composition $h = \mathrm{PARITY} \circ g$ evaluates $\mathrm{PARITY}(g(x_1), \dots, g(x_K))$, which is precisely $(\mathrm{PARITY}\circ\mathrm{OR})_{K,N}$. 
The quantum adversary method provides lower bounds on the query complexity of PARITY and $g$ as $\mathrm{ADV}(\mathrm{PARITY}) = c_1K$ and $\mathrm{ADV}(g) = c_2 \sqrt{N}$, for constants $c_1>0$ and $c_2>0$.
References~\cite{hoyer2007negative,reichardt2014span} 
imply a lower bound on the query complexity of the composition of Boolean functions as the product of corresponding lower bounds. Hence, 
\begin{align}
    \mathrm{ADV}(\mathrm{PARITY} \circ g) = c K \sqrt{N}
\end{align}
for $c=c_1c_2$. Since our phase oracle is equivalent to a bit-flip oracle, the implication is that solving $(\mathrm{PARITY}\circ\mathrm{OR})_{K,N}$ requires $\Omega(K\sqrt{N})$ queries to $O_f$.
\end{proof}

Next, we use the lower bound on the query complexity of $(\mathrm{PARITY}\circ\mathrm{OR})_{K,N}$ to show a lower bound on preparing a particular state of $K$ quantum systems that emerges naturally from time evolution. Since the search within each partition is over $N$ elements, we rewrite the oracle to implement $f_1, \dots, f_K : \{0,1\}^n \rightarrow \{0,1\}$ for $N=2^n$. Although one-hot encoding each marked element was convenient when writing the problem in terms of Boolean functions above, this formulation is more useful in practice as it reduces the overhead in the second register presented to the oracle.

\begin{lemma}[Query complexity of state preparation]
\label{cor:state-lb}
Let $f_1,\ldots,f_K:\{0,1\}^n \rightarrow \{0,1\}$ be  functions, where $N=2^n$, $K \ge 1$, and $n \ge 1$ are integer. For each partition $k \in [K]$, assume that $f_k(x_k)=1$ for some marked item $x_k \in \{0,1\}^n$, and $f_k(x)=0$ otherwise. Then, 
at least $\Omega(K\sqrt{N})$ queries to the $K$-partition Grover oracle $O_f$ that implements $\ket k \ket x \mapsto (-1)^{f_k(x)}\ket k \ket x$  are required to prepare a $Kn$-qubit  state
\begin{align}
    \ket{\phi} = \bigotimes_{k=1}^K \ket{\phi_k} \;,
\end{align}
satisfying $|\!\bra{\phi_k}\ket{x_k}\!|^2 \ge \alpha^2$ for constant $\alpha \in (0, 1)$, where $\ket{x_k}$ is the $n$-qubit state corresponding to the marked item $x_k$ in the $k^{\rm th}$ partition. Here, $\ket{\phi_k} \in \mathbb C^N$ and $\ket \phi \in \mathbb C^{N^k}$.
\end{lemma}

\begin{proof}
Observe that with probability at least $2/3$ and with $\cO(K)$ queries to the oracle $O_f$, concentration implies that a constant number of copies of the state $\ket{\phi}$ are sufficient to successfully find $\beta K$ elements $S_{\beta K} = \{x_{i_1}, \dots, x_{i_{\beta K}}\}$ for $i_1, \dots, i_{\beta K} \in [K]$ and any constant $\beta$. This follows from measuring the copies of $\ket \phi$ in the computational basis and verifying each measurement outcome with $O_f$ for the identification of successful search outcomes. We denote the number of queries to $O_f$ required to obtain $S_{\beta K}$ by $Q$.

Consider the remaining unsuccessful $(1-\beta)K$ search problems. By Ref.~\cite{boyer1998tight}, it suffices to perform $(1-\beta)K$ independent Grover searches using a total of $\frac{\pi}{4}(1-\beta)K\sqrt{N}$ queries to the oracle to solve these instances with probability at least $2/3$. Hence, preparing $\beta K$ copies of $\ket{\phi}$ in addition to the Grover searches would allow us to solve Problem~\ref{def:KParityOr}.
By Lemma~\ref{lem:por}, there exists some constant $c$ such that at least $cK\sqrt{N}$ queries are required to decide $(\mathrm{PARITY}\circ\mathrm{OR})_{K,N}$, implying that
\begin{align}
    Q + \frac{\pi}{4}(1-\beta)K\sqrt{N} \geq cK\sqrt{N}\;.
\end{align}
Choosing constant $\beta$ such that $1 > \beta > 1-\frac{4c}{\pi}$ is always possible since $c>0$ (and we can assume $4c < \pi$). This implies that $Q = \Omega(K\sqrt{N})$ is necessary to satisfy the lower bound of Lemma~\ref{lem:por}; i.e., $\Omega(K\sqrt{N})$ queries to $O_f$ are required to prepare $\ket{\phi}$.
\end{proof}

Lemma~\ref{cor:state-lb} directly implies Lemma~\ref{lem:state-lb}, which
is the main result used to prove the lower bounds in Secs.~\ref{sec:nff} and~\ref{sec:nffsga}.

\newpage
\section{Proofs of Lemmas~\ref{lem:Hxgapamplifiable} and~\ref{lem:Hxspectrum}}
\label{app:spectrumHtilde}

We begin with the proof of Lemma~\ref{lem:Hxgapamplifiable}.
Following Appendix A of Ref.~\cite{somma2013spectral},
let $\tilde H_x$ be a Hamiltonian acting on $\mathbb C^N$ defined over an expander graph $G=(N, d, \eta)$ with $N=2^n$ vertices, constant degree $d = \cO(1)$, and second-largest eigenvalue $\eta d$ of the adjacency matrix, taking $\eta \leq 1/2$. Here, $x \in \{0,1\}^n$ refers to a marked item in Grover's search. The matrix elements of $\tilde H_x$ are given by
\begin{align}
    \bra{y}\tilde H_x\ket{z} = \begin{cases}
    \frac{1}{N-1} & y=z=x\\
    -\frac{1}{d\sqrt{N-1}} & \{y, z\} \in E \text{ and } y=x \text{ or } y=z\\
    -\frac{1}{d} & \{y, z\} \in E \text{ and } x \neq y \neq z \neq x\\
    1 & y=z\neq x\\
    0 & \text{otherwise}
    \end{cases}\;,
\end{align}
for edge set $E$. By construction, this Hamiltonian is $(d+1)$-sparse; we must convert it into a Hamiltonian that is gap-amplifiable. For each ordered edge $(y, z)$ with $y \leq z$, we observe that
\begin{align}
    \tilde H_x &= \sum_{(y, z) \in E} \left(c_y\ket{y} - c_z\ket{z}\right)\left(c_y\bra{y} - c_z\bra{z}\right)
\end{align}
for
\begin{align}
    c_v = \begin{cases}
    \frac{1}{\sqrt{d(N-1)}} & v=x\\
    \frac{1}{\sqrt{d}} & \text{else}
    \end{cases}\;.
\end{align}
To rewrite the Hamiltonian as a sum of projectors or PSD terms, we rescale the above terms to obtain
\begin{align}
    H_x &= \sum_{(y, z) \in E} \frac{\left(c_y\ket{y} - c_z\ket{z}\right)\left(c_y\bra{y} - c_z\bra{z}\right)}{c_y^2 + c_z^2} = \sum_{(y, z) \in E} \Pi_{(y, z)}\;.
\end{align}
Note that $(\Pi_{(y, z)})^2=\Pi_{(y, z)}$.
Hence, $H_x$ remains $(d+1)$-sparse; moreover, as shown in Ref.~\cite{somma2013spectral}, $H_x$ has a spectral gap of at least $\frac 1 {4(N-1)}$ and  $\|H_x\|_\mathrm{max} = \cO(\| H_x\|) = \cO(1)$.

In this decomposition over projectors $\Pi_{(y,z)}$, the number of terms in the frustration-free basis scales extensively with $N$. To ensure efficient access, we color the graph with $\cO(d)=\cO(1)$ colors, grouping together projectors that do not share a vertex using their orthogonality. This decomposition produces $H_x = \sum_{l=1}^L \Pi_{l,x}$ where the number of terms $L=\cO(1)$ is constant, and each $\Pi_{l,x}$ is a projector that depends on $x$. Note that this coloring also implies an upper bound on $\norm{H_x}$: since there are at most $d+1$ colors and hence $d+1$ projectors, the largest eigenvalue is at most $d+1 = \cO(1)$.
\qed

\vspace{.5cm}
Next we provide the proof of Lemma~\ref{lem:Hxspectrum}.

\vspace{.5cm}

Let $\ket{\perp} = \frac{1}{\sqrt{N-1}} \sum_{y\neq x} \ket{y}$. The ground state of $H_x$ has zero energy and is given by
\begin{align}
    \ket{\phi_0} = \frac{1}{\sqrt{2}}\ket{x} + \frac{1}{\sqrt{2}}\ket{\perp}\;.
\end{align}
Hence, $|\bra{\phi_0}\ket{s}| > 1/\sqrt{2}$. To show the uniform superposition $\ket{s}$ is low-energy in $H_x$, we evaluate powers of $\bra{s}H_x^k \ket{s}$. First, we write $H_x$ explicitly, using $(a,b)\in E$ to denote ordered edges $(a,b)$ in the edge set and $\mathcal{N}(x)$ to denote the edges with $x$ at one vertex:
\begin{align}
    H_x &= \sum_{y \in \mathcal{N}(x)} \left(\frac{1}{d} + \frac{1}{d(N-1)}\right)^{-1} \left(\frac{1}{\sqrt{d(N-1)}}\ket{x}-\frac{1}{\sqrt{d}}\ket{y}\right)\left(\frac{1}{\sqrt{d(N-1)}}\bra{x}-\frac{1}{\sqrt{d}}\bra{y}\right)\\
    &\qquad + \sum_{\substack{(y, z) \in E\\ y\neq x, z \neq x}} \left(\frac{1}{d} + \frac{1}{d}\right)^{-1}\left(\frac{1}{\sqrt{d}}\ket{y}-\frac{1}{\sqrt{d}}\ket{z}\right)\left(\frac{1}{\sqrt{d}}\bra{y}-\frac{1}{\sqrt{d}}\bra{z}\right)\\
    &= \frac{1}{N}\sum_{y \in \mathcal{N}(x)}\left(\ketbra{x}{x} + (N-1)\ketbra{y}{y} - \sqrt{N-1}(\ketbra{x}{y} + \ketbra{y}{x})\right) \\
    &\qquad + \frac{1}{2}\sum_{\substack{(y,z)\in E\\ y\neq x, z \neq x}} \left(\ketbra{y}{y} + \ketbra{z}{z} - \ketbra{y}{z}-\ketbra{z}{y}\right)\;.
\end{align}
In general, we find that states of the form $\alpha\ket{x}+\beta\ket{\perp}$ have energy $d(\alpha-\beta)^2/N$. We introduce the states
\begin{align}
    \ket{\psi_1} &= \sqrt{\frac{1+\sqrt{\frac{1}{N}}}{2}}\ket{x} + \sqrt{\frac{1-\sqrt{\frac{1}{N}}}{2}}\ket{\perp}, \quad 
    \ket{\psi_2} = \frac{1}{\sqrt{2\left(1 + \frac{1}{\sqrt{N}-1}\right)}}\ket{x} - \sqrt{\frac{1+\sqrt{\frac{1}{N}}}{2}}\ket{\perp}\;,
\end{align}
which produce the convenient decomposition
\begin{align}
    \ket{s} = \sqrt{\frac{1+\frac{1}{\sqrt{N}}}{2}}\ket{\psi_1} - \sqrt{\frac{1 - \frac{1}{\sqrt{N}}}{2}}\ket{\psi_2}
\end{align}
and have energies
\begin{align}
    \bra{\psi_1} H_x \ket{\psi_1} &= \cO(1/N^2)\\
    \bra{\psi_2} H_x \ket{\psi_2} &= \cO(1/N)\;.
\end{align}
By explicit computation, $|\bra{\psi_1}\ket{\phi_0}|^2 = 1 - \cO(1/N)$ and $|\bra{\psi_2}\ket{\phi_0}| = \cO(1/N)$. By Markov's inequality, for decomposition
\begin{align}
    \ket{\psi_2} = \sum_{j} d_j \ket{\phi_j}\;,
\end{align}
we have that
\begin{align}
    \sum_{j \;:\; \lambda_j \geq \frac{c}{N}} |d_j|^2 = \cO\left(\frac{1}{c}\right)
\end{align}
for any $c > 0$. Consequently, for any choice of (arbitrarily small) constant $\delta > 0$, there exists a choice of $\Delta=\cO(1/N)$ such that for projector $\Pi_\Delta$ onto the low-energy subspace, the state $\ket{\psi} = \Pi_\Delta \ket{s} / \norm{\Pi_\Delta \ket{\psi}}$ has overlap $|\bra{\psi}\ket{s}| \geq 1-\delta$. Moreover, due to the support of $\ket{s}$ on $\ket{\psi_1}$, we can choose $\delta$ such that $\ket{\psi}$ satisfies $|\bra{\psi}\ket{\phi_0}| > 1/\sqrt{2}-\delta > 1/2$.
\qed

\newpage
\section{Proof of Lemma~\ref{lem:1d}}
\label{app:gates}

For convenience, we reproduce the definitions and lemma statement before proving the result.

The clock Hamiltonian $H$ encodes a quantum circuit $\cU=U_GU_{G-1}\cdots U_1$ acting on $n$ qubits. We pad the circuit with identity gates, introducing unitaries $V_x$ given by
\begin{align}
    V_x = \begin{cases}
        \one_N &{\rm if} \ \ 1 \leq x \leq c_1 G\\
        U_{x-c_1 G} &{\rm if} \ \ c_1 G+1 \leq x \leq (c_1 + 1)G\\
        \one_N &{\rm if} \ \ (c_1+1)G + 1 \leq x \leq (c_1+c_2)G\;.
    \end{cases}
\end{align}
There are $G'=(c_1+c_2)G$ unitaries $V_x$.
The clock Hamiltonian is
\begin{align}
\label{eq:clockHamiltonianapp}
    H &= \sum_{x=1}^{G'} \left(-V_x \otimes \ketbra{x}{x-1} - V_x^\dagger \otimes \ketbra{x-1}{x} + \one_N \otimes \ketbra{x}{x} + \one_N \otimes\ketbra{x-1}{x-1}\right) \;,
\end{align}
where $G'=(c_1+c_2)G$. We assume periodic boundary conditions, i.e., $x=0$ is equivalent to $x=G'$. The initial state is supported only on basis states corresponding to initial padding with the identities $\one_N$, that is,
\begin{align}
\label{eq:gaussianwaveapp}
    \ket{\psi} = \ket{0}^{\otimes n} \otimes \sum_{x=0}^{c_1 G} \psi(x) \ket{x}\;.
\end{align}
The amplitudes are those of a Gaussian wavepacket centered at site $x_0 = (c_1/2) G$ and with initial rightwards momentum $p_0>0$, such that
\begin{align}
\label{eq:amplitudesapp}
    \psi(x) &= \eta\exp\left[-\frac{(x - x_0)^2}{2\sigma^2} + i p_0 x\right], \quad \eta = \left(\sum_{x=0}^{c_1 G} \exp\left[-\frac{(x-x_0)^2}{\sigma^2}\right]\right)^{-1/2}\;.
\end{align}
Without loss of generality, we choose to start the computation by $\cU$ in state $\ket{0}^{\otimes n}$. We seek to prove the following result from Sec.~\ref{sec:circuittimeevolution}.

\begin{customlemma}{6.3}[Time evolution of initial Gaussian-like state under the clock Hamiltonian]
Let $H$ be the clock Hamiltonian for the $G$-gate   unitary $\cU=U_G U_{G-1}\ldots U_1$ acting on $n$ qubits as defined in Eq.~\eqref{eq:clockHamiltonianapp}; it contains $c_1G$ initial identity gates and $(c_2-c_1-1)G$ final identity gates, for even integers $c_1, c_2$. Let $\ket{\psi}$ be the state defined in Eq.~\eqref{eq:gaussianwaveapp}, with amplitudes determined by Eqs.~\eqref{eq:amplitudesapp}, and centered at site $x_0=(c_1/2)G$. Fix the width of the wavepacket to $\sigma = \hat\sigma x_0$ and the momentum to $p_0 = \hat p_0/x_0$ for positive constants $\hat\sigma$, and $\hat p_0$. Consider the time-evolved state $e^{-itH}\ket{\psi}$ for time $t = \hat t x_0^2 = \cO(G^2)$ for constant $\hat t>0$. Then, for any $G\ge 1$, there exists a choice of constants $c_1$, $c_2$, $\hat \sigma$, $\hat p_0$, and $\hat t$, such that
\begin{align}
    \frac{1}{2}\norm{\cU\ketbra{0}{0}^{\otimes n} \cU^\dagger - \Pi_n e^{-itH} \ketbra{\psi}{\psi} e^{itH} \Pi_n}_1 \leq \frac{1}{4} \;,
\end{align}
where $\Pi_n$ is the projector onto the space of the $n$ qubits where the computation occurs, and $\frac{1}{2}\norm{\cdot}_1$ indicates the trace distance.
\end{customlemma}

\begin{proof}
    On the first register of $n$ qubits where the computation occurs, we introduce an initial state $\ket{h_0} := \ket{0}^{\otimes n}$ and subsequent states $\ket{h_1}$ through $\ket{h_{G'}}$ given by
\begin{align}
    \ket{h_i} = V_i V_{i-1}\cdots V_1 \ket{h_0}\;.
\end{align}
Here, $G'=(c_1+c_2)G$, for some even constants $c_1>0$ and $c_2>0$.
Due to its connection to a ``tight-binding model'', the clock Hamiltonian admits a convenient diagonalization through the Fourier transform.
In the subspace fixed by $\ket{h_0}$, the eigenstates are
\begin{align}
    \ket{\varphi_k} = \frac{1}{\sqrt{G'}}\sum_{j=0}^{G'-1} e^{i 2\pi jk/G'}\ket{h_j}\ket{j}
\end{align}
where $0 \le k \le G'-1$. The corresponding eigenvalues are $\gamma_k = 2(1-\cos(2\pi k /G'))$.
These satisfy $0 \le \gamma_k \le 2$, showing that $H$ is PSD. The other eigenspaces are degenerate and fixed by other choices of $\ket{h_0}$.

Since $c_1 G=2x_0$,
we can express the initial state $\ket{\psi(0)}:=\ket \psi$
of Eq.~\eqref{eq:gaussianwaveapp} in terms of the eigenstates as
\begin{align}
    \ket{\psi(0)} &= \frac{\eta}{\sqrt{G'}}\sum_{x=0}^{2x_0}\sum_{k=0}^{G'-1} \exp\left[-\frac{(x-x_0)^2}{2(\hat\sigma x_0)^2} + i \frac{\hat p_0}{x_0}x\right]  e^{-i 2\pi x k/G'} \ket{\varphi_k} \;.
\end{align}
Time evolution under $H$ for time $t$ produces
\begin{align}
    \ket{\psi(t)} &= \frac{\eta}{\sqrt{G'}}\sum_{x=0}^{2x_0}\sum_{k=0}^{G'-1} \exp\left[-\frac{(x-x_0)^2}{2(\hat\sigma x_0)^2} + i \frac{\hat p_0}{x_0}x\right] e^{-it\gamma_k} e^{-i 2\pi  x k/G'} \ket{\varphi_k}\\
    &= \frac{\eta}{G'}\sum_{x=0}^{2x_0}\sum_{k=0}^{G'-1}\sum_{j=0}^{G'-1} \exp\left[-\frac{(x-x_0)^2}{2(\hat\sigma x_0)^2} + i \frac{\hat p_0}{x_0}x\right] e^{-it2(1-\cos(2\pi k /G'))} e^{i2\pi (j-x) k/G'} \ket{h_j}\ket{j}\;.
\end{align}

We will ultimately use a state $\ket{\psi_\Delta}$ that is entirely supported in a low-energy subspace $\Delta$ of $H$; all eigenvalues $\gamma_k$ in this subspace will satisfy
\begin{align}
    \gamma_k \leq \Delta = \cO(\hat p_0^2/G^2) \;,
\end{align}
implying that all $k = \cO((c_1+c_2)\hat p_0)$ supported by $\ket{\psi_\Delta}$ are also bounded by a constant. In this low-energy subspace, the eigenvalues of $H$ can be approximated as
\begin{align}
\label{eq:appsmallk}
    \gamma_k &= \left(\frac{2\pi k}{(c_1+c_2)G}\right)^2 + \cO(G^{-4}) \;.
\end{align}
We choose evolution time $t=\hat t x_0^2$ for constant $\hat t$ and evaluate the support of the time-evolved state on a particular basis state $\ket{h_j}\ket j$.
We replace $\gamma_k$ in the time evolution with the small-$k$ approximation in Eq.~\eqref{eq:appsmallk} and keep track of the additional error introduced by the approximation. Due to the choice of simulation time $t =\Theta( G^2)$, the $\cO(G^{-4})$ error in the approximation of the eigenvalue produces $\cO(G^{-2})$ error in the exponent, i.e.,
\begin{align}
 \bra{h_j}   \bra{j}\ket{\psi(t)} &= \frac{\eta}{(c_1+c_2)G} \sum_{x=0}^{c_1 G}\sum_{k=0}^{G'-1} \exp\Bigg[-\frac{2}{(\hat\sigma c_1)^2}\frac{(x-c_1G/2)^2}{G^2} \nonumber\\
    &\qquad + i \left(\frac{2\hat p_0}{c_1}\frac{x}{G} - \hat t \left(\frac{\pi k}{1+c_2/c_1}\right)^2 + \frac{2\pi k}{c_1+c_2}\frac{j-x}{G} + \cO(G^{-2})\right)\Bigg]\;.
\end{align}
Here, we used $G' = (c_1+c_2)G$ and $x_0 = c_1G/2$.
Equivalently, we can offset the sum to be
\begin{align}
   \bra{h_j} \bra{j}\ket{\psi(t)} &= \frac{1}{c_1+c_2}\frac{\eta}{G} \sum_{x=-c_1G/2}^{c_1 G/2-1}\sum_{k=0}^{G'-1} \exp\Bigg[-\frac{2}{(\hat\sigma c_1)^2}\frac{x^2}{G^2} + \cO(G^{-2})\Bigg] \nonumber\\
    &\qquad \times\exp\Bigg[i \left(\hat p_0\left(\frac{2}{c_1}\frac{x}{G}+1\right) - \hat t \left(\frac{\pi k}{1+c_2/c_1}\right)^2 + \frac{2\pi k}{c_1+c_2}\left(\frac{j-x}{G} - \frac{c_1}{2}\right)\right)\Bigg]\\
    &= \frac{1}{c_1+c_2}\frac{\eta}{G} \sum_{k=0}^{G'-1} \exp\Bigg[- \frac{\hat \sigma^2}{2} \left(\hat p_0 - \alpha k\right)^2 + i \left(\hat p_0 - (\alpha k)^2 \hat t + \alpha k\left(\frac{2}{c_1}\frac{j}{G} - 1\right)\right)\Bigg] \nonumber\\
    &\qquad \times \sum_{x=-c_1G/2}^{c_1 G/2-1}\exp\Bigg[-\frac{1}{2}\left(\frac{2}{\hat \sigma c_1 G}x - i \hat \sigma \left(\hat p_0 - \alpha k\right)\right)^2 + \cO(G^{-2})\Bigg]\;,
\end{align}
where we defined for convenience a constant
\begin{align}
    \alpha := \frac{\pi}{1+c_2/c_1}\;.
\end{align}
For fixed $k$, the last sum is of the form
\begin{align}
    \sum_{x=-S/2}^{S/2-1} \exp\left[-\frac{1}{2}\left(\frac{x}{S\hat\sigma/2} - i\beta\right)^2\right]
\end{align}
for $\beta = \hat\sigma(\hat p_0 - \alpha k)$ and $S = c_1 G$. We will expand its range at the cost of additional error $\epsilon_\mathrm{trunc}$, i.e.,
\begin{align}
    \sum_{x=-S/2}^{S/2-1} \exp\left[-\frac{1}{2}\left(\frac{x}{S\hat\sigma/2} - i\beta\right)^2\right] &= \sum_{x=-\infty}^\infty \exp\left[-\frac{1}{2}\left(\frac{x}{S\hat\sigma/2} - i\beta\right)^2\right] + 2\epsilon_\mathrm{trunc}\;.
\end{align}
For small $\hat\sigma$, the truncation error is bounded by
\begin{align}
    \epsilon_\mathrm{trunc} &= \left|\sum_{x=S/2}^\infty \exp\left[-\frac{1}{2}\left(\frac{x}{S\hat\sigma/2} - i\beta\right)^2\right]\right| \\
    &\leq \sum_{x=S/2}^\infty \exp\left[-\frac{1}{2}\left(\frac{x^2}{(S\hat\sigma/2)^2} - \beta^2\right)\right] \\
    &\leq \exp\left[\frac{\beta^2}{2}\right] \int_{S/2-1}^\infty \exp\left[-\frac{1}{2}\frac{x^2}{(S\hat\sigma/2)^2}\right] dx\\
    &\leq \frac{\sqrt{2\pi}}{4} S \hat \sigma \, \mathrm{erfc}\left(\frac{S-2}{\sqrt{2} S \hat \sigma}\right) \exp\left[\frac{\beta^2}{2}\right]\\
    &= \cO\left(S\hat\sigma^2 e^{-1/8\hat\sigma^2}\right)\;,
\end{align}
where after evaluating the integral we used the fact that $\beta$ is proportional to $\hat\sigma$. Applying Poisson summation, we have
\begin{align}
    \sum_{x=-\infty}^\infty \exp\left[-\frac{1}{2}\left(\frac{x}{S\hat\sigma/2} - i\beta\right)^2\right] &= \sqrt{\frac{\pi}{2}}S \hat\sigma e^{-\beta^2/2} \sum_{\omega=-\infty}^\infty \exp\left[-\frac{1}{2}\left(\pi\omega S \hat\sigma - \beta\right)^2\right]\\
    &= \sqrt{\frac{\pi}{2}} c_1 G \hat\sigma  + \cO(\epsilon_{|\omega| > 0})\;,
\end{align}
where
\begin{align}
    \epsilon_{|\omega|>0} &\leq \sqrt{\frac{\pi}{2}} \hat\sigma S e^{-\beta^2/2}\int_{-\infty}^\infty \exp\left[-\frac{1}{2}\left(\pi\omega\hat\sigma S - \beta\right)^2\right]d\omega\\
    &\leq e^{-\beta^2/2}\\
    &\leq 1\;.
\end{align}
Note that above, we took $\beta$ to be constant compared to $G$, despite its dependence on $k$. For the purposes of our proof, this holds because the state $\ket{\psi_\Delta}$ has arbitrarily large overlap with $\ket{\psi(0)}$, and $\ket{\psi_\Delta}$ is confined to a low-energy subspace where $k$ is at most constant-sized with respect to $G$.
The time-evolved state thus satisfies
\begin{align}
   \bra{h_j} \bra{j}\ket{\psi(t)} &= \sqrt{\frac{\pi}{2}}\frac{\hat\sigma \eta}{1+c_2/c_1}\left(1 + \cO\left( \hat\sigma e^{-1/8\hat\sigma^2}\right)\right) \nonumber\\
    &\quad \times\sum_{k=0}^{G'-1} \exp\Bigg[- \frac{\hat \sigma^2}{2} \left(\hat p_0 - \alpha k\right)^2 + i \left(\hat p_0 - (\alpha k)^2 \hat t + \alpha k\left(\frac{2}{c_1}\frac{j}{G} - 1\right) + \cO(G^{-2})\right)\Bigg]\\
    &= \sqrt{\frac{1}{2\pi}} \alpha \hat\sigma \eta\left(1 + \cO\left( \hat\sigma e^{-1/8\hat\sigma^2}\right)\right) \exp\left[\frac{1}{2}\left(\hat p_0(2i - \hat p_0 \hat \sigma^2) - \left(\frac{\alpha(\hat j - i \hat p_0 \hat \sigma^2)}{\gamma}\right)^2\right)\right] \nonumber\\
    &\quad \times \sum_{k=-G'/2}^{G'/2-1} \exp\left[-\frac{\gamma^2}{2}\left(k - \left(\frac{\alpha(\hat p_0 \hat\sigma^2 + i \hat j)}{\gamma^2} -\frac{G'}{2}\right)\right)^2 + \cO(G^{-2})\right]
\end{align}
for
\begin{align}
    \gamma = \alpha\left(\hat \sigma^2 + 2i\hat t\right)^{1/2}, \quad \hat j &= \frac{2}{c_1}\frac{j}{G} - 1\;.
\end{align}
To evaluate the last sum, we introduce constant-size variables
\begin{align}
    \gamma_1 = \hat\sigma^2 \alpha^2, \quad \gamma_2 = 2 \hat t \alpha^2
\end{align}
as well as
\begin{align}
    \mu_1 = \frac{\alpha^3 (\hat p_0 \hat \sigma^4 + 2 \hat j \hat t)}{\gamma_1^2 + \gamma_2^2}-\frac{G'}{2}, \quad \mu_2 = \frac{\alpha^3 \hat \sigma^2(\hat j - 2\hat p_0 \hat t)}{\gamma_1^2 + \gamma_2^2}, \quad \mu = \mu_1 - \frac{\gamma_2}{\gamma_1}\mu_2\;.
\end{align}
The last sum can now be rewritten as
\begin{align}
    \sum_{k=-G'/2}^{G'/2-1} \exp\left[-\frac{\gamma_1 + i\gamma_2}{2}\left(k - (\mu_1+i\mu_2)\right)^2\right]\;.
\end{align}
Similarly to before, we expand the limits of the sum at the cost of error
\begin{align}
    \epsilon_\mathrm{trunc}' &= \left|\sum_{x=G'/2}^\infty \exp\left[-\frac{\gamma_1 + i\gamma_2}{2}\left(k - (\mu_1+i\mu_2)\right)^2\right]\right|\\
    &\leq \exp\left[\frac{\mu_2^2(\gamma_1^2 + \gamma_2^2)}{2\gamma_1}\right]\int_{G'/2-1}^\infty \exp\left[-\frac{\gamma_1}{2}\left(k - \mu\right)^2\right] dk\\
    &\leq \sqrt{\frac{\pi}{2\gamma_1}} \exp\left[\frac{\mu_2^2(\gamma_1^2 + \gamma_2^2)}{2\gamma_1}\right] \mathrm{erfc}\left(\sqrt{\frac{\gamma_1}{2}}\left(\frac{G'}{2} - (\mu+1)\right)\right)\\
    &= \cO\left(\frac{1}{G}e^{-(\alpha \hat\sigma G)^2/2}\right)\;.
\end{align}
Up to additional error $\epsilon_\mathrm{trunc}'$, our sum is now approximated by
\begin{align}
    \sum_{k=-\infty}^\infty \exp\left[-\frac{\gamma_1 + i\gamma_2}{2}\left(k - (\mu_1+i\mu_2)\right)^2\right] &= \vartheta_3(0, e^{-(\gamma_1+i\gamma_2)/2})
\end{align}
for Jacobi elliptic theta function $\vartheta_3$. Note that this is independent of the site $j$. Moreover, since the error from $\epsilon_\mathrm{trunc}'$ is dominated by the error from $\epsilon_\mathrm{trunc}$, we have
\begin{align}
	|\bra{h_j}\bra{j}\ket{\psi(t)}|^2 &\propto \left|\exp\left[\frac{1}{2}\left(\hat p_0(2i - \hat p_0 \hat \sigma^2) - \left(\frac{\alpha(\hat j - i \hat p_0 \hat \sigma^2)}{\gamma}\right)^2\right) + \cO(G^{-2})\right]\right|^2 \nonumber\\
	&\qquad \times\left(1 + \cO\left(\hat \sigma e^{-1/8\hat \sigma^2}\right)\right)\\
	&\propto \exp\left[-\frac{1}{2}\frac{\hat \sigma^2}{\hat \sigma^4 + 4\hat t^2} \left(\frac{2}{c_1}\frac{j}{G} - 1 - 2\hat p_0 \hat t\right)^2 + \cO(G^{-2})\right]\left(1 + \cO\left(\hat \sigma e^{-1/8\hat \sigma^2}\right)\right)
\end{align}
for a constant of proportionality independent of $j$.
We wish to show that the wavepacket is localized on the identity gates at the end of the circuit. To achieve this behavior, we must appropriately choose the identity padding, wavepacket width, wavepacket momentum, and evolution time. One such choice is as follows, where $R$ will ultimately be a large constant:
\begin{align}
    c_2 = R^2, \quad c_1 = R, \quad \hat\sigma^2 = 1/R, \quad \hat p_0 = R^2, \quad \hat t = 1/R\;,
\end{align}
which implies
\begin{align}
    |\bra{h_j}\bra{j}\ket{\psi(t)}|^2 &\propto \exp\left[-\frac{R}{10} \left(\frac{2}{R}\frac{j}{G} - 1 - 2R)^2 + \cO(G^{-2})\right)\right]\left(1 + e^{-\cO(R)}\right)\;.
\end{align}
To show that the wavepacket is localized on the identity padding at the end of the Hamiltonian, we show that the support on the first $(c_1+1)G$ sites (i.e., before $\cU$ is computed) is vanishingly small compared to the support on the final identity padding (i.e., after $\cU$ is computed). The LHS below defines this ratio $\zeta$ after substituting in the $R$ parameter. For our choice of time and momentum, it suffices to only evaluate the support on the last $G$ sites of the Hamiltonian, i.e., we can upper-bound the ratio $\zeta$ by the quantity
\begin{align}
	\zeta = \frac{\sum_{j=0}^{(R+1)G} |\bra{h_j}\bra{j}\ket{\psi(t)}|^2}{\sum_{j'=(R+1)G}^{(R^2+R+1)G} |\bra{h_{j'}}\bra{j'}\ket{\psi(t)}|^2} \leq \frac{\sum_{j=0}^{(R+1)G} |\bra{h_j}\bra{j}\ket{\psi(t)}|^2}{\sum_{j'=(R^2+R)G}^{(R^2+R+1)G} |\bra{h_{j'}}\bra{j'}\ket{\psi(t)}|^2}\;.
\end{align}
We then bound the numerator and denominator by the maximal and minimal terms in each sum respectively to obtain
\begin{align}
	\zeta &\leq \left(1 + e^{-\cO(R)}\right)\frac{((R+1)G+1)\exp\left[-\frac{R}{10} \left(\frac{2(R+1)}{R} - 1 - 2R)^2 + \cO(G^{-2}\right)\right]}{G\exp\left[-\frac{R}{10} \left(\frac{2(R^2+R)}{R} - 1 - 2R)^2 + \cO(G^{-2}\right)\right]}\\
	&= \exp\left[-\Omega(R^3 \log R)\right]\;.
\end{align}
Hence, we can always choose a sufficiently large constant $R$ such that the support on the sites where the computation is unfinished becomes arbitrarily small, compared to the support on the sites where the computation is finished. This completes the proof.
\end{proof}

\end{document}